\documentclass[11 pt]{article}

\usepackage{fullpage}
\usepackage{xspace}
\usepackage{graphicx}
\usepackage{enumitem}
\usepackage{subcaption}

\usepackage{amsmath,amsthm,amssymb}
\newtheorem{theorem}{Theorem}[section]
\newtheorem{lemma}[theorem]{Lemma}
\newtheorem{corollary}[theorem]{Corollary}
\newtheorem{claim}{Claim}
\newtheorem{proposition}[theorem]{Proposition}

\newtheorem{fact}[theorem]{Fact}

\newcommand{\cA}{{\cal{A}}}
\newcommand{\cB}{{\cal{B}}}
\newcommand{\cG}{{\cal{G}}}
\newcommand{\cU}{{\cal{U}}}
\newcommand{\cV}{{\cal{V}}}
\newcommand{\select}{\ensuremath{\mathit{S}\xspace}}
\newcommand{\pe}{\ensuremath{\mathit{PE}\xspace}}
\newcommand{\ppe}{\ensuremath{\mathit{PPE}\xspace}}
\newcommand{\cppe}{\ensuremath{\mathit{CPPE}\xspace}}

\newcommand\doubleplus{\ensuremath{\mathbin{+\mkern-10mu+}}}

\begin{document}
\def\thefootnote{\fnsymbol{footnote}}

\title{Four Shades of Deterministic Leader Election in Anonymous Networks} 

	\author{
	Barun Gorain\footnotemark[1]
	\and Avery Miller\footnotemark[2]
	\and Andrzej Pelc\footnotemark[3]
}

\footnotetext[1]{Department of Electrical Engineering and Computer Science, Indian Institute of Technology Bhilai, India. {\tt barun@iitbhilai.ac.in}}
\footnotetext[2]{Department of Computer Science, University of Manitoba, Winnipeg, Manitoba, R3T 2N2, Canada. {\tt avery.miller@umanitoba.ca}. Supported by NSERC Discovery Grant RGPIN--2017--05936.}
\footnotetext[3]{D\'epartement d'informatique, Universit\'e du Qu\'ebec en Outaouais, Gatineau,
	Qu\'ebec J8X 3X7, Canada. {\tt pelc@uqo.ca}. Partially supported by NSERC Discovery Grant RGPIN--2018--03899
	and by the Research Chair in Distributed Computing at the
	Universit\'e du Qu\'ebec en Outaouais.}

\maketitle

\thispagestyle{empty}

\begin{abstract}
Leader election is one of the fundamental problems in distributed computing: a single node, called the leader, must be specified. This task can be formulated either in a weak way, where one node outputs {\em leader} and all other nodes output {\em non-leader}, or in a strong way, where all nodes must also learn which node is the leader. If the nodes of the network have distinct identifiers, then such an agreement means that all nodes have to output the identifier of the elected leader. For anonymous networks, the strong version of leader election requires that all nodes must be able to find a path to the leader, as this is the only way to identify it. In this paper, we study variants of deterministic leader election in arbitrary anonymous networks. 

Leader election is impossible in some anonymous networks, regardless of the allocated amount of time, even if nodes know the entire map of the network. This is due to possible symmetries in the network. However, even in networks in which it is possible to elect a leader knowing the map, the task may be still impossible without any initial knowledge, regardless of the allocated time. On the other hand, for any network in which leader election (weak or strong) is possible knowing the map, there is a minimum time, called the {\em election index}, in which this can be done. We consider four formulations of leader election discussed in the literature in the context of anonymous networks : one is the weak formulation, and the three others specify three different ways of finding the path to the leader in the strong formulation. Our aim is to compare the amount of initial information needed to accomplish each of these ``four shades'' of leader election in minimum time. Following the framework of {\em algorithms with advice}, this information (a single binary string) is provided to all nodes at the start by an oracle knowing the entire network. The length of this string is called the {\em size of advice}.  

We show that the size of advice required to accomplish leader election in the weak formulation in minimum time is exponentially smaller than that needed for any of the strong formulations. Thus, if the required amount of advice is used as a measure of the difficulty of the task,  the weakest version of leader election in minimum time is drastically easier than any version of the strong formulation in minimum time.
\end{abstract}

\newpage
\section{Introduction}

\subparagraph*{Background.} Leader election is one of the fundamental problems in distributed computing: a single node, called the leader, must be specified. This task
was first formulated in \cite{LL} in the study of local area token ring networks, where, at all times, exactly one node (the owner of a circulating token) is allowed to initiate
communication. When the token is accidentally lost, a leader must be elected as the initial owner of the token.

The task of leader election can be formulated either in a weak way, where one node outputs {\em leader} and all other nodes output {\em non-leader}, or in a strong way, where all nodes must also learn which node is the leader.
If the nodes of the network have distinct identifiers, then such an agreement means that all nodes have to output the identifier of the elected leader. In the labeled case, the weak and the strong version do not differ much: once a node knows that it is a leader, it can simply broadcast its identifier to all other nodes.
By contrast, for anonymous networks,
in the strong version of leader election all nodes must be able to find a path to the leader, as this is the only way to identify it. This turns out to be much more difficult.

In this paper, we study variants of deterministic leader election in arbitrary anonymous networks. 
In many
applications, even if nodes have distinct identities, they may decide to refrain from revealing them, e.g., for privacy or security reasons. Hence it is important to design leader election algorithms that do not rely on knowing distinct labels of nodes, and that can work in anonymous networks as well. This was done, e.g.,  in \cite{BSVCGS,DiPe,GMP,YK3}.

\subparagraph*{Model and Problem Description.} The network is modeled as a simple undirected  connected $n$-node graph with maximum degree $\Delta$.
Nodes do not have any identifiers.
On the other hand, we assume that, at each node $v$,
each edge incident to $v$ has a distinct {\em port number} from 
$\{0,\dots,d-1\}$, where $d$ is the degree of $v$. Hence, each edge has two corresponding port numbers, one at each of its endpoints. 
Port numbering is {\em local} to each node, i.e., there is no relation between
port numbers at  the two endpoints of an edge. Initially, each node has no knowledge of the network, apart from its own degree.

We use the extensively studied $\cal{LOCAL}$ communication model \cite{Pe}. In this model, communication proceeds in synchronous
rounds and all nodes start simultaneously. In each round, each node
can exchange arbitrary messages with all of its neighbors and perform arbitrary local computations.
It is well known that the synchronous process of the $\cal{LOCAL}$  model can be simulated in an asynchronous network using time-stamps. 

We now formulate precisely four versions of the leader election task, in increasing order of their strength.
The weakest version of all  is called {\bf Selection} and will be abbreviated by $\select$: one node of the network must output {\em leader} and all other nodes must output {\em non-leader}.
This is the basic version defined, e.g., in \cite{Ly}.
All other versions of leader election in anonymous networks require one node to output {\em leader} and require enabling all other nodes to find a path to the leader.
Arguably, the weakest way to do it is the following: each node outputs the first port number on a simple path from it to the leader. 
This will be called {\bf Port Election} and will be abbreviated by $\pe$.
Then every node can, for example, send a message to the leader that will be conveyed by using the port numbers on the resulting simple path.  This simple and natural way was never analyzed in detail in the context of leader election in anonymous networks, but was mentioned in
\cite{DiPe,GMP} as an alternative possibility. A stronger way is for each non-leader to output the {\em entire path} to the leader. This, in turn, can be done in two ways. In \cite{GMP},  a simple path from a node $v=v_0$ to the leader was coded as a sequence $(p_1,\dots, p_k)$ of port numbers, such that node $v_{i+1}$ is reached from $v_i$ by taking port $p_{i+1}$ and the path $(v_0,v_1,\dots ,v_k)$ is a simple path, where $v_k$ is the leader. This version will be called {\bf Port Path Election} and will be abbreviated by $\ppe$. Finally, in \cite{DiPe}, a simple path from a node to the leader was coded by listing all port numbers on the path, in their order of appearance. More precisely, every node $v$ must output a sequence $P(v)=(p_1,q_1,\dots,p_k,q_k)$ of nonnegative integers.
For each node $v$, let $P^*(v)$ be the simple path starting at $v$, such that port numbers $p_i$ and $q_i$ correspond to the $i$-th edge of $P^*(v)$, in the order from $v$ to the other end of this path.
All paths $P^*(v)$ must end at a common node, called the leader. This version will be called {\bf Complete Port Path Election} and will be abbreviated by $\cppe$. Notice that in the absence of port numbers, there would be no way to identify the elected leader by non-leaders, as all
ports, and hence all neighbors, would be indistinguishable to a node. Thus, the tasks $\pe$, $\ppe$ and $\cppe$ would be impossible to formulate.

In \cite{GMP}, there is a discussion comparing the above versions of leader election. The authors mention that Selection is sufficient for some tasks, e.g.,  if the leader has 
to broadcast a message to all other nodes, but insufficient for others, e.g., if all nodes have to send a message to the leader. For the latter task, Port Election is enough, as packets could be routed to the leader from node to node using only the local port that each node outputs. However, the authors argue that this holds only
if nodes want to cooperate with others by revealing the local port towards the leader when retransmitting packets. They observe that, in some applications, such cooperation may be uncertain, and even when it occurs, it may slow down transmission as the local port has to be retrieved from the memory of the relaying node. Putting the entire path to the leader as a header of the packet by the original sender
(i.e., using version $\ppe$ or $\cppe$) may in some cases speed up transmissions, because relaying may then be done at the router level.

The central notion in the study of anonymous networks is that of the view of a node \cite{YK3}. 
Let $G$ be any graph and $v$ a node in this graph. The {\em view} from $v$ in $G$, denoted ${\cal V}(v)$, is the infinite tree of all finite paths in $G$, starting from node $v$ and coded as sequences
$(p_1,q_1,\dots,p_k,q_k)$  of port numbers, where $p_i, q_i$ are the port numbers corresponding to the $i$-th edge of the path, in the order starting at the root $v$, with the rooted tree structure defined by the prefix relation of sequences.  The {\em truncated view} $\cV^l(v)$ is the truncation of ${\cal V}(v)$  to level $l$, for each $l$.

%

The information that $v$ gets about the graph in $r$ rounds
is precisely the truncated view $\cV^r(v)$, together with degrees of leaves of this tree. Denote by $\cB^r(v)$ the truncated view $\cV^r(v)$ whose leaves are labeled by their degrees in the graph, and call it the {\em augmented truncated view} at depth $r$.
If no additional knowledge is provided {\em a priori} to the nodes, the decisions of a node $v$ in round $r$ in any deterministic algorithm are a function of $\cB^r(v)$.
The {\em time} of (any version of) leader election in a given graph is the minimum number of rounds sufficient to complete it by all nodes of this graph.

Unlike in labeled networks, if the network is anonymous then leader election is sometimes impossible, regardless of the allocated time, even if the topology of the network is known to nodes and even in the weakest version, i.e., Selection.
This is due to symmetries, and the simplest example is the two-node graph. It follows from \cite{YK3} that if nodes know the map of the graph
(i.e., its isomorphic copy with all port numbers indicated)
then leader election is possible if and only if views of all nodes are distinct. This holds for all four versions of the leader election task discussed above. We will call such networks {\em feasible} and restrict attention to them. However, even in the class of 
feasible networks, even Selection is impossible without any {\em a priori} knowledge about the network.  On the other hand, for any fixed feasible network $G$, whose map is given to the nodes,
and for any version $Z \in \{\select,\pe,\ppe,\cppe\}$ of leader election, there is a minimum time, called
the $Z$-{\em index} of $G$ and denoted by $\psi_Z(G)$, in which version $Z$ of leader election can be completed on $G$. For example, $\psi_S(G)=0$ if and only if $G$ contains a node whose degree is unique. On the other hand, if $G$ is the 3-node line with ports $0,0,1,0$ from left to right, then $\psi_{\cppe}(G)=1$.

We observe that the election tasks defined above form a hierarchy with respect to their election indexes. In particular, note that if $\cppe$ can be solved in $G$ in $k$ rounds, then the non-leaders can simply output the outgoing ports of their output sequence in order to solve $\ppe$ at the end of $k$ rounds. Further, if $\ppe$ can be solved in $k$ rounds, then the non-leaders can output the first outgoing port of their output sequence in order to solve $\pe$ at the end of $k$ rounds. Finally, if $\pe$ can be solved in $k$ rounds, then the non-leaders can simply output `non-leader' in order to solve $\select$ at the end of $k$ rounds. Hence we have the following fact.

\begin{fact}\label{indexHierarchy}
	$\psi_{\cppe}(G) \geq \psi_{\ppe}(G) \geq \psi_{\pe}(G) \geq \psi_{\select}(G)$ for any graph $G$.
\end{fact}

Our aim is to compare the amount of information needed to accomplish each of these ``four shades'' of leader election in minimum time. In order to avoid ``comparing apples to oranges'', we should make comparisons between any versions $A$ and $B$ of leader election for graphs in which the minimum time to accomplish version
$A$ is equal to the minimum time to accomplish version $B$. In other words, in order to prove, e.g.,  that the amount of information needed to accomplish version $\pe$ in minimum time is much larger than that required to accomplish version $\select$ in minimum time, we have to show that for all graphs $G$ version $S$ in time 
$\psi_{\select}(G)$ can be accomplished using a small amount of information, but there is a class $\cal G$ of graphs $G$ for which $\psi_{\select}(G)=\psi_{\pe}(G)$, and for which accomplishing version $\pe$ in time $\psi_{\pe}(G)$ requires a much larger amount of initial information. 

Following the framework of {\em algorithms
	with advice}, see, e.g.,  \cite{DP,EFKR,FGIP,FKL,FP,IKP,SN}, information (a unique binary string) is provided to all nodes at the start by an oracle knowing the entire network. The length of this string is called the {\em size of advice}.  It should be noted that, since the advice given to all nodes is the same, this information does not increase the asymmetries 
of the network (unlike in the case when different pieces of information could be given to different nodes) but only helps to take advantage of the existing asymmetries and use them to elect the leader.

The paradigm of algorithms with advice has been proven very important in the domain of network algorithms. Establishing a strong lower bound on the minimum size of advice sufficient to accomplish a given task implies that entire classes of algorithms can be ruled out. For example, one of our results shows that, for some class of graphs, Selection in minimum time requires advice of size polynomial in the maximum degree of the graph. This permits to eliminate potential Selection algorithms relying  only on knowing the maximum degree of the network, as this information is a piece of advice of size logarithmic in this maximum degree. 
Lower bounds on the size of advice
give us impossibility results based strictly on the \emph{amount} of initial knowledge available to nodes. Hence this is a quantitative approach.
This is much more general than the traditional approach that could be called ``qualitative'', based on
specific {categories} of information given to nodes, such as the size, diameter, or maximum node degree.

\subparagraph*{Our results.} 
We show that the size of advice needed to accomplish leader election in the weakest formulation, i.e., Selection, in minimum time, is exponentially smaller than that needed for any of the strong formulations, i.e., Port Election, Port Path Election, or Complete Port Path Election. More precisely, we show that this minimum size of advice for Selection is polynomial in the maximum degree $\Delta$ for all graphs, but, for each $Z \in\{\pe,\ppe,\cppe\}$ and for sufficiently large $\Delta$ and $k$, there exists a class ${\cal C} (Z)$ of graphs of maximum degree $\Delta$ such that  $\psi_S(G)=\psi_{Z}(G)=k$ for all graphs $G$ in ${\cal C} (Z)$, and the size of advice required to accomplish the task $Z$ in minimum time, for some graph of class ${\cal C} (Z)$, is exponential in $\Delta$.

It should be stressed that, while accomplishing $\cppe$ obviously implies accomplishing $\ppe$ which in turn implies accomplishing $\pe$, the above three separations between $\select$ and any $Z$ in $\{\pe,\ppe,\cppe\}$ must be proved separately. For example, the fact that there exists a  class ${\cal C} (\pe)$ of graphs of maximum degree $\Delta$ such that  $\psi_S(G)=\psi_{\pe}(G)=k$ for all graphs $G$ in ${\cal C} (\pe)$, and the size of advice required to accomplish $\pe$ in minimum time, for some graph of class ${\cal C} (\pe)$, is exponential in $\Delta$, does not necessarily imply the same statement when $\pe$ is replaced by $\ppe$, because for graphs $G$ in the class ${\cal C} (\pe)$, we could have $\psi_{\ppe}(G)$ much larger than $\psi_{\select}(G)$. Indeed, the class of graphs that we construct to prove a lower bound on advice size for the $\pe$ task is different than the class of graphs that we use for the $\ppe$ and $\cppe$ tasks. Actually, the construction for the tasks PPE and CPPE is more difficult than for the task PE because it seems more difficult to reconcile small election index with the need of large advice in the case of PPE and CPPE than in the case of PE.


From the technical standpoint, our main contributions are constructions showing lower bounds on the size of advice needed to accomplish various versions of leader election in minimum time.
These lower bounds are a crucial tool to show separations of difficulty between the  weakest version of leader election (i.e., Selection) and the three strong versions.

\subparagraph*{Related work.}
Early papers on leader election focused on the scenario 
with distinct labels. Initially, it was investigated for rings in the message passing model.
A synchronous algorithm based on label comparisons was given in \cite{HS},  using
$O(n \log n)$ messages.  In \cite{FL}, the authors proved that
this complexity is optimal for comparison-based algorithms, while they showed
a leader election algorithm using only a linear number of messages but running in very large time.
An asynchronous algorithm using $O(n \log n)$ messages was given, e.g., in \cite{P}, and
the optimality of this message complexity was shown in \cite{B}. Leader election was also investigated for radio networks,
both in the deterministic \cite{JKZ,KP,NO} and in the randomized \cite{Wil} scenarios.
In \cite{HKMMJ}, leader election for labeled networks was
studied using mobile agents.

Many authors \cite{An,AtSn,ASW,BSVCGS,BV,FP1,YK2,YK3} studied leader election
in anonymous networks. In particular, \cite{BSVCGS,YK3} characterize message-passing networks in which
leader election is feasible. In \cite{YK2}, the authors study
leader election in general networks, under the assumption that node labels exist but are
not unique. 
In  \cite{DoPe,FKKLS},  the authors
study message complexity of leader election in rings with possibly
nonunique labels. 
Memory needed for leader election in unlabeled networks was studied in \cite{FP}. 
In \cite{DP1}, the authors investigated the feasibility of leader election among anonymous agents that
navigate in a network in an asynchronous way.

Providing nodes or agents with arbitrary types of knowledge that can be used to increase efficiency of solutions to network problems 
was previously 
proposed in \cite{AKM01,DP,EFKR,FGIP,FIP1,FIP2,FKL,FP,FPR,GPPR02,IKP,KKKP02,KKP05,MP,SN,TZ05}. This approach was referred to as
{\em algorithms with advice}.  
The advice is given either to the nodes of the network or to mobile agents performing some task in a network.
In the first case, instead of advice, the term {\em informative labeling schemes} is sometimes used if (unlike in our scenario) different nodes can get different information.

Several authors studied the minimum size of advice required to solve
network problems in an efficient way. 
In \cite{FIP1}, the authors compared the minimum size of advice required to
solve two information dissemination problems using a linear number of messages. 
In \cite{FKL}, it was shown that advice of constant size given to the nodes enables the distributed construction of a minimum
spanning tree in logarithmic time. 
In \cite{EFKR}, the advice paradigm was used for online problems.
In the case of \cite{SN}, the issue was not efficiency but feasibility: it
was shown that $\Theta(n\log n)$ is the minimum size of advice
required to perform monotone connected graph clearing.
In \cite{FPR}, the authors studied the problem of topology recognition with advice given to the nodes.

Among papers studying the impact of information on the time of leader election, the papers \cite{DiPe,GMP,MP} are closest to the present work.
In \cite{MP}, the authors investigated the minimum size of advice sufficient to find the largest-labelled node in a graph, all of whose nodes have distinct labels.
They compared the task of selection with that of election requiring all nodes to know the identity of the leader.
The main difference between  \cite{MP} and the present paper is that we consider networks without node labels. This is a fundamental difference:
breaking symmetry in anonymous networks relies heavily on the structure of the graph, rather than on labels, and, as far as results
are concerned, much more advice is needed for a given allocated time.

The authors of  \cite{GMP} studied leader election under the advice paradigm for anonymous networks, but they restricted attention to trees. They studied the version that we call $\ppe$ and established upper and lower bounds on the size of advice for various allocated time values. On the other hand, authors of \cite{DiPe} investigated leader election in arbitrary anonymous networks. They used the version that we call $\cppe$ and studied the minimum size of advice to accomplish it both in minimum possible time and for much larger time values allocated to leader election. The different versions of leader election studied in
\cite{DiPe,GMP,MP} inspired us to investigate comparisons of their difficulty measured by the required size of advice.

\section{Solving Selection in minimum time}\label{sec1}
In this section, we prove tight upper and lower bounds on the size of the advice needed to solve Selection on any graph $G$ in time $\psi_{\select}(G)$.

\subsection{Upper Bound}
 First, we show that if $\select$ is solvable using $k$ rounds in a graph $G$, then there must be a node in $G$ whose augmented truncated view $\cB^k(v)$ is unique.

\begin{proposition}\label{uniquekview}
	For any graph $G$ and any positive integer $k$, suppose that there exists an algorithm $\mathcal{A}$ that solves $\select$ using $k$ rounds. At the end of every execution, the node $u$ that outputs 1 must satisfy $\cB^{k}(u) \ne \cB^{k}(w)$, for all $w \in V(G) \setminus\{u\}$.
\end{proposition}
\begin{proof}
	To obtain a contradiction, assume that there exists an execution $e$ of $\mathcal{A}$ such that a node $u$ outputs 1 and there exists a node $u'$ such that $\cB^{k}(u) = \cB^{k}(u')$. Since each node $v$'s output is a function that depends only on $\cB^{k}(v)$, it follows that, in execution $e$, node $u'$ also outputs 1. This contradicts the correctness of $\mathcal{A}$ since, to solve $\select$, exactly one node must output 1.
\end{proof}

To solve $\select$ in time $\psi_{\select}(G)$ with advice, we specify an oracle that picks a node $u$ whose augmented truncated view is unique (as guaranteed by Proposition \ref{uniquekview}), and provides as advice to all nodes the augmented truncated view of $u$. Our distributed algorithm consists of each node computing its own augmented truncated view and comparing it to the advice they receive from the oracle. The unique node whose view matches the advice outputs 1, and all other nodes output 0. This gives us the following upper bound on the size of advice sufficient to solve $\select$.

\begin{theorem}\label{S-ub}
	There exists a distributed algorithm that solves $\select$ in every graph $G$ whose election index is finite, uses $\psi_{\select}(G)$ communication rounds, and uses advice of size at most $O((\Delta-1)^{\psi_{\select}(G)} \log \Delta)$, where $\Delta$ is the maximum degree of nodes in $G$.
\end{theorem}
\begin{proof}
	We specify an oracle and algorithm pair that solves $\select$ in every graph $G$ whose election index is finite. 
	
	By Proposition \ref{uniquekview}, we know that there exists at least one node $u$ whose augmented truncated view $\cB^{\psi_{\select}(G)}(u)$ is unique. Among all such nodes, the oracle chooses the node $u$ whose $\cB^{\psi_{\select}(G)}(u)$ is lexicographically smallest. The oracle encodes $\cB^{\psi_{\select}(G)}(u)$ as a binary string $A$ using at most $O((\Delta-1)^{\psi_{\select}(G)} \log \Delta)$ bits (this is possible since there are most $\Delta\cdot(\Delta-1)^{(\psi_{\select}(G)-1)}$ edges in this view, and each edge's two port numbers can be encoded using $O(\log \Delta)$ bits). This binary string $A$ is provided as advice to all nodes in the network.
	
	Our distributed algorithm works as follows: each node decodes  the augmented truncated view encoded in the provided advice $A$, and calculates the height $h$ of this view. Then, using $h$ communication rounds, each node $w$ calculates $\cB^{h}(w)$. Finally, each node compares its $\cB^{h}(w)$ with the augmented truncated view encoded in $A$. If these are equal, then the node outputs 1, and outputs 0 otherwise. Correctness is guaranteed by the fact that the augmented truncated view encoded in $A$ by the oracle is equal to $\cB^{h}(u)$ for exactly one node $u$ in the network. The number of communication rounds used is equal to the height $h$ of the augmented truncated view encoded in $A$, i.e., $\psi_{\select}(G)$.
\end{proof}

\subsection{Lower bound}
In this section, we prove a tight lower bound on the size of advice needed to solve $\select$ in minimum time. In particular, for arbitrary positive integers $\Delta \geq 3$ and $k \geq 1$, we construct a class of graphs $\cG_{\Delta,k}$ in which each graph has maximum degree $\Delta$ and has finite $\select$-index $k$ such that every deterministic distributed algorithm solving $\select$ in graphs of this class requires advice of size at least $\Omega((\Delta-1)^k \log \Delta)$. 

\subsubsection{Construction of $\cG_{\Delta,k}$}\label{constructGDeltak}

Consider any positive integers $\Delta \geq 3$ and $k \geq 1$. Our construction involves various building blocks, which we present in an incremental fashion.

{\bf Building Block 1: Rooted Tree $T$.} We define a rooted tree $T$ of height $k$ whose root $r$ has degree $\Delta-2$, and all other internal nodes have degree $\Delta$ (i.e., $\Delta-1$ children and one parent). The ports at the root leading to the root's children are labeled $1,\ldots,\Delta-2$. For each internal node other than the root, the port leading to its parent is labeled 0, and the ports leading to its children are labeled $1,\ldots,\Delta-1$. Let $z$ denote the number of leaves in $T$, and note that $z = (\Delta-2)\cdot(\Delta-1)^{k-1}$. 

{\bf Building Block 2: Augmented Trees.} Using the rooted tree $T$, we construct a large set of trees $\mathcal{T}_{\Delta,k}$ by attaching new nodes to each leaf of $T$. In particular, let $\ell_1,\ldots,\ell_z$ be the leaves of $T$, indexed in increasing order using the lexicographic ordering of the sequence of ports leading from $r$ to each leaf. We construct a tree $T_X$ for each sequence $X = (x_1,\ldots,x_z)$ of $z$ positive integers such that $1 \leq x_i \leq \Delta-1$ by attaching $x_i$ degree-one nodes to $\ell_i$ for each $i \in \{1,\ldots,z\}$. The ports at each $\ell_i$ leading to its new children are labeled $1,\ldots,x_i$. The set $\mathcal{T}_{\Delta,k}$ is defined to be the set of all such trees $T_X$. Note that the number of trees in $\mathcal{T}_{\Delta,k}$ is the number of different sequences $X$ described above, i.e., $|\mathcal{T}_{\Delta,k}| = (\Delta-1)^z$ where $z = (\Delta-2)\cdot(\Delta-1)^{k-1}$.


{\bf Building Block 3: Augmented Trees with Appended Paths} For each tree $T_X \in \mathcal{T}_{\Delta,k}$, create two new trees $T_{X,1}$ and $T_{X,2}$. The tree $T_{X,1}$ is constructed by taking a copy of $T_X$ and creating a new path of length $k+1$ starting at its root $r$. In particular, the new path consists of nodes $r,p_1,\ldots,p_{k+1}$, the ports at $r$ and $p_{k+1}$ on this path are labeled 0, and, for each $i \in \{1,\ldots,k\}$, the port at $p_i$ leading to $p_{i-1}$ is labeled 1, and the port at $p_i$ leading to $p_{i+1}$ is labeled 0. The tree $T_{X,2}$ is similar: take a copy of $T_{X,1}$, but swap the port labels at $p_k$ on the newly-created path so that the port at $p_k$ leading to $p_{k-1}$ (or $r$, if $k=1$) is labeled 0, and the port at $p_k$ leading to $p_{k+1}$ is labeled 1. 
See Figure \ref{TXs} for an illustration of the trees $T_{X,1}$ and $T_{X,2}$.

To make the notation cleaner, we will often index the trees of $\mathcal{T}_{\Delta,k}$ using integers rather than sequences of integers. To enable this, we order the trees of $\mathcal{T}_{\Delta,k}$ as $T_1,\ldots,T_{|\mathcal{T}_{\Delta,k}|}$, in increasing lexicographic order of the integer sequence $X$ used to generate each tree. For each $j \in \{1,\ldots,|\mathcal{T}_{\Delta,k}|\}$, we denote by $r_{j,1}$ the root node of tree $T_{j,1}$, and we denote by $r_{j,2}$ the root node of the tree $T_{j,2}$.

\begin{figure}[h!]
	
	\begin{minipage}{.5\linewidth}
		\centering
		\includegraphics[scale=0.4]{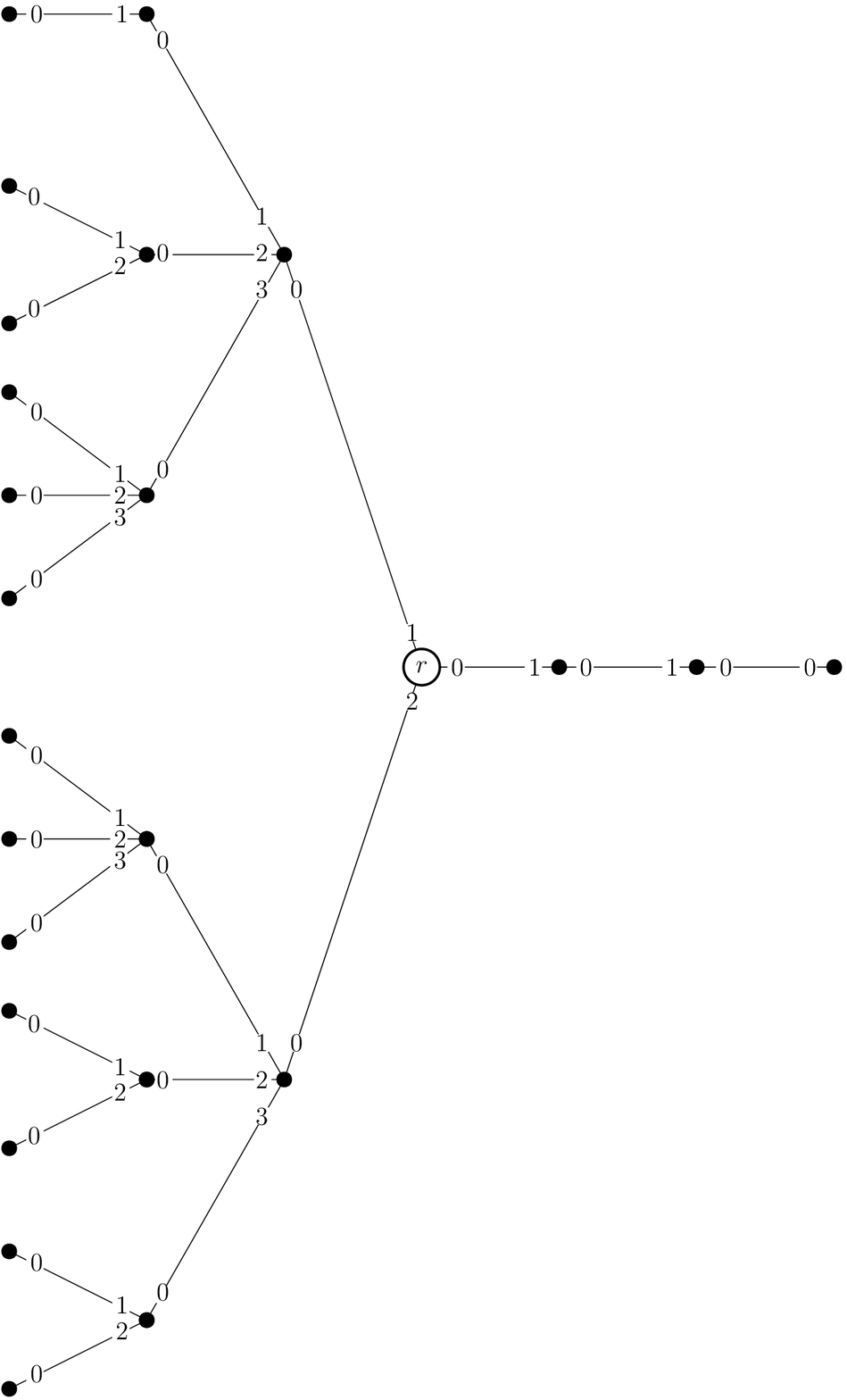} 
	\end{minipage}%
	\begin{minipage}{.5\linewidth}
		\centering
	\includegraphics[scale=0.4]{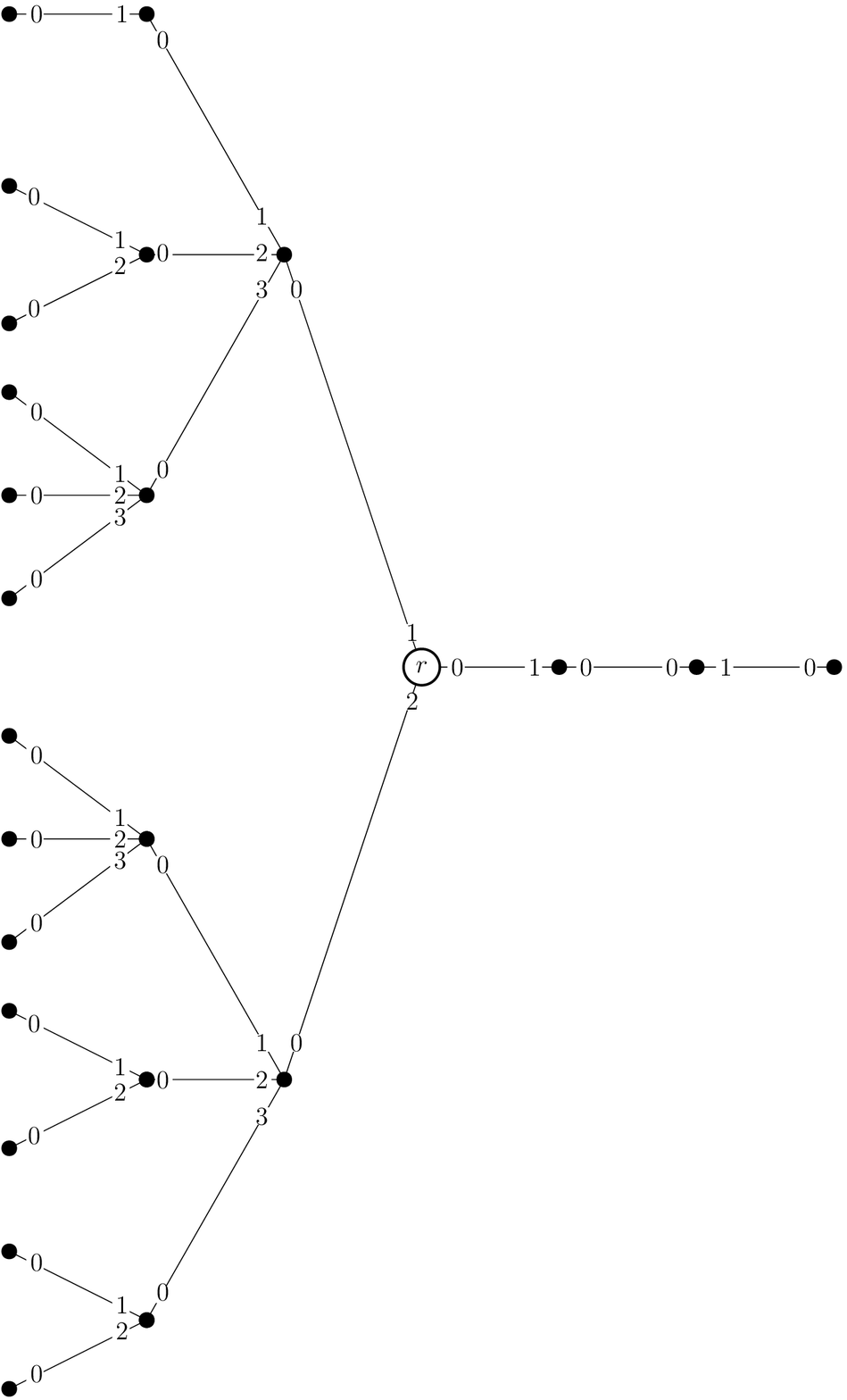} 
	\end{minipage}
	\caption{The tree $T_{X,1}$ (left) and $T_{X,2}$ (right) when $k=2$, $\Delta=4$, and $X = (1,2,3,3,2,2)$}
	\label{TXs}
\end{figure}

{\bf Final Construction of $\cG_{\Delta,k}$.} The class $\cG_{\Delta,k}$ consists of graphs $\{G_1,\ldots,G_{|\mathcal{T}_{\Delta,k}|}\}$, where each $G_i$ is constructed by taking the disjoint union of the following graphs: the tree $T_{i,2}$, two copies of each tree $T_{j',2}$ for $j' \in \{1,\ldots,i-1\}$, two copies of each tree $T_{j,1}$ for $j \in \{1,\ldots,i\}$, and a cycle $C_i$ of $4i-1$ nodes $c_1,\ldots,c_{4i-1}$ with ports alternately labeled 0 and 1. Further, we add the following edges: for each $j \in \{1,\ldots,i\}$, we add an edge between $c_{4j-3}$ and the root node $r_{j,1}$ in the first copy of $T_{j,1}$, an edge between $c_{4j-2}$ and the root node $r_{j,1}$ in the second copy of $T_{j,1}$, an edge between $c_{4j-1}$ and the root node $r_{j,2}$ in the first copy of $T_{j,2}$, and, for each $j' \in \{1,\ldots,i-1\}$, an edge between $c_{4j'}$ and the root node $r_{j',2}$ in the second copy of $T_{j',2}$. For each of these added edges, the port at the cycle node is labeled 2, and the port at the $r$ node is labeled $\Delta-1$. See Figure \ref{Gis} for an illustration of the graph $G_i$. From the description of the construction, we can verify the following calculation of the number of graphs in the class.
\begin{figure}[h]
	\centering
	\includegraphics[scale=0.7]{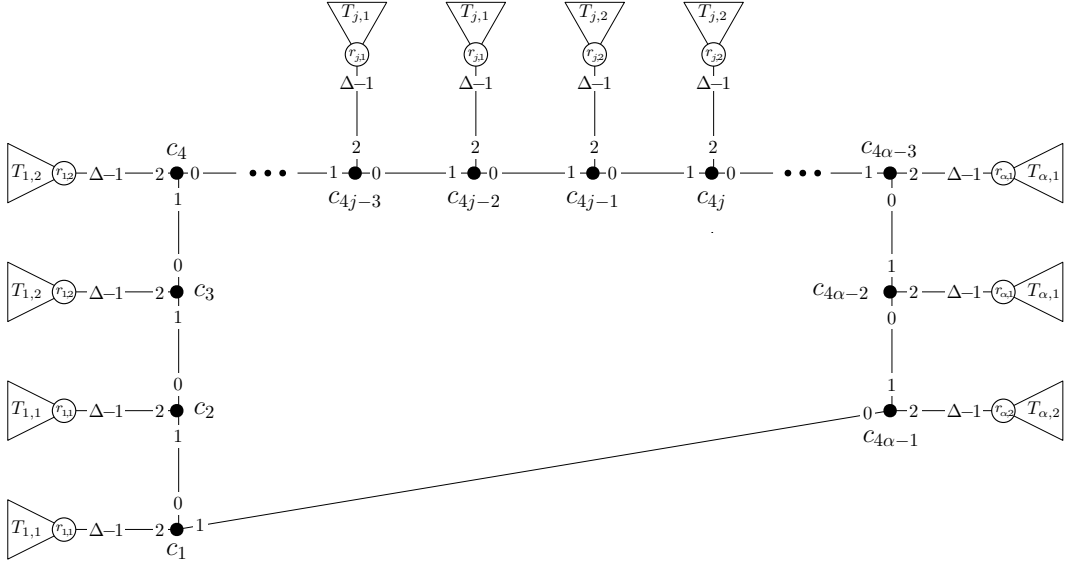}
	\caption{The graph $G_i$}
	\label{Gis}
\end{figure}


\begin{fact}\label{classCountG}
	$|\cG_{\Delta,k}| = |\mathcal{T}_{\Delta,k}| = (\Delta-1)^{(\Delta-2)\cdot(\Delta-1)^{k-1}}$ for any positive integers $\Delta \geq 3$ and $k \geq 1$.
\end{fact}





\subsubsection{Lower bound proof}

The idea behind the lower bound is to prove that, in each graph $G_i \in \cG_{\Delta,k}$, the root node $r_{i,2}$ of $T_{i,2}$ is the only node that has a unique truncated augmented view at depth $k$. At a high level, this is because only the roots of $T_{1,1},T_{1,1},T_{1,2},T_{1,2},\ldots,T_{i,1},T_{i,1},T_{i,2}$ can ``see'' far enough to determine which tree they are in, and, since there are two copies of each tree other than $T_{i,2}$, each root other than the root of $T_{i,2}$ has a ``twin'' that has the exact same view. The fact that the root of $T_{i,2}$ has no ``twin'' implies that it has a unique view up to depth $k$, which is the key fact that is used to prove that the $\select$-index of each $G_i \in \cG_{\Delta,k}$ is $k$. These ideas are formalized in the following results.

\begin{proposition}\label{k-1Tjbviews}
	For any $h \in \{0,\ldots,k-1\}$, any $j,j' \in \{1,\ldots,|\mathcal{T}_{\Delta,k}|\}$, and any $b,b' \in \{1,2\}$, we have that $\cB^{h}(r_{j,b})$ in $T_{j,b}$ is equal to $\cB^{h}(r_{j',b'})$ in $T_{j',b'}$.
\end{proposition}
\begin{proof}
	First, we note that it is sufficient to prove the desired result for $h=k-1$, since any two augmented truncated views that are equal at some depth $d$ are also equal at any depth less than $d$.
	
	From our construction, we note that for any integer sequence $X$, each node within distance $k-1$ from the root of the augmented tree $T_X$ is defined to have $\Delta-1$ children, and the ports leading to these children are labeled $1,\ldots,\Delta-1$. Moreover, each non-root node within distance $k-1$ from the root of the augmented tree $T_X$ is defined to have a port labeled 0 leading to its parent in $T_X$. Thus, regardless of the sequence $X$ used to construct $T_X$, the augmented truncated view at depth $k-1$ of the root of $T_X$ is always the same. 
	
	Next, we recall from the definition of every $T_{j,1}$ that the (ordered) appended path of nodes $r_{j,1},p_1,\ldots,p_{k-1}$ is labeled such that the port leading to the next node is labeled 0 and the port leading to the previous node is labeled 1. Similarly, from the definition of every $T_{j,2}$ we know that the (ordered) appended path of nodes $r_{j,2},q_1,\ldots,q_{k-1}$ is labeled such that the port leading to the next node is labeled 0 and the port leading to the previous node is labeled 1. In particular, this means that the augmented truncated view at depth $k-1$ of node $r_{j,b}$ on the appended path is the same for any $j \in \{1,\ldots,|\mathcal{T}_{\Delta,k}|\}$ and any $b \in \{1,2\}$.
	
	For each $j \in \{1,\ldots,|\mathcal{T}_{\Delta,k}|\}$ and $b \in \{1,2\}$, each $T_{j,b}$ consists of some $T_X$ with an appended path. So, the above observations are sufficient to conclude that the root node of $T_{j,b}$ has the same augmented truncated view at depth $k-1$ regardless of the values of $j$ and $b$.
\end{proof}

\begin{lemma}\label{rootsandcyclenonunique}
	For any integers $k \geq 1$ and $\Delta \geq 3$, for any $G_\alpha,G_\beta \in \cG_{\Delta,k}$ with $\alpha \leq \beta$, and for each $h \in \{0,\ldots,k\}$, both of the following statements hold:
	\begin{enumerate}
		\item $\cB^{h}(c_{m})$ in $G_\alpha$ is equal to $\cB^{h}(c_{m'})$ in $G_\beta$ for all $m,m' \in \{1,\ldots,4\alpha-1\}$, and,
		\item if $h < k$, then $\cB^{h}(r_{j,b})$ in $G_\alpha$ is equal to $\cB^{h}(r_{j',b'})$ in $G_\beta$ for all $j,j' \in \{1,\ldots,\alpha\}$ and all $b,b' \in \{1,2\}$.
	\end{enumerate}
\end{lemma}
\begin{proof}
	Fix arbitrary $G_\alpha,G_\beta \in \cG_{\Delta,k}$ such that $\alpha \leq \beta$. The proof of the two required statements proceeds by simultaneous induction on $h$. In the base case, we note that each $r_{j,b}$ in $G_\alpha$ has degree $\Delta$ and each $r_{j',b'}$ in $G_\beta$ has degree $\Delta$, so the second statement holds for $h=0$. Similarly, each $c_m$ in $G_\alpha$ has degree 3 and each $c_{m'}$ in $G_\beta$ has degree 3, so the first statement holds for $h=0$. As induction hypothesis, assume that both of the following statements hold for some $h \in \{1,\ldots,k\}$: 
	\begin{enumerate}
		\item $\cB^{h-1}(c_{m})$ in $G_\alpha$ is equal to $\cB^{h-1}(c_{m'})$ in $G_\beta$ for all $m,m' \in \{1,\ldots,4\alpha-1\}$, and,
		\item if $h < k$, then $\cB^{h-1}(r_{j,b})$ in $G_\alpha$ is equal to $\cB^{h-1}(r_{j',b'})$ in $G_\beta$ for all $j,j' \in \{1,\ldots,\alpha\}$ and all $b,b' \in \{1,2\}$.
	\end{enumerate}
	First, we set out to prove that $\cB^{h}(c_{m})$ in $G_i$ is equal to $\cB^{h}(c_{m'})$ in $G_i$ for all $m,m' \in \{1,\ldots,4\alpha-1\}$. In what follows, it is assumed that arithmetic in subscripts ``wraps around'', i.e., $c_{m+1} = c_1$ when $m=4\alpha-1$, and $c_{m-1} = c_{4\alpha-1}$ when $m=1$, and similarly for $c_{m'+1}$ and $c_{m'-1}$. The proof proceeds by noticing from our construction that $\cB^{h}(c_m)$ in $G_\alpha$ consists of: the view $\cB^{h-1}(c_{m-1})$ in $G_\alpha$, the view $\cB^{h-1}(c_{m+1})$ in $G_\alpha$, and the view $\cB^{h-1}(r_{j,b})$ in $T_{j,b}$ for some $j \in \{1,\ldots,\alpha\}$ and $b \in \{1,2\}$, with edges connecting the roots of these views  to $c_m$. In particular, the edge between $c_{m}$ and $r_{j,b}$ is labeled $2$ at $c_{m}$ and $\Delta-1$ at $r_{j,b}$, the edge between $c_{m}$ and $c_{m-1}$ is labeled $1$ at $c_{m}$ and $0$ at $c_{m+1}$, and the edge between $c_{m}$ and $c_{m+1}$ is labeled 0 at $c_{m}$ and $1$ at $c_{m+1}$. Similarly, we know from our construction that $\cB^{h}(c_{m'})$ in $G_\beta$ consists of: the view $\cB^{h-1}(c_{m'-1})$ in $G_\beta$, the view $\cB^{h-1}(c_{m'+1})$ in $G_\beta$, and the view $\cB^{h-1}(r_{j',b'})$ in $T_{j',b'}$ for some $j' \in \{1,\ldots,\alpha\}$ and $b' \in \{1,2\}$, with edges connecting the roots of these views to $c_{m'}$. In particular, the edge between $c_{m'}$ and $r_{j',b'}$ is labeled $2$ at $c_{m'}$ and $\Delta-1$ at $r_{j',b'}$, the edge between $c_{m'}$ and $c_{m'-1}$ is labeled $1$ at $c_{m'}$ and $0$ at $c_{m'+1}$, and the edge between $c_{m'}$ and $c_{m'+1}$ is labeled 0 at $c_{m'}$ and $1$ at $c_{m'+1}$.
	By the first statement in the induction hypothesis, we know that the view $\cB^{h-1}(c_{m-1})$ in $G_\alpha$ is equal to the view $\cB^{h-1}(c_{m'-1})$ in $G_\beta$, that the view $\cB^{h-1}(c_{m+1})$ in $G_\alpha$ is equal to the view $\cB^{h-1}(c_{m'+1})$ in $G_\beta$. Further, since $h-1 < k$, the second statement of the induction hypothesis tells us that the view $\cB^{h-1}(r_{j,b})$ in $T_{j,b}$ is equal to the view $\cB^{h-1}(r_{j',b'})$ in $T_{j',b'}$. It follows that $\cB^{h}(c_{m})$ in $G_\alpha$  is equal to $\cB^{h}(c_{m'})$ in $G_\beta$, as required.
	
	Next, suppose that $h < k$. We set out to prove that $\cB^{h}(r_{j,b})$ in $G_\alpha$ is equal to $\cB^{h}(r_{j',b'})$ in $G_\beta$ for all $j,j' \in \{1,\ldots,\alpha\}$ and all $b,b' \in \{1,2\}$. The key is to notice from our construction that $\cB^{h}(r_{j,b})$ in $G_\alpha$ consists of the view $\cB^{h}(r_{j,b})$ in $T_{j,b}$ and the view $\cB^{h-1}(c_m)$ in $G_\alpha$ for some $m \in \{1,\ldots,4\alpha-1\}$, with an edge joining the two roots of these views, and this edge is labeled $\Delta-1$ at endpoint $r_{j,b}$ and labeled 2 at endpoint $c_m$. Similarly, $\cB^{h}(r_{j',b'})$ in $G_\beta$ consists of $\cB^{h}(r_{j',b'})$ in $T_{j',b'}$ and $\cB^{h-1}(c_{m'})$ in $G_\beta$ for some $m' \in \{1,\ldots,4\alpha-1\}$, with an edge joining the two roots of these views, and this edge is labeled $\Delta-1$ at endpoint $r_{j',b'}$ and labeled 2 at endpoint $c_{m'}$. Since $h \leq k-1$, Proposition \ref{k-1Tjbviews} implies that $\cB^{h}(r_{j,b})$ in $T_{j,b}$ is equal to $\cB^{h}(r_{j',b'})$ in $T_{j',b'}$. Further, by the first statement of the induction hypothesis, we know that $\cB^{h-1}(c_{m})$ in $G_\alpha$ is equal to $\cB^{h-1}(c_{m'})$ in $G_\beta$. It follows that $\cB^{h}(r_{j,b})$ in $G_\alpha$  is equal to $\cB^{h}(r_{j',b'})$ in $G_\beta$, as required.
\end{proof}

\begin{lemma}\label{onlyunique}
	For any graph $G_i \in \cG_{\Delta,k}$, the root $r_{i,2}$ of $T_{i,2}$ is the only node $u \in V(G_i)$ that satisfies the property $\cB^k(u) \neq \cB^k(u')$ for all $u' \in V(G_i) \setminus \{u\}$.
\end{lemma}
\begin{proof}
	Let $v$ be an arbitrary node in $V(G_i)$. 

	 We separately consider the following exhaustive list of cases, which come directly from the various building blocks used in the construction of $G_i$.
\begin{itemize}
	
	\item Suppose that node $v$ is contained in a tree $T_{j,b}$ such that $j \neq i$ or $b \neq 2$.
	
	We prove that there is another node $v'$ in $G_i$ that has the same truncated view at depth $k$. From the construction of $G_i$, if $j \neq i$ or $b \neq 2$, then there are two copies of $T_{j,b}$ in $G_i$. 
	Let $v'$ be the corresponding copy of node $v$ that is contained in the other copy of $T_{j,b}$ (i.e., the $T_{j,b}$ that does not contain $v$). 
	Let $r$ be the root of the $T_{j,b}$ containing $v$, and 
	let $r'$ be the root of the $T_{j,b}$ containing $v'$. 
	By construction, $r$ is connected by an edge $e$ to some node $c_m \in C_i$, the port number at $r$ on this edge is $\Delta-1$, and the port number at $c_m$ on this edge is 2. 
	Similarly, by construction, $r'$ is connected by an edge $e'$ to some node $c_{m'} \in C_i$, with $m' \neq m$, the port number at $r'$ on this edge is $\Delta-1$, and the port number at $c_{m'}$ on this edge is 2. 
	
	Consider each root-to-leaf path $\pi$ in the view $\cB^k(v)$. There are two cases to consider:
	\begin{itemize}
		\item If all nodes in $\pi$ are contained in $T_{j,b}$, then the same path appears in the view $\cB^k(v')$ since $v'$ is the corresponding copy of $v$ in the other copy of $T_{j,b}$.
		\item If there exists a node in $\pi$ that is not in $T_{j,b}$, then consider the first such node $w$ along the path starting from the root of $\cB^k(v)$. By the construction of $G_i$, the only node in $T_{j,b}$ that has a neighbour outside of $T_{j,b}$ is $r$, and this neighbour is $c_m$. It follows that $w = c_m$, and the parent of $w$ is $r$. Let $\pi_{in}$ denote the prefix of the path $\pi$ that is entirely contained in $T_{j,b}$ (i.e., starting at $v$ and ending at $r$), and let $\pi_{out}$ denote the remainder of the path (i.e., starting at $c_m$ until the leaf of $\pi$). Since $v'$ is the corresponding copy of $v$ in the other copy of $T_{j,b}$, the same path prefix $\pi_{in}$ appears in the view $\cB^k(v')$ from $v'$ to $r'$. As observed above, by construction, $r'$ is connected by an edge $e'$ to $c_{m'}$ and is labeled with the same port numbers as the edge $e = \{r,c_m\}$. Finally, by Lemma \ref{rootsandcyclenonunique}, the views $\cB^{|\pi_{out}|}(c_m)$ and $\cB^{|\pi_{out}|}(c_{m'})$ are the same, so the path suffix $\pi_{out}$ appears in the view $\cB^k(v')$ as well. This concludes the proof that the entire path $\pi$ appears as a root-to-leaf path in $\cB^k(v')$ as well.
	\end{itemize}
By symmetry (swapping the roles of $v$ and $v'$), each root-to-leaf path $\pi'$ in $\cB^k(v')$ appears as a root-to-leaf path in $\cB^k(v)$, which concludes the proof that $\cB^k(v)=\cB^k(v')$.
	

	\item Suppose that $v \in C_i$.
	
	We prove that there is at least one other node $v'$ in $G_i$ that has the same truncated view at depth $k$. By Lemma \ref{rootsandcyclenonunique} with $\alpha=\beta=i$, we know that $\cB^{k}(c_{m}) = \cB^{k}(c_{m'})$ for all $m,m' \in \{1,\ldots,4\alpha-1\}$, which means that $\cB^{k}(v)$ is the same for all nodes in $v \in C_i$.
	
	\item Suppose that $v$ is the root $r_{i,2}$ of $T_{i,2}$.
	
	We prove that no other node $v'$ in $G_i$ has the same truncated view at depth $k$ as $v$. We consider all possible cases for $v'$:
	\begin{itemize}
		\item Suppose that $v' \in C_i$.
		In this case, by construction, $v'$ has degree 3, and has two neighbours with degree 3 (its neighbours in $C_i$), and one neighbour with degree $\Delta \geq 3$ (the root of some $T_{j,b}$). However, $v = r_{i,2}$ has at least one neighbour with degree 2, i.e., its neighbour in the appended path. Thus, $\cB^{1}(v) \neq \cB^{1}(v')$, which implies that $\cB^{k}(v) \neq \cB^{k}(v')$.
		
		\item Suppose that $v' = r_{j,b}$ for some $j \neq i$ or $b \neq 2$.
		First, consider the case where $b \neq 2$, i.e., $b=1$. By construction, the tree $T_{j,1}$ has an appended path of length $k+1$ starting at its root $r_{j,1} = v'$, and the port sequence $\pi$ along this path in $\cB^{k}(v')$ has its first port equal to 0 and its final port equal to 1. However, by construction, the tree $T_{i,2}$ has an appended path starting at its root $r_{i,2} = v$, and the port sequence along this path in $\cB^{k}(v)$ has first port equal to 0 and its final port equal to 0. (This is precisely the reason why the construction of $T_{X,2}$ swaps the port numbers at node $p_k$ in the appended path.) Since no other port at $v$ is labeled 0, it follows that the port sequence $\pi$ does not exist in $\cB^{k}(v)$ starting at $v$. This proves that $\cB^{k}(v) \neq \cB^{k}(v')$.
		
		Next, consider the case where $b=2$ and $j \neq i$. Since $j \neq i$, it follows from our construction that $T_{i,2}$ is built using a tree $T_{X}$ and $T_{j,2}$ is built using a tree $T_{Y}$ where $X=(x_1,\ldots,x_z)$ and $Y=(y_1,\ldots,y_z)$ are distinct sequences of integers in the range $1,\ldots,\Delta-1$. Recall that $T_{X}$ has height $k+1$, has $z$ nodes at distance $k$ from its root $r_{i,2}$, and the $i^{th}$ node at distance $k$ from the root has $x_i$ degree-1 neighbours. (The ordering of nodes at distance $k$ is based on the lexicographic ordering of the sequence of ports leading from the root to each such node.) Similarly, $T_{Y}$ has height $k+1$, has $z$ nodes at distance $k$ from its root $r_{j,2}$, and the $i^{th}$ node at distance $k$ from the root has $y_i$ degree-1 neighbours. As $X$ and $Y$ are distinct, there exists an index $i$ such that $x_i \neq y_i$. Therefore, in the augmented view $\cB^{k}(r_{i,2})$, there is a path with port sequence $\pi$ starting at $r$ whose other endpoint is labeled with $x_i$, and, in the augmented view $\cB^{k}(r_{j,2})$, the path with the same port sequence $\pi$ starting at $r$ exists, but the other endpoint is labeled with $y_i \neq x_i$. This proves that $\cB^{k}(r_{i,2}) \neq \cB^{k}(r_{j,2})$, i.e., $\cB^{k}(v) \neq \cB^{k}(v')$.
		
		\item Suppose that $v' \in T_{j,b}$ for some $j \in \{1,\ldots,i\}$, $b \in \{1,2\}$, and $v' \neq r_{j,b}$.
	     By construction, every root-to-leaf path in the tree $T_{j,b}$ has length exactly $k+1$. As $v=r_{j,b}$, it follows that $\cB^{k}(v)$ contains no node that has degree 1 in $G_i$. Moreover, as $v'$ is assumed to be a node in $T_{j,b}$ other than the root $r_{j,b}$, the distance from $v'$ to a leaf node in $T_{j,b}$ is at most $k$. Therefore, $\cB^{k}(v')$ contains a node $w$ that has degree 1 in $G_i$. It follows that $\cB^{k}(v) \neq \cB^{k}(v')$.
	\end{itemize}
	
\end{itemize}
\end{proof}

%
%

\begin{lemma}
$\psi_{\select}(G_i)=k$ for any graph $G_i \in \cG_{\Delta,k}$.
\end{lemma}
\begin{proof}
Consider an arbitrary graph $G_i \in \cG_{\Delta,k}$. We first observe that no node in $G_i$ has a unique augmented truncated view at depth $k-1$: Lemma \ref{rootsandcyclenonunique} proves that this is true for the root $r_{i,2}$ of $T_{i,2}$, and Lemma \ref{onlyunique} proves that this is true for all other nodes in $G_i$. Thus, by Proposition \ref{uniquekview}, it follows that $\psi_{\select}(G_i) \geq k$. 

Next, we give an algorithm that solves $\select$ in time $k$ in any $G_i \in \cG_{\Delta,k}$ given a map of the graph $G_i$. First, each node uses the map to deduce the value of $k$ by subtracting 2 from the shortest distance from a leaf in $G_i$ to a node in the cycle $C_i$. Then, using $k$ communication rounds, each node $v$ learns $\cB^k(v)$. Finally, each node $v$ finds, in the map of $G_i$, the node with the unique augmented truncated view at depth $k$ (Lemma \ref{onlyunique} guarantees that there is exactly one such node), and compares it to its own $\cB^k(v)$. If these match, then the node outputs 1, and otherwise outputs 0. This shows that $\psi_{\select}(G_i) \leq k$, which concludes the proof.
\end{proof}


To obtain a lower bound on the size of advice, we first observe that, using only its truncated augmented view at depth $k$, a root node $r_{j,2}$ of $T_{j,2}$ cannot determine whether it is in $G_j$ or some other graph $G_i$ with $i \neq j$. Using this fact, we show that, for any algorithm using insufficient advice, there exists a $G_i \in \cG_{\Delta,k}$ and an $r_{j,2}$ in $G_i$ with $i \neq j$ that is ``fooled'' into outputting 1, and that $r_{i,2}$ will also output 1, which proves that the algorithm does not solve $\select$.

\begin{lemma}\label{sameviewsameroots}
	For any $G_\alpha,G_\beta \in \cG_{\Delta,k}$ with $\alpha \leq \beta$, the view $\cB^{k}(r_{j,b})$ in $G_\alpha$ is the same as $\cB^{k}(r_{j,b})$ in $G_\beta$ for any $j \in \{1,\ldots,\alpha\}$ and $b \in \{1,2\}$.
\end{lemma}
\begin{proof}
A node $r_{j,b}$ in $G_\alpha$ is the root node of a tree $T_{j,b}$ in $G_\alpha$, and a node $r_{j,b}$ in $G_\beta$ is the root node of a tree $T_{j,b}$ in $G_\beta$. The part of $\cB^{k}(r_{j,b})$ in $G_\alpha$ belonging to $T_{j,b}$ is equal to the part of $\cB^{k}(r_{j,b})$ in $G_\beta$ belonging to $T_{j,b}$. The remainder of node $r_{j,b}$'s view in $G_\alpha$ consists of some $\cB^{k-1}(c_{m})$ along with an edge between $r_{j,b}$ and $c_m$ with the port at $r_{j,b}$ labeled $\Delta-1$ and the port at $c_m$ labeled 2. The remainder of node $r_{j,b}$'s view in $G_\beta$ consists of some $\cB^{k-1}(c_{m'})$ along with an edge between $r_{j,b}$ and $c_{m'}$ with the port at $r_{j,b}$ labeled $\Delta-1$ and the port at $c_{m'}$ labeled 2. However, by Lemma \ref{rootsandcyclenonunique}, we know that $\cB^{k-1}(c_{m})$ in $G_\alpha$ is equal to $\cB^{k-1}(c_{m'})$ in $G_\beta$, which concludes the proof.
\end{proof}

\begin{theorem}
Consider any algorithm $\cA$ that solves $\select$ in $\psi_{\select}(G)$ rounds for every graph $G$. For all integers $k\ge 1,\Delta \ge 5$, there exists a
graph $G$ with maximum degree $\Delta$ and with $\psi_{\select}(G)=k$ for which algorithm $\cA$ requires
advice of size $\Omega((\Delta-1)^k \log \Delta)$.
\end{theorem}
\begin{proof}
To obtain a contradiction, assume that there exists an algorithm $\cA$ that solves $\select$ in $k$ rounds for the class of graphs $\cG_{\Delta,k}$ with the help of an oracle that provides advice of size $\frac{1}{8}(\Delta-1)^k \log_2 \Delta$. There are at most $2^{1+(\frac{1}{8}(\Delta-1)^k \log_2 \Delta)} \leq 2^{\frac{1}{4}(\Delta-1)^k \log_2 \Delta} = \Delta^{\frac{1}{4}(\Delta-1)^k}$ binary advice strings whose length is at most $\frac{1}{8}(\Delta-1)^k \log_2 \Delta$. 
By Fact \ref{classCountG}, the total number of graphs in $\cG_{\Delta,k}$ is $|\mathcal{T}_{\Delta,k}| = (\Delta-1)^{(\Delta-2)\cdot(\Delta-1)^{k-1}}  > \Delta^{\frac{1}{2}(\Delta-2)\cdot(\Delta-1)^{k-1}} > \Delta^{\frac{1}{4}(\Delta-1)^k}$. Therefore, by the Pigeonhole Principle, the oracle provides the same advice for at least two graphs $G_\alpha$ and $G_\beta$ from $\cG_{\Delta,k}$. Suppose $\alpha < \beta$. 
By Lemma \ref{sameviewsameroots}, the root node $r_{\alpha,2}$ in $T_{\alpha,2}$ of $G_\alpha$ and the root node $r_{\alpha,2}$ in $T_{\alpha,2}$ of $G_\beta$ have the same augmented truncated view at the end of $k$ communication rounds. Hence, the output of $r_{\alpha,2}$ is the same when the algorithm $\cA$ is executed in $G_\alpha$ and $G_\beta$. By Lemma \ref{uniquekview}, since the node $r_{\alpha,2}$ in $T_{\alpha,2}$ of $G_\alpha$ is the only node with a unique augmented truncated view at depth $k$, it will output 1 when $\cA$ is executed in $G_\alpha$. Therefore, the node $r_{\alpha,2}$ in $T_{\alpha,2}$ of $G_\beta$ also outputs 1 when $\cA$ is executed in $G_\beta$. But, according to the construction of $G_\beta$, there are two copies of $T_{\alpha,2}$ in $G_\beta$, and, by Lemma \ref{sameviewsameroots} (with $\alpha=\beta$), the two copies of node $r_{\alpha,2}$ have the same augmented truncated view at depth $k$. Therefore, there are two nodes in $G_\beta$ that output 1, which contradicts the correctness of $\mathcal{A}$.
\end{proof}

\section{Port Election vs. Selection}\label{secPE}

In this section, we prove that the size of advice needed to solve $\pe$ in minimum time is exponentially larger than the size of advice needed to solve $\select$. More specifically, for any fixed integers $k \geq 1$ and $\Delta \geq 4$, we construct a class of graphs such that: $\psi_{\select}(G) = \psi_{\pe}(G) = k$ for each graph in the class, solving $\select$ in time $k$ in this class can be done with advice of size at most $O((\Delta-1)^k\log \Delta)$ (in view of Theorem \ref{S-ub}), but there exists a graph in the class for which the size of advice needed to solve $\pe$ in time $k$ is at least $\Omega((\Delta-1)^{(\Delta-2)(\Delta-1)^{k-1}}\log\Delta)$.

\subsection{Construction of $\cU_{\Delta,k}$}\label{constructUDeltak}

Consider any positive integers $\Delta \geq 4$ and $k \geq 1$. Recall, from Building Block 2 of Section \ref{constructGDeltak}, the construction of the set of augmented trees $\mathcal{T}_{\Delta,k}$. Our construction proceeds by first constructing the following template graph $U$ with maximum degree $2\Delta-1$ (which is illustrated in Figure \ref{constructU}):
\begin{figure}[h]
	\centering
	\includegraphics[scale=1.2]{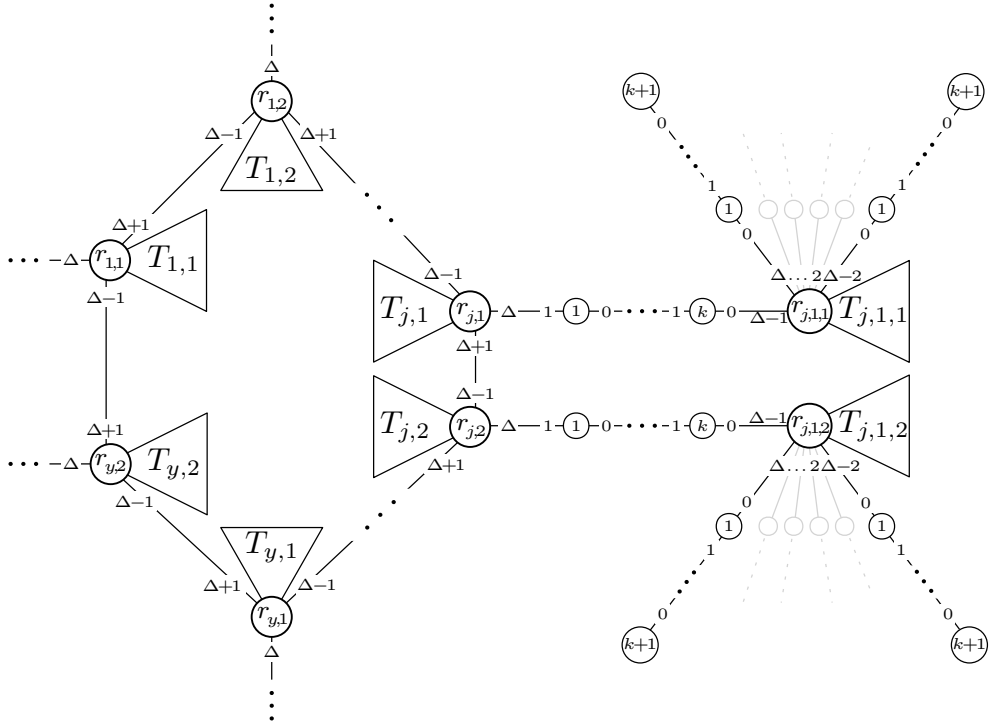}
	\caption{The graph $U$, where $y = |\mathcal{T}_{\Delta,k}|$. Node labels in the diagram are not known by the nodes: they correspond to variables used in the description of the construction, or to indicate the lengths of paths.}
	\label{constructU}
\end{figure}
\begin{enumerate}
	\item Create a graph $U$ by taking the disjoint union of all the trees $T_{j,b}$ for $j \in \{1,\ldots,|\mathcal{T}_{\Delta,k}|\}$ and $b \in \{1,2\}$ (as they are defined in Building Block 3 of Section \ref{constructGDeltak}). Add edges so that the roots of these trees form the cycle $r_{1,1}, r_{1,2}, r_{2,1}, r_{2,2}, \ldots, r_{|\mathcal{T}_{\Delta,k}|,1},r_{|\mathcal{T}_{\Delta,k}|,2},r_{1,1}$. As for port numbers on these added edges, orient the cycle by labeling the port at $r_{1,1}$ leading to $r_{1,2}$ as $\Delta+1$, label the port at $r_{1,2}$ leading to $r_{1,1}$ as $\Delta-1$, and keep alternating between $\Delta+1$ and $\Delta-1$ around the cycle.
	\item For each $j \in \{1,\ldots,|\mathcal{T}_{\Delta,k}|\}$, add two more copies of $T_{j,1}$ to $U$. These additional trees will be denoted by $T_{j,1,1}$ and $T_{j,1,2}$, and their roots will be denoted by $r_{j,1,1}$ and $r_{j,1,2}$, respectively.
	\item For each $j \in \{1,\ldots,|\mathcal{T}_{\Delta,k}|\}$, create a path of length $k+1$ between $r_{j,1}$ and $r_{j,1,1}$ (by introducing $k$ new nodes). Label the new port created at $r_{j,1}$ as $\Delta$, label the new port created at $r_{j,1,1}$ as $\Delta-1$, and label the remaining ports on the new path by assigning a 1 in the direction leading towards $r_{j,1}$ and assigning a 0 in the direction leading towards $r_{j,1,1}$. Similarly, for each $j \in \{1,\ldots,|\mathcal{T}_{\Delta,k}|\}$, create a path of length $k+1$ between $r_{j,2}$ and $r_{j,1,2}$ (by introducing $k$ new nodes), label the new port created at $r_{j,2}$ as $\Delta$, label the new port created at $r_{j,1,2}$ as $\Delta-1$, and label the remaining ports on the new path by assigning a 1 in the direction leading towards $r_{j,2}$ and assigning a 0 in the direction leading towards $r_{j,1,2}$.
	\item For each $j \in \{1,\ldots,|\mathcal{T}_{\Delta,k}|\}$, introduce $\Delta-1$ new paths of length $k+1$, each with $r_{j,1,1}$ as an endpoint. Label the $\Delta-1$ new ports at $r_{j,1,1}$ using the integers $\Delta,\ldots,2\Delta-2$. At each of the new nodes introduced to form these paths, label the port leading towards $r_{j,1,1}$ with 0, and the other port with 1. Similarly, for each $j \in \{1,\ldots,|\mathcal{T}_{\Delta,k}|\}$, introduce $\Delta-1$ new paths of length $k+1$, each with $r_{j,1,2}$ as an endpoint. Label the $\Delta-1$ new ports at $r_{j,1,2}$ using the integers $\Delta,\ldots,2\Delta-2$. At each of the new nodes introduced to form these paths, label the port leading towards $r_{j,1,2}$ with 0, and the other port with 1.
\end{enumerate}


 Using the template graph $U$, we construct each graph of the class $\cU_{\Delta,k}$ as follows. Consider every integer sequence $\sigma = (s_1,\ldots,s_{|\mathcal{T}_{\Delta,k}|})$ where $s_j \in \{1,\ldots,\Delta-1\}$ for each $j \in \{1,\ldots,|\mathcal{T}_{\Delta,k}|\}$. For each such sequence $\sigma$, construct the graph $G_\sigma$ by taking the template graph $U$ and, for each $j \in \{1,\ldots,|\mathcal{T}_{\Delta,k}|\}$, exchanging the ports $\Delta-1$ and $\Delta-1+s_j$ at both of the nodes $r_{j,1,1}$ and $r_{j,1,2}$. The class $\cU_{\Delta,k}$ consists of all such graphs $G_\sigma$. The number of different sequences $\sigma$ is $(\Delta-1)^{|\mathcal{T}_{\Delta,k}|}$, which we use along with Fact \ref{classCountG} to get the following result about the number of graphs in $\cU_{\Delta,k}$.
 
 \begin{fact}\label{classCountU}
 	$|\cU_{\Delta,k}| = (\Delta-1)^{|\mathcal{T}_{\Delta,k}|} = (\Delta-1)^{(\Delta-1)^{(\Delta-2)\cdot(\Delta-1)^{k-1}}}$ for any positive integers $\Delta \geq 4$ and $k \geq 1$.
 \end{fact}

In the template graph $U$ or any fixed $G_\sigma \in \cU_{\Delta,k}$, we refer to the set of nodes $\{r_{j,b}\ |\ j \in \{1,\ldots,|\mathcal{T}_{\Delta,k}|\} \textrm{ and } b \in \{1,2\} \}$ as the \emph{cycle nodes}.

\subsection{Minimum Election Time and Advice}
For each graph in $\cU_{\Delta,k}$ (where $\Delta \geq 4$ and $k \geq 1$), we show that the $\select$-index and the $\pe$-index are both $k$.

To prove that the $\select$-index is at least $k$ for any graph $G_\sigma \in \cU_{\Delta,k}$, the idea is that $G_\sigma$ was carefully constructed such that, when considering truncated views up to distance $k-1$, each node $v$ has at least one `twin' elsewhere in the graph with the same view. We divide the proof into three parts (Propositions \ref{prop:k-1rjbviewsGS}, \ref{prop:kTijviewsGS}, and \ref{prop:kHijviewsGS}) based on $v$'s location within $G_\sigma$. First, we prove this about the cycle nodes in $G_\sigma$.

\begin{proposition}\label{prop:k-1rjbviewsGS}
	Consider any $G_\sigma \in \cU_{\Delta,k}$. For any $h \in \{0,\ldots,k-1\}$, any $j,j' \in \{1,\ldots,|\mathcal{T}_{\Delta,k}|\}$, and any $b,b' \in \{1,2\}$, we have that $\cB^{h}(r_{j,b})$ in $G_\sigma$ is equal to $\cB^{h}(r_{j',b'})$ in $G_\sigma$.
\end{proposition}
\begin{proof}
	It is sufficient to prove the claim about the template graph $U$, as any differences between $U$ and $G_\sigma$ are at a distance greater than $k$ from $r_{j,b}$ and $r_{j',b'}$. The proof proceeds by induction on $h$. For the base case, consider $h=0$.  By the construction of the template graph $U$, all nodes $r_{j,b}$ with $j \in \{1,\ldots,|\mathcal{T}_{\Delta,k}|\}$ and $b \in \{1,2\}$ have the same degree $\Delta+2$, which proves the desired result for $h=0$.
	
	As induction hypothesis, assume that, for any $h \in \{1,\ldots,k-1\}$, any $j,j' \in \{1,\ldots,|\mathcal{T}_{\Delta,k}|\}$, and any $b,b' \in \{1,2\}$, we have that $\cB^{h-1}(r_{j,b})$ in $U$ is equal to $\cB^{h-1}(r_{j',b'})$ in $U$.
	
	To prove the inductive step, consider any $j,j' \in \{1,\ldots,|\mathcal{T}_{\Delta,k}|\}$, and any $b,b' \in \{1,2\}$. By the construction of $U$, the view $\cB^{h}(r_{j,b})$ in $U$ consists of the following parts:
	\begin{enumerate}[label=(\arabic*)]
		\item The view of $r_{j,b}$ at distance $h$ within $T_{j,b}$,
		\item A path $P$ of length $h$ with:
		\begin{itemize}
			\item the node $r_{j,b}$ as one endpoint of $P$ with port labeled $\Delta$,
			\item $h-1$ internal nodes of degree 2 with port 1 leading in the direction towards $r_{j,b}$ (and the other port labeled 0), and,
			\item a second endpoint of $P$ with port 1 leading in the direction towards $r_{j,b}$, but with degree 2 in $U$.
		\end{itemize}
		\item An edge on the cycle from $r_{j,b}$ to some $r_{\alpha_1,\beta_1}$ with $\alpha_1 \in \{1,\ldots,|\mathcal{T}_{\Delta,k}|\}$ and $\beta_1 \in \{1,2\}$, with the port at $r_{j,b}$ labeled $\Delta+1$ and the port at $r_{\alpha_1,\beta_1}$ labeled $\Delta-1$, together with the view $\cB^{h-1}(r_{\alpha_1,\beta_1})$, and,
		\item An edge on the cycle from $r_{j,b}$ to some $r_{\alpha_2,\beta_2}$ with $\alpha_2 \in \{1,\ldots,|\mathcal{T}_{\Delta,k}|\}$ and $\beta_2 \in \{1,2\}$, with the port at $r_{j,b}$ labeled $\Delta-1$ and the port at $r_{\alpha_2,\beta_2}$ labeled $\Delta+1$, together with the view $\cB^{h-1}(r_{\alpha_2,\beta_2})$.
	\end{enumerate}
	Similarly, by the construction of $U$, the view $\cB^{h}(r_{j',b'})$ in $U$ consists of the following parts:
	\begin{enumerate}[label=(\arabic*)]
		\item The view of $r_{j',b'}$ at distance $h$ within $T_{j',b'}$,
		\item A path $P'$ of length $h$ with:
		\begin{itemize}
			\item the node $r_{j',b'}$ as one endpoint of $P'$ with port labeled $\Delta$,
			\item $h-1$ internal nodes of degree 2 with port 1 leading in the direction towards $r_{j',b'}$ (and the other port labeled 0), and,
			\item a second endpoint of $P$ with port 1 leading in the direction towards $r_{j',b'}$, but with degree 2 in $U$.
		\end{itemize}
		\item An edge on the cycle from $r_{j',b'}$ to some $r_{\alpha_1',\beta_1'}$ with $\alpha_1' \in \{1,\ldots,|\mathcal{T}_{\Delta,k}|\}$ and $\beta_1' \in \{1,2\}$, with the port at $r_{j',b'}$ labeled $\Delta+1$ and the port at $r_{\alpha_1',\beta_1'}$ labeled $\Delta-1$, together with the view $\cB^{h-1}(r_{\alpha_1',\beta_1'})$, and,
		\item An edge on the cycle from $r_{j',b'}$ to some $r_{\alpha_2',\beta_2'}$ with $\alpha_2' \in \{1,\ldots,|\mathcal{T}_{\Delta,k}|\}$ and $\beta_2' \in \{1,2\}$, with the port at $r_{j',b'}$ labeled $\Delta-1$ and the port at $r_{\alpha_2',\beta_2'}$ labeled $\Delta+1$, together with the view $\cB^{h-1}(r_{\alpha_2',\beta_2'})$.
	\end{enumerate}
	
	First, we see that part (1) of the two views are equal by Lemma \ref{k-1Tjbviews} since $h \leq k-1$. Part (2) of the two views are equal since $P$ and $P'$ are labeled in the same way. Part (3) of the two views are equal since, by the induction hypothesis, the views $\cB^{h-1}(r_{\alpha_1,\beta_1})$ and $\cB^{h-1}(r_{\alpha_1',\beta_1'})$ are equal. Similarly, part (4) of the two views are equal since, by the induction hypothesis, the views $\cB^{h-1}(r_{\alpha_2,\beta_2})$ and $\cB^{h-1}(r_{\alpha_2',\beta_2'})$ are equal. This concludes the proof that $\cB^{h}(r_{j,b})$ in $U$ is equal to $\cB^{h}(r_{j',b'})$ in $U$, which completes the inductive step.
\end{proof}

Next, we show that, when considering truncated views up to distance $k-1$, each node in each $T_{j,b}$ in $G_\sigma$ has a `twin' elsewhere in the graph with the same view. In fact, we will prove a stronger version of the result so that it holds for truncated views up to distance $k$, which we will reuse later when considering the advice size for the Port Election task in $G_\sigma$.

\begin{proposition}\label{prop:kTijviewsGS}
		Consider any $G_\sigma \in \cU_{\Delta,k}$, any $j \in \{1,\ldots,|\mathcal{T}_{\Delta,k}|\}$ and any $b \in \{1,2\}$. For each node $v \in T_{j,b} - \{r_{j,b}\}$, there exists a node $v' \neq v$ such that $\cB^{k}(v)$ in $G_\sigma$ is equal to $\cB^{k}(v')$ in $G_\sigma$.
\end{proposition}
\begin{proof}
	
	Consider any fixed $j \in \{1,\ldots,|\mathcal{T}_{\Delta,k}|\}$ and $b \in \{1,2\}$, and consider any node $v \in T_{j,b} - \{r_{j,b}\}$. There are two cases to consider depending on which part of $T_{j,b}$ the node $v$ is located in:
	\begin{itemize}
		\item Suppose that $v$ is contained in the augmented tree $T_j$. We prove that there is a node in $T_{j,3-b}$ with the same view as $v$ up to distance $k$. Note that $T_{j,b}$ and $T_{j,3-b}$ are the trees $T_{j,1}$ and $T_{j,2}$. According to the construction in Building Block 3 of Section \ref{constructGDeltak}, $T_{j,1}$ and $T_{j,2}$ are constructed using the same tree $T_j$, so it follows that there is a node $v'$ in $T_{j,3-b}$ that is at the same location within $T_j$ as $v$ is located within $T_{j,b}$. We show that $v$ and $v'$ have the same view up to distance $k$. The view of $v$ up to distance $k$ consists of two parts: paths of length $k$ contained entirely within $T_j$ (we'll call these type-1 paths), and, paths of length $k$ that pass through $r_{j,b}$ and have a subpath of length at most $k-1$ outside of $T_j$ (we'll call these type-2 paths). Regarding type-1 paths: by the choice of $v'$ within $T_{j,3-b}$ (i.e., $v$ and $v'$ are copies of the same node within $T_j$) the set of paths that originate at $v$ and lie entirely within $T_j$ is the same as the set of paths that originate at $v'$ and lie entirely within $T_j$. Regarding type-2 paths: the fact established in the previous sentence implies that the path with $v$ and $r_{j,b}$ as endpoints in $T_{j,b}$ is the same as the path with $v'$ and $r_{j,3-b}$ as endpoints in $T_{j,3-b}$. By Proposition \ref{prop:k-1rjbviewsGS}, the views of $r_{j,b}$ and $r_{j,3-b}$ are identical up to distance $k-1$. So, the set of paths of length $h \leq k-1$ starting from $r_{j,b}$ is the same as the set of paths of length $h$ starting from $r_{j,3-b}$. It follows that the set of type-2 paths starting at $v$ is the same as the set of type-2 paths starting at $v'$. This concludes the proof that the views of $v$ and $v'$ are identical up to distance $k$.
		\item Suppose that $v$ is contained in the path $P$ appended to $T_j$ to form $T_{j,b}$. Consider any $j' \in \{1,\ldots,|\mathcal{T}_{\Delta,k}|\} - \{j\}$. According to the construction in Building Block 3 of Section \ref{constructGDeltak}, $T_{j,b}$ and $T_{j',b}$ are constructed by appending the same path $P$ of length $k+1$ to $T_{j}$ and $T_{j'}$, respectively. Denoting by $h \geq 1$ the distance from $v$ to $r_{j,b}$, let $v'$ be the node at distance $h$ from $r_{j',b}$ in the appended path $P$ of $T_{j',b}$. We show that $v$ and $v'$ have the same view up to distance $k$. The view of $v$ up to distance $k$ consists of two parts: paths of length $k$ contained entirely within $P$ (we'll call these type-1 paths), and, paths of length $k$ that pass through $r_{j,b}$ and have a subpath of length $k-h$ outside of $P$ (we'll call these type-2 paths). Regarding type-1 paths: by the choice of $v'$ within the appended path $P$ (i.e., $v$ and $v'$ are copies of the same node within $P$) the set of paths that originate at $v$ and lie entirely within $P$ is the same as the set of paths that originate at $v'$ and lie entirely within $P$. Regarding type-2 paths: the fact established in the previous sentence implies that the path with $v$ and $r_{j,b}$ as endpoints in $P$ is the same as the path with $v'$ and $r_{j',b}$ as endpoints in $T_{j',b}$. By Proposition \ref{prop:k-1rjbviewsGS}, the views of $r_{j,b}$ and $r_{j',b}$ are identical up to distance $k-h$. So, the set of paths of length $k-h$ starting from $r_{j,b}$ is the same as the set of paths of length $k-h$ starting from $r_{j',b}$. It follows that the set of type-2 paths starting at $v$ is the same as the set of type-2 paths starting at $v'$. This concludes the proof that the views of $v$ and $v'$ are identical up to distance $k$.
	\end{itemize}

\end{proof}

Finally, we consider the remaining nodes in $G_\sigma$, i.e., each node in $G_\sigma$ that is reachable from an $r_{j,b}$ using a path whose first outgoing port is $\Delta$. We show that, when considering truncated views up to distance $k-1$, each such node has a `twin' elsewhere in the graph with the same view. Once again, we will prove a stronger version of the result so that it holds for truncated views up to distance $k$, which will be reused later when considering the advice size for the Port Election task in $G_\sigma$.
For any $j \in \{1,\ldots,|\mathcal{T}_{\Delta,k}|\}$, denote by $H_{j,1}$ the subtree rooted at $r_{j,1}$ consisting of all nodes reachable from $r_{j,1}$ by a path with first outgoing port $\Delta$. Similarly, denote by $H_{j,2}$ the subtree rooted at $r_{j,2}$ consisting of all nodes reachable from $r_{j,2}$ by a path with first outgoing port $\Delta$.

\begin{fact}\label{H1H2identical}
	For any $G_\sigma \in \cU_{\Delta,k}$ and any $j \in \{1,\ldots,|\mathcal{T}_{\Delta,k}|\}$, the trees $H_{j,1}$ and $H_{j,2}$ are identical.
\end{fact}
\begin{proof}
	The result is true about the template graph $U$, which can be verified by the description of its construction. To create $G_\sigma$ from template graph $U$, the same port $\Delta-1+s_j$ was swapped with $\Delta-1$ at the nodes $r_{j,1,1}$ and $r_{j,1,2}$, so it follows that the subtrees are identical in $G_\sigma$ as well.
\end{proof}

\begin{proposition}\label{prop:kHijviewsGS}
	Consider any $G_\sigma \in \cU_{\Delta,k}$ and any $j \in \{1,\ldots,|\mathcal{T}_{\Delta,k}|\}$. For any node in $v_1 \in H_{j,1}$ other than the root $r_{j,1}$, the corresponding copy $v_2$ of $v_1$ in $H_{j,2}$ has the same view as $v_1$ up to distance $k$ in $G_\sigma$.
\end{proposition}
\begin{proof}
	Consider any $v_1 \in H_{j,1}$, and let $v_2$ be the corresponding copy of $v_1$ in $H_{j,2}$. There are two cases to consider:
	\begin{itemize}
		\item Suppose that the distance between $v_1$ and $r_{j,1}$ is greater than $k$. Then $v_1$'s view up to distance $k$ does not include $r_{j,1}$, and $v_2$'s view up to distance $k$ does not include $r_{j,2}$, i.e., the views of $v_1$ and $v_2$ are contained strictly inside $H_{j,1}$ and $H_{j,2}$, respectively. By Fact \ref{H1H2identical}, $H_{j,1}$ and $H_{j,2}$ are identical, which implies that the views of $v_1$ and $v_2$ up to distance $k$ are identical in $G_\sigma$.
		\item Suppose that the distance between $v_1$ and $r_{j,1}$ is at most $k$. Let $d \geq 1$ denote the distance between $v_1$ and $r_{j,1}$ (which is also the distance between $v_2$ and $r_{j,2}$). The part of $v_1$'s view that lies completely within $H_{j,1}$ is identical to the part of $v_2$'s view that lies completely within $H_{j,2}$, since $H_{j,1}$ and $H_{j,2}$ are identical (by Fact \ref{H1H2identical}). The rest of $v_1$'s view consists of $r_{j,1}$'s view up to distance $k-d$, and the rest of $v_2$'s view consists of $r_{j,2}$'s view up to distance $k-d$. However, as $d \geq 1$, Proposition \ref{prop:k-1rjbviewsGS} tells us that $\cB^{k-d}(r_{j,1})$ in $G_\sigma$ is equal to $\cB^{k-d}(r_{j,2})$ in $G_\sigma$, which completes the proof that the views of $v_1$ and $v_2$ up to distance $k$ are identical in $G_\sigma$.
	\end{itemize}
\end{proof}

Together, Propositions \ref{prop:k-1rjbviewsGS}, \ref{prop:kTijviewsGS}, and \ref{prop:kHijviewsGS} give us the following result, which tells that no node has a unique view up to distance $k-1$ in $G_\sigma$. This implies that the $\select$-index of $G_S$ is at least $k$.
\begin{lemma}\label{k-1viewsGS}
	For any $G_\sigma \in \cU_{\Delta,k}$ and any node $v \in G_\sigma$, there exists a node $v' \in G_\sigma - \{v\}$ such that $\cB^{k-1}(v) = \cB^{k-1}(v')$ in $G_\sigma$.
\end{lemma}

\begin{corollary}\label{cor:Sgeqk}
	$\psi_{\select}(G_\sigma) \geq k$ for any graph $G_\sigma \in \cU_{\Delta,k}$.
\end{corollary}


We now turn our attention to the $\pe$-index of graphs in $\cU_{\Delta,k}$. We set out to prove that the $\pe$-index is at most $k$ for any graph $G_\sigma \in \cU_{\Delta,k}$. This fact will allow us to conclude that $\psi_{\pe}(G_\sigma) = \psi_{\select}(G_\sigma) = k$ since, by Fact \ref{indexHierarchy} and Corollary \ref{cor:Sgeqk}, we have already shown that $\psi_{\pe}(G_\sigma) \geq \psi_{\select}(G_\sigma) \geq k$. 

We first prove a structural result which shows that each cycle node in $G_\sigma$ has a unique truncated view up to distance $k$. This fact will be exploited by our Port Election algorithm when choosing a leader node, since, given a complete map of the graph, each node can find the lexicographically smallest view from among the views of the cycle nodes, and compare it to its own view up to distance $k$.
%

\begin{lemma}\label{kviewsGS}
	For any integers $\Delta \geq 4$ and $k \geq 1$, consider any $G_\sigma \in \cU_{\Delta,k}$. For any $j \in \{1,\ldots,|\mathcal{T}_{\Delta,k}|\}$, any $b \in \{1,2\}$, and any node $v \neq r_{j,b}$, we have that $\cB^{k}(r_{j,b}) \neq \cB^{k}(v)$ in $G_\sigma$.
\end{lemma}

\begin{proof}
	By construction of the template graph $U$ when $\Delta \geq 4$, the nodes with degree exactly $\Delta+2$ are precisely the cycle nodes (all other nodes either have degree at most $\Delta$, or degree exactly $2\Delta-1$). The construction of each $G_\sigma \in \cU_{\Delta,k}$ does not change the node degrees, so the same is true for each $G_\sigma$. 
	
	Fix an arbitrary cycle node $r_{j,b} \in G_\sigma$. From the above observation, for any node $v \in G_\sigma$ that is not a cycle node, the nodes $r_{j,b}$ and $v$ have different degree, so $\cB^{k}(r_{j,b}) \neq \cB^{k}(v)$ in $G_\sigma$.
	
	The remaining case to consider is when $v = r_{j',b'}$ with $j' \neq j$ or $b' \neq b$.
	\begin{itemize}
		\item Suppose that $b' \neq b$. Without loss of generality, assume that $b=1$ and $b'=2$. By construction (Building Block 3 of Section \ref{constructGDeltak}), the tree $T_{j,b}$ has an appended path of length $k+1$ starting at its root $r_{j,b}$, and the port sequence $\pi$ along this path in $\cB^{k}(r_{j,b})$ has its first port equal to 0 and its final port equal to 1. However, by construction (Building Block 3 of Section \ref{constructGDeltak}), the tree $T_{j,b'}$ has an appended path starting at its root $r_{j,b'} = v$, and the port sequence along this path in $\cB^{k}(v)$ has first port equal to 0 and its final port equal to 0. Since no other port at $v$ is labeled 0, it follows that the port sequence $\pi$ does not exist in $\cB^{k}(v)$ starting at $v$. This proves that $\cB^{k}(v) \neq \cB^{k}(r_{j,b})$.
		\item Suppose that $j' \neq j$.  It follows from our construction (Building Blocks 2 and 3 of Section \ref{constructGDeltak}) that $T_{j,b}$ is built using a tree $T_j = T_{X}$ , and $T_{j',b'}$ is built using a tree $T_{j'} = T_{Y}$, where $X=(x_1,\ldots,x_z)$ and $Y=(y_1,\ldots,y_z)$ are distinct sequences of integers in the range $1,\ldots,\Delta-1$. Recall that $T_{X}$ has height $k+1$, has $z$ nodes at distance $k$ from its root $r_{j,b}$, and the $i^{th}$ node at distance $k$ from the root has $x_i$ degree-1 neighbours. (The ordering of nodes at distance $k$ is based on the lexicographic ordering of the sequence of ports leading from the root to each such node.) Similarly, $T_{Y}$ has height $k+1$, has $z$ nodes at distance $k$ from its root $r_{j',b'}$, and the $i^{th}$ node at distance $k$ from the root has $y_i$ degree-1 neighbours. As $X$ and $Y$ are distinct, there exists an index $i$ such that $x_i \neq y_i$. Therefore, in the augmented view $\cB^{k}(r_{j,b})$, there is a path with port sequence $\pi$ starting at $r$ whose other endpoint is labeled with $x_i$, and, in the augmented view $\cB^{k}(r_{j',b'})$, the path with the same port sequence $\pi$ starting at $r$ exists, but the other endpoint is labeled with $y_i \neq x_i$. This proves that $\cB^{k}(r_{j,b}) \neq \cB^{k}(r_{j',b'})$, i.e., $\cB^{k}(r_{j,b}) \neq \cB^{k}(v)$.
	\end{itemize}

\end{proof}

Lemma \ref{kviewsGS} implies that there is a unique cycle node with the smallest lexicographic view, and, in what follows, we will denote this node by $r_{min}$.


%

To prove that the $\pe$-index is at most $k$ for any graph $G_\sigma \in \cU_{\Delta,k}$, we provide a distributed algorithm that solves Port Election using $k$ rounds of communication, assuming that each node has a complete map of the graph $G_\sigma$. At a high level, the algorithm partitions the nodes of $G_\sigma$ into three types: `heavy' nodes with degree $2\Delta-1$, `medium' nodes with degree $\Delta+2$, and `light' nodes with degree less than $\Delta+2$. The medium nodes are precisely the cycle nodes, so they follow a strategy similar to the one described above to solve $\select$: each cycle node has a unique truncated view up to distance $k$, and so given a complete map of the graph, each cycle node can find $r_{min}$ and compare $\cB^k(r_{min})$ with its own view up to distance $k$. The unique cycle node $r_{min}$ whose view is a match will output `leader', and all other cycle nodes output the port $\Delta+1$, which is the first port on a simple path around the cycle towards $r_{min}$. The heavy nodes are precisely the $r_{j,1,1}$ and $r_{j,1,2}$ nodes for each $j \in \{1,\ldots,|\mathcal{T}_{\Delta,k}|\}$. After $k$ rounds of communication, their view contains the entire $T_{j,1,1}$ or $T_{j,1,2}$ of which they are the root. Given a complete map of the network, they can locate $T_{j,1,1}$ and $T_{j,1,2}$ and conclude that they are root of one of these two trees. Regardless of which tree they belong to, the same port leads towards the cycle in $G_\sigma$, which they will output.
Finally, after $k$ rounds of communication, each light node either: has degree 1, in which case it outputs 0 (its only outgoing port); or, has at least one medium node in its view, in which case it outputs the port leading towards the closest such medium node; or, otherwise, it has a heavy node in its view, in which case it outputs the port leading towards this heavy node. The above argument is formalized in the following result.

\begin{lemma}\label{PEkU}
	For any integers $\Delta \geq 4$ and $k \geq 1$, $\psi_{\pe}(G_\sigma) = \psi_{\select}(G_\sigma) = k$ for any graph $G_\sigma \in \cU_{\Delta,k}$.
\end{lemma}
\begin{proof}
 For any graph $G_\sigma \in \cU_{\Delta,k}$, we note that $\psi_{\pe}(G_\sigma) \geq \psi_{\select}(G_\sigma) \geq k$ by Fact \ref{indexHierarchy} and Corollary \ref{cor:Sgeqk}. To complete the proof of the desired result, it suffices to prove that $\psi_{\pe}(G_\sigma) \leq k$. We provide a distributed algorithm that solves Port Election using $k$ rounds of communication, assuming that each node has a complete map of the graph $G_\sigma$. First, each node $v$ uses $k$ communication rounds to obtain $\cB^k(v)$, i.e., their view up to distance $k$. Based on $v$'s own degree, it follows one of the following strategies:
	\begin{itemize}
		\item If $v$'s degree is 1: output 0.
		\item If $v$'s degree is $\Delta+2$: compute $\cB^k(v')$ for each cycle node $v'$ in the network map. Let $r_{min}$ be the lexicographically smallest such view. If $\cB^k(v) = \cB^k(r_{min})$, then output `leader', otherwise output $\Delta+1$.\label{medium}
		\item If $v$'s degree is $2\Delta-1$: compute $\cB^k(v')$ for each node $v'$ with degree $2\Delta-1$ in the network map. Let $r_{match}$ be one such node on the map that has $\cB^k(r_{match}) = \cB^k(v)$. Output the port $p$ that is the first port on a simple path from $r_{match}$ to the network's cycle.\label{heavy}
		\item In all other cases: if $\cB^k(v)$ contains a node with degree $\Delta+2$, then let $v'$ be such a node, and, otherwise, let $v'$ be a node with degree $2\Delta-1$ in $\cB^k(v)$. Output the port $p$ that is the first port on a simple path from $v$ to $v'$.\label{allother}
	\end{itemize}
We now confirm that this $k$-round algorithm solves leader election. 

First, observe that exactly one cycle node outputs `leader': indeed, according to the algorithm, Case \ref{medium} is the only one that results in a node outputting `leader'. In this case, the node's degree is $\Delta+2$, which means that it is only executed by the cycle nodes. Further, by Lemma \ref{kviewsGS}, there is a unique cycle node whose view up to distance $k$ is lexicographically smallest, so only one cycle node will output `leader'. All other cycle nodes will output $\Delta+1$, which is the first port on a simple path towards the leader (i.e., the path obtained by following port $\Delta+1$ repeatedly). 

Next, we show that all other nodes output the first port on a simple path towards a cycle node (which is sufficient since, from such a cycle node, there is a simple path on the cycle that leads to the leader). Each node $v$ that is not on the cycle in $G_\sigma$ falls into one of the following cases:
\begin{itemize}
	\item Suppose that $v$ has degree equal to 1. As $v$ only has one incident edge, this edge is on the simple path from $v$ to the closest cycle node. By the construction of $G_\sigma$, $v$'s port on this edge is labeled 0, and, according to the algorithm, $v$'s output is 0.
	\item Suppose that, for some $j \in \{1,\ldots,|\mathcal{T}_{\Delta,k}|\}$ and $b \in \{1,2\}$, node $v$ is contained in $T_{j,b} - \{r_{j,b}\}$. We may assume that $v$ does not have degree equal to 1, as this is handled in the previous case. By the construction of $T_{j,b}$, all nodes other than $r_{j,b}$ have degree at most $\Delta$, so $v$ follows Case \ref{allother} of the algorithm's strategy. Further, as $T_{j,b}$ has height $k+1$, all non-leaf nodes are within distance $k$ of the root $r_{j,b}$. It follows that $v$'s view up to distance $k$ contains $r_{j,b}$, i.e., a node with degree $\Delta+2$. According to the algorithm, $v$ outputs the first port on a simple path towards $r_{j,b}$, i.e., towards a cycle node.
	\item Suppose that $v = r_{j,1,1}$ or $v=r_{j,1,2}$ for some $j \in \{1,\ldots,|\mathcal{T}_{\Delta,k}|\}$. By construction, $v$ has degree $2\Delta-1$, so $v$ follows Case \ref{heavy} of the algorithm's strategy. According to this strategy, $v$ finds a node $r_{match}$ in the network map such that $\cB^k(r_{match}) = \cB^k(v)$. Claim \ref{rmatch} (proven below) implies that $r_{match}$ is equal to $r_{j',1,1}$ or $r_{j',1,2}$ for exactly one $j' \in \{1,\ldots,|\mathcal{T}_{\Delta,k}|\}$, which implies that $j' = j$. In particular, $v$ has narrowed its own location within the map down to two possibilities, one of which is its actual location within the network. By the construction of $G_\sigma$ using the template graph $U$, the same port is swapped at both $r_{j,1,1}$ and $r_{j,1,2}$ for any fixed $j \in \{1,\ldots,|\mathcal{T}_{\Delta,k}|\}$. This implies that the same port $p$ at $r_{j,1,1}$ and $r_{j,1,2}$ is the first port on a simple path leading towards the closest cycle node, which is the algorithm's output.
	
	\begin{claim}\label{rmatch}
		 For any $j \in \{1,\ldots,|\mathcal{T}_{\Delta,k}|\}$, we have that $\cB^{k}(r_{j,1,1}) = \cB^{k}(r_{j,1,2})$, and, for any node $v \not\in \{r_{j,1,1},r_{j,1,2}\}$, we have that $\cB^{k}(v) \neq \cB^{k}(r_{j,1,1})$.
	\end{claim}
		To prove the claim, fix an arbitrary $j \in \{1,\ldots,|\mathcal{T}_{\Delta,k}|\}$. 
		
		Proposition \ref{prop:kHijviewsGS} implies that $\cB^{k}(r_{j,1,1}) = \cB^{k}(r_{j,1,2})$. 
		
		Next, consider any $v \not\in \{r_{j,1,1},r_{j,1,2}\}$. If $v$ does not have degree exactly $2\Delta-1$, then $\cB^{k}(v) \neq \cB^{k}(r_{j,1,1})$ since $r_{j,1,1}$ has degree $2\Delta-1$. 
		
		The remaining case to consider is when $v \in \{r_{j',1,1},r_{j',1,2}\}$ for some $j' \in \{1,\ldots,|\mathcal{T}_{\Delta,k}|\} - \{j\}$. Without loss of generality, assume that $v = r_{j',1,1}$ (since the case $v=r_{j',1,2}$ follows from the fact that $\cB^{k}(r_{j',1,1}) = \cB^{k}(r_{j',1,2})$). 
		It follows from our construction of $G_\sigma$ that $r_{j,1,1}$ is the root of $T_{j,1,1}$ and that $r_{j',1,1}$ is the root of $T_{j',1,1}$. Further, $T_{j,1,1}$ is a copy of $T_{j,1}$, and $T_{j',1,1}$ is a copy of $T_{j',1}$. By the construction of $T_{j,1}$ and $T_{j',1}$ (Building Blocks 2 and 3 of Section \ref{constructGDeltak}), it follows that $T_{j,1}$ is built using a tree $T_j = T_{X}$ , and $T_{j',1}$ is built using a tree $T_{j'} = T_{Y}$, where $X=(x_1,\ldots,x_z)$ and $Y=(y_1,\ldots,y_z)$ are distinct sequences of integers in the range $1,\ldots,\Delta-1$. Recall that $T_{X}$ has height $k+1$, has $z$ nodes at distance $k$ from its root, and the $i^{th}$ node at distance $k$ from the root has $x_i$ degree-1 neighbours. (The ordering of nodes at distance $k$ is based on the lexicographic ordering of the sequence of ports leading from the root to each such node.) Similarly, $T_{Y}$ has height $k+1$, has $z$ nodes at distance $k$ from its root, and the $i^{th}$ node at distance $k$ from the root has $y_i$ degree-1 neighbours. As $X$ and $Y$ are distinct, there exists an index $i$ such that $x_i \neq y_i$. Therefore, in the augmented view $\cB^{k}(r_{j,1,1})$, there is a path with port sequence $\pi$ starting at $r$ whose other endpoint is labeled with $x_i$, and, in the augmented view $\cB^{k}(r_{j',1,1})$, the path with the same port sequence $\pi$ starting at $r$ exists, but the other endpoint is labeled with $y_i \neq x_i$. This proves that $\cB^{k}(r_{j,1,1}) \neq \cB^{k}(r_{j',1,1})$, i.e., $\cB^{k}(r_{j,1,1}) \neq \cB^{k}(v)$, and concludes the proof of the claim.

	\item Suppose that, for some $j \in \{1,\ldots,|\mathcal{T}_{\Delta,k}|\}$ and some $b \in \{1,2\}$, node $v$ is contained in $T_{j,1,b} - \{r_{j,1,b}\}$. We may assume that $v$ does not have degree equal to 1, as this is handled in a previous case. By the construction of $U$, the tree $T_{j,1,b}$ is simply a copy of the tree $T_{j,1}$. So, by construction, all nodes other than $r_{j,1,b}$ in $T_{j,1,b}$ have degree at most $\Delta$, therefore $v$ follows Case \ref{allother} of the algorithm's strategy. Further, as $T_{j,1,b}$ has height $k+1$, all non-leaf nodes are within distance $k$ of the root $r_{j,1,b}$. It follows that $v$'s view up to distance $k$ contains $r_{j,1,b}$, i.e., a node with degree $2\Delta-1$. Note that only the cycle nodes in $G_\sigma$ have degree $\Delta+2$, and $v$'s distance to the cycle is greater than $k$, so $v$'s view up to distance $k$ does not contain a node with degree $\Delta+2$. Thus, according to the algorithm, $v$ outputs the first port on a simple path towards $r_{j,1,b}$. Since $r_{j,1,b}$ is on a simple path between $v$ and $r_{j,b}$, it follows that $v$'s output is the first port on a simple path towards a cycle node.
	\item Suppose that $v$ does not belong to any of the previous cases. From the construction of the template graph $U$, it follows that $v$ is an internal node on an induced path $P$ of length $k+1$ with one endpoint of $P$ equal to $r_{j,1,b}$ for some $j \in \{1,\ldots,|\mathcal{T}_{\Delta,k}|\}$ and $b \in \{1,2\}$. If the other endpoint of $P$ is $r_{j,b}$, then it follows that $v$ is within distance $k$ from $r_{j,b}$, which is a node with degree $\Delta+2$, so it outputs the first port on a simple path towards the cycle node $r_{j,b}$. Otherwise, the other endpoint of $P$ has degree 1, so there is no node within distance $k$ from $v$ that has degree $\Delta+2$. However, $v$ is within distance $k$ from $r_{j,1,b}$, which has degree $2\Delta-1$, so $v$ outputs the first port on a simple path towards $r_{j,1,b}$. Since $r_{j,1,b}$ is on a simple path from $v$ to $r_{j,b}$, it follows that $v$ outputs the first port on a simple path towards the cycle node $r_{j,b}$.
\end{itemize}
\end{proof}
%

We now proceed to analyze the amount of advice needed to solve Port Election in the graph class $\cU_{\Delta,k}$. First, we observe that any algorithm that solves $\pe$ for the graphs in $\cU_{\Delta,k}$ must elect one of the cycle nodes as leader. The proof of this fact follows directly from Propositions \ref{prop:kTijviewsGS} and \ref{prop:kHijviewsGS}, as these two results show that any node that does not belong to the network's cycle has at least one `twin' in the network that has the same view up to distance $k$. hence we have the following lemma.

\begin{lemma}\label{mustBeRoot}
	Consider any algorithm $\cA$ that, for any graph $G_\sigma \in \cU_{\Delta,k}$, solves $\pe$ in $\psi_{\pe}(G_\sigma)$ rounds. At the end of the execution of $\cA$ on any fixed graph $G_\sigma \in \cU_{\Delta,k}$, the node in $G_\sigma$ that outputs `leader' is the root node $r_{j,b}$ of $T_{j,b}$ for some $j \in \{1,\ldots,|\mathcal{T}_{\Delta,k}|\}$ and some $b \in \{1,2\}$.
\end{lemma}

The main theorem of this section shows that the size of advice needed to solve $PE$ in the class $ \cU_{\Delta,k}$ in minimum time is exponential in $\Delta$, and thus exponentially larger than the size of advice needed to solve $\select$ in minimum time in this class.

\begin{theorem}
	Consider any algorithm $\cA$ that solves $\pe$ in $\psi_{\pe}(G)$ rounds for every graph $G$. For all integers $\Delta \ge 4, k\ge 1$, there exists a
	graph $G$ with maximum degree $2\Delta-1$ and with $\psi_{\pe}(G)=\psi_\select(G)=k$ for which algorithm $\cA$ requires
	advice of size $\Omega((\Delta-1)^{(\Delta-2)(\Delta-1)^{k-1}}\log\Delta)$.
\end{theorem}
\begin{proof}
	To obtain a contradiction, assume that there exists an algorithm $\cA$ that solves $\pe$ in $k$ rounds for all graphs in the class $\cU_{\Delta,k}$ with the help of an oracle that provides advice of size $\frac{1}{4}|\mathcal{T}_{\Delta,k}|\log_2\Delta$. There are at most $2^{1+(\frac{1}{4}|\mathcal{T}_{\Delta,k}| \log_2 \Delta)}$ binary advice strings whose length is at most $\frac{1}{4}|\mathcal{T}_{\Delta,k}| \log_2 \Delta$. By Fact \ref{classCountG}, along with the assumptions that $\Delta \geq 4$ and $k \geq 1$, we see that $\frac{1}{4}|\mathcal{T}_{\Delta,k}|\log_2\Delta = \frac{1}{4}(\Delta-1)^{(\Delta-2)\cdot (\Delta-1)^{k-1}}\cdot\log_2\Delta \geq 1$, so the number of binary advice strings whose length is at most $\frac{1}{4}|\mathcal{T}_{\Delta,k}| \log_2 \Delta$ is at most $2^{1+(\frac{1}{4}|\mathcal{T}_{\Delta,k}| \log_2 \Delta)} \leq 2^{\frac{1}{2}|\mathcal{T}_{\Delta,k}| \log_2 \Delta} = \Delta^{\frac{1}{2}|\mathcal{T}_{\Delta,k}|}$. But, by Fact \ref{classCountU}, the total number of graphs in $\cU_{\Delta,k}$ is $(\Delta-1)^{|\mathcal{T}_{\Delta,k}|} > \Delta^{\frac{1}{2}|\mathcal{T}_{\Delta,k}|}$, so, by the Pigeonhole Principle, the oracle provides the same advice for at least two graphs $G_\alpha$ and $G_\beta$ from $\cU_{\Delta,k}$, where $\alpha$ and $\beta$ are distinct sequences of integers of length $|\mathcal{T}_{\Delta,k}|$.
	
	By Lemma \ref{mustBeRoot}, at the end of any execution of $\cA$ on any graph $G_\sigma \in \cU_{\Delta,k}$, the elected leader is a cycle node. 
	It follows that, in both $G_\alpha$ and $G_\beta$, all nodes $r_{j,1,1}$ and $r_{j,1,2}$ for $j \in \{1,\ldots,|\mathcal{T}_{\Delta,k}|\}$ output the unique integer that labels the first port along the simple path towards their corresponding $r_{j,1}$ and $r_{j,2}$ on the cycle. Since $\alpha \neq \beta$, there exists an $j' \in \{1,\ldots,|\mathcal{T}_{\Delta,k}|\}$ such that $\alpha_{j'} \neq \beta_{j'}$. Recall from our construction that $G_\alpha$ was obtained from the template graph $U$ by swapping the ports label $\Delta-1$ and $\Delta-1+\alpha_{j'}$ at nodes $r_{j',1,1}$ and $r_{j',1,2}$, and $G_\beta$ was obtained from the template graph $U$ by swapping the port labels $\Delta-1$ and $\Delta-1+\beta_{j'}$ at nodes $r_{j',1,1}$ and $r_{j',1,2}$. In particular, this means that, at $r_{j',1,1}$, the first port along the simple path towards its corresponding $r_{j',1}$ is \emph{different} in $G_{\alpha}$ than in $G_{\beta}$. However, for any fixed $j \in \{1,\ldots,|\mathcal{T}_{\Delta,k}|\}$, the node $r_{j,1,1}$ has the same view up to distance $k$ in both graphs $G_{\alpha}$ and $G_{\beta}$ (which is the same as its view in the template graph, i.e., the subtree $T_{j,1,1}$, along with $\Delta$ induced paths of length $k$, each with a different incident port at $r_{j,1,1}$ from the range $\Delta-1,\ldots,2\Delta-2$). Since each node receives the same advice string in the execution of $\mathcal{A}$ in both graphs $G_\alpha$ and $G_\beta$, the node $r_{j',1,1}$ will output the same port in both executions (due to indistinguishability), which is a contradiction.
\end{proof}


\section{Port Path Election and Complete Port Path Election}\label{secPPEandCPPE}

In this section, we prove that the size of advice needed to solve $\cppe$ (or $\ppe$) in minimum time is exponentially larger than the size of advice needed to solve $\select$. More specifically, for sufficiently large $k$ and $\Delta$, we construct a class of graphs such that: $\psi_{\select}(G) = \psi_{\ppe}(G) = \psi_{\cppe}(G) = k$ for each graph in the class, solving $\select$ in time $k$ in this class can be done with advice of size at most $O((\Delta-1)^k\log \Delta)$ (in view of Theorem \ref{S-ub}), but there exists a graph in the class for which the size of advice needed to solve $\cppe$ (or $\ppe$) in time $k$ is at least $\Omega(2^{\Delta^{k/6}})$.

\subsection{Construction of $\mathcal{J}_{\mu,k}$}
Consider any two integers $\mu \geq 2$ and $k \geq 4$. Denote by $T^h$ a port-labeled full $\mu$-ary tree of height $h$. In particular, the root has degree $\mu$ with ports labeled $0,\ldots,\mu-1$, each internal node has degree $\mu+1$ with port $\mu$ leading to its parent and ports $0,\ldots,\mu-1$ leading to its children, and each leaf has port 0 leading to its parent.

\subparagraph*{Part 1: Construct layer graphs.} We begin by constructing a collection of graphs $L_0, \ldots, L_k$, called \emph{layer graphs}, where the graph $L_j$ in this set has diameter $j$. Define $L_0$ to be a single node called $r_0^0$. Define $L_1$ to be a clique with $\mu$ nodes (and any port labeling using $0,\ldots,\mu-2$). For each $j \geq 1$, construct $L_{2j}$ as follows:
\begin{enumerate}
	\item Take two copies of $T^j$, denoted by $T^j_0$ and $T^j_1$ with roots $r^{2j}_0$ and $r^{2j}_1$, respectively. 
	\item For each leaf $\ell_0 \in T^j_0$, identify $\ell_0$ with the leaf $\ell_1 \in T^j_1$ such that the sequence of ports on the path from $r^{2j}_1$ to $\ell_1$ is the same as the sequence of ports on the path from $r^{2j}_0$ to $\ell_0$. The nodes corresponding to these identified leaves will be called the \emph{middle nodes of} $L_{2j}$.
	\item For each middle node of $L_{2j}$, label the port on the incident edge from $T^j_0$ with 0, and label the port on the incident edge from $T^j_1$ with 1.
\end{enumerate}
For each $j \geq 1$, construct $L_{2j+1}$ as follows:
\begin{enumerate}
	\item Take two copies of $T^j$, denoted by $T^j_0$ and $T^j_1$ with roots $r^{2j+1}_0$ and $r^{2j+1}_1$, respectively. 
	\item For each leaf $\ell_0 \in T^j_0$, add an edge between $\ell_0$ and the leaf $\ell_1 \in T^j_1$ such that the sequence of ports on the path from $r^{2j+1}_1$ to $\ell_1$ is the same as the sequence of ports on the path from $r^{2j+1}_0$ to $\ell_0$. The leaves of $T^j_0$ and $T^j_1$ will be called the \emph{middle nodes of} $L_{2j+1}$.
	\item  For each edge connecting two middle nodes, label both ports on the edge with 1.
\end{enumerate}
Figure \ref{figLayers} provides examples of the resulting layer graphs. From the description of the construction, we obtain the following result about the size of each layer graph.

\begin{fact}\label{fact:layersizes}
	The number of nodes in $L_0$ is 1, and the number of nodes in $L_1$ is $\mu$. For each $j \geq 1$, the number of nodes in $L_{2j}$ is $\frac{\mu^{j+1} + \mu^j - 2}{\mu-1}$, and the number of nodes in $L_{2j+1}$ is $\frac{2\mu^{j+1}-2}{\mu-1}$.
\end{fact}

\captionsetup[subfigure]{justification=centering, labelformat=empty}

\begin{figure}
	\subcaptionbox{$L_0$}[0.2\linewidth]{\includegraphics[scale=1]{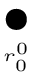}}
	\subcaptionbox{$L_1$}[0.2\linewidth]{\includegraphics[scale=0.5]{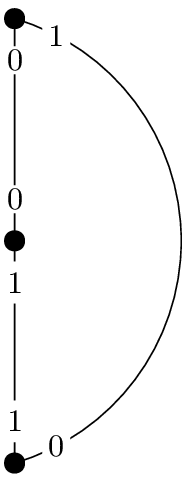}}
	\subcaptionbox{$L_2$}[0.25\linewidth]{\includegraphics[scale=0.5]{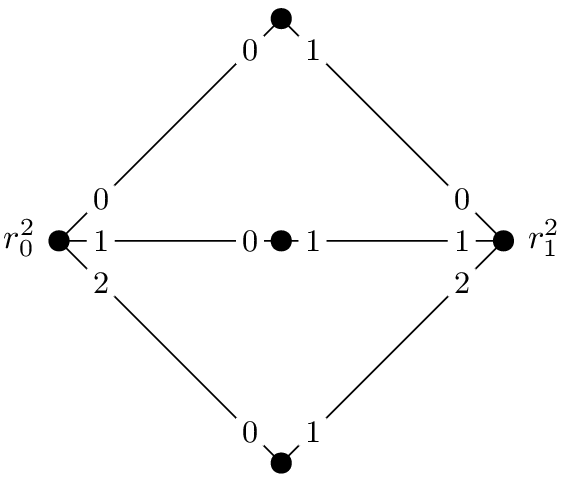}}
	\subcaptionbox{$L_3$}[0.25\linewidth]{\includegraphics[scale=0.5]{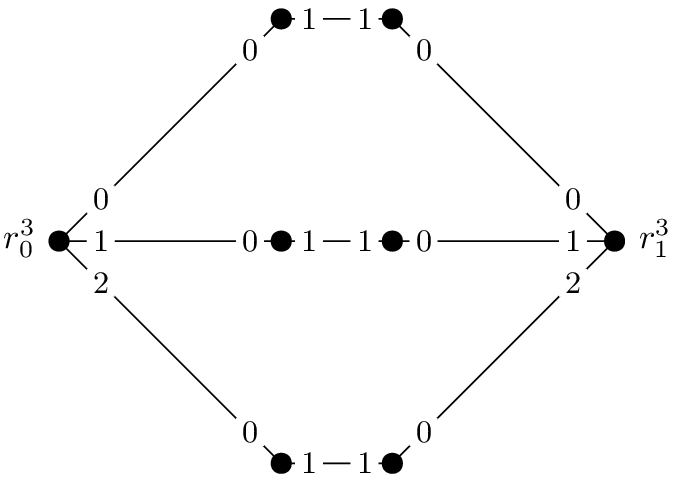}}\vspace{2mm}\\
	\subcaptionbox{$L_4$}[0.5\linewidth]{\includegraphics[scale=0.4]{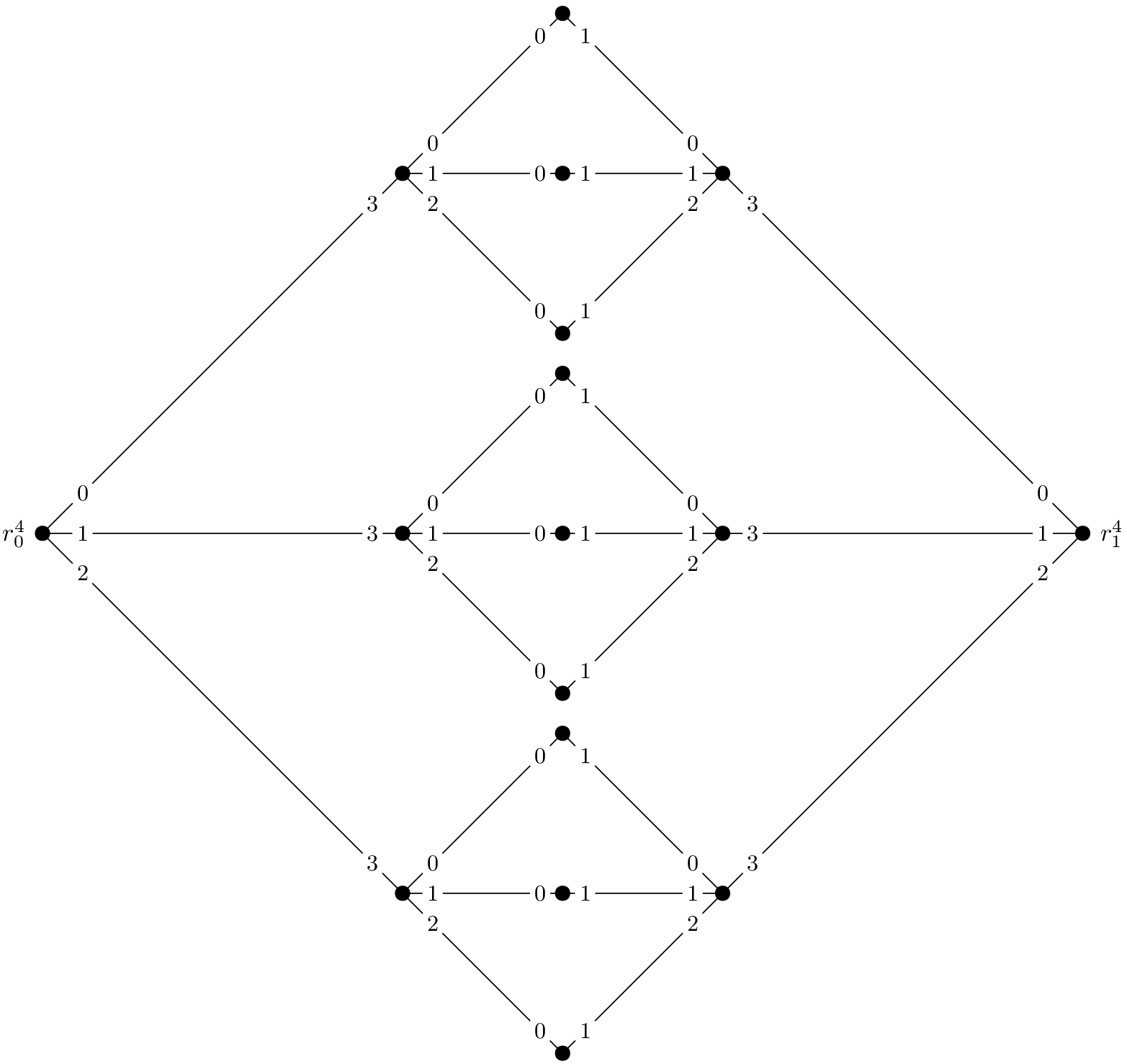}}
	\subcaptionbox{$L_5$}[0.5\linewidth]{\includegraphics[scale=0.4]{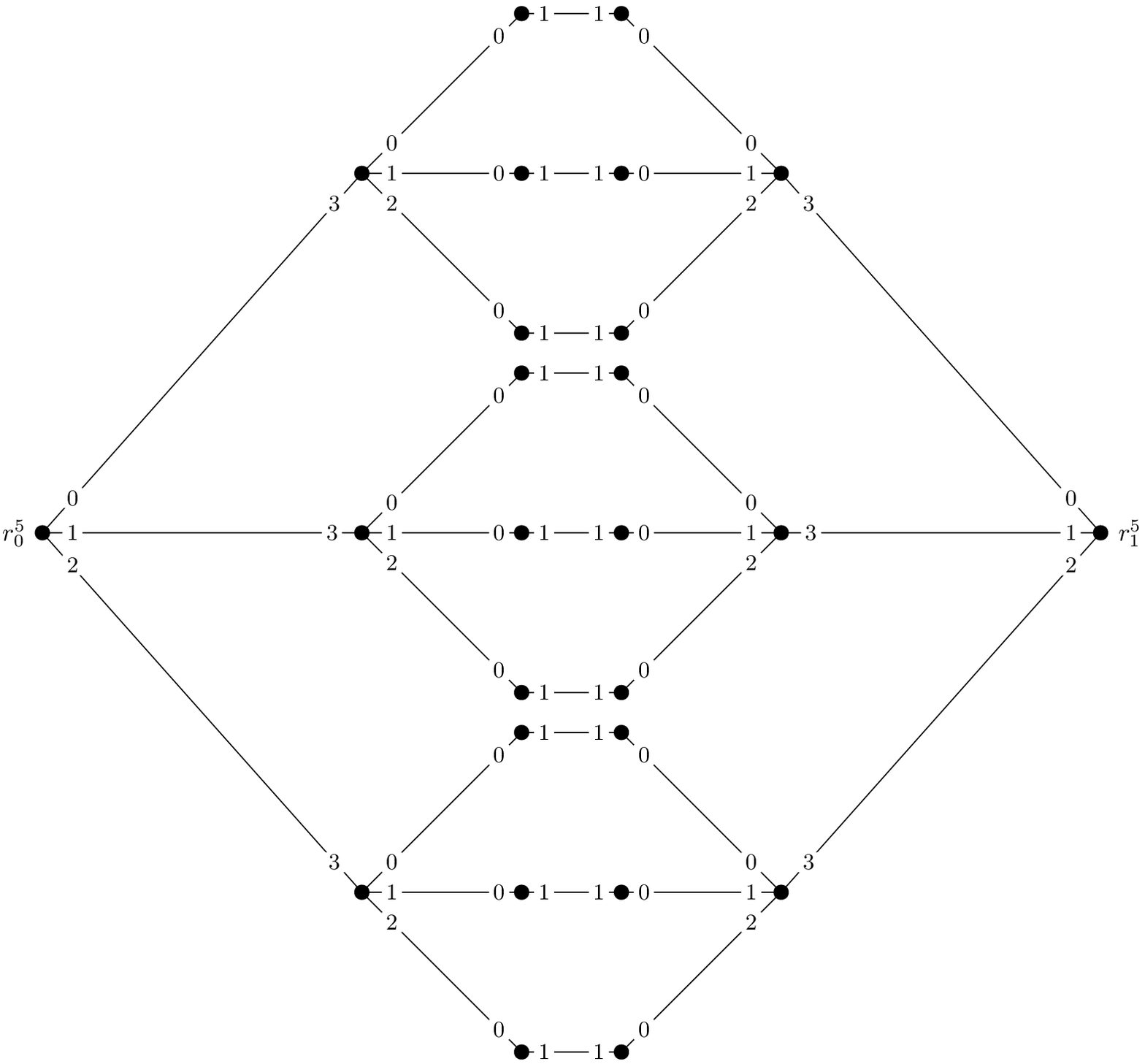}}
		\caption{Layer graphs $L_0,L_1,L_2,L_3,L_4,L_5$ when $\mu = 3$}
		\label{figLayers}
\end{figure}

In what follows, we will sometimes refer to a node in a layer graph using the outgoing port sequence that can be used to reach it starting from a node $r_0^m$ or $r_1^m$ for some $m \geq 0$. In particular, for any given integer sequence $\sigma$, the notation $v^m_{0 \leadsto \sigma}$ will denote the node in $L_m$ that can be reached starting from $r_0^m$ using the integers in $\sigma$ to trace a path of outgoing ports. The notation $v^m_{1 \leadsto \sigma}$ will be used analogously but starting from the node $r_1^m$. If $|\sigma| = 0$, then $v^m_{0 \leadsto \sigma} = r^m_0$ and $v^m_{1 \leadsto \sigma} = r^m_1$.

\subparagraph*{Part 2: Join layer graphs together to create component graphs.} Next, we construct a \emph{component graph} $H$ by starting with the disjoint union of the layer graphs $L_0, L_1, \ldots, L_{k-1}$ along with two copies of $L_k$ (denoted by $L_{k,1}$ and $L_{k,2}$) and then adding edges between consecutive layer graphs in the following way:
\begin{itemize}
	\item {\bf Edges between $L_0$ and $L_1$.} For each node $v \in L_1$, add an edge $\{r^0_0,v\}$. Label the ports at $r^0_0$ using $0,\ldots,\mu-1$, and label the newly-created port at each node in $L_1$ by $\mu-1$.
	\item {\bf Edges between $L_1$ and $L_2$.} For each $i \in 0,\ldots,\mu-1$, add an edge between $v^0_{0 \leadsto (i)}$ and $v^2_{0 \leadsto (i)}$. At each of these added edges, label the port at $v^0_{0 \leadsto (i)}$ with $\mu$, and label the port at $v^2_{0 \leadsto (i)}$ with 2. Next, add an edge connecting $v^0_{0 \leadsto (0)}$ to $r^2_0$. On this edge, label the port at $v^0_{0 \leadsto (0)}$ with $\mu+1$, and label the port at $r^2_0$ with $\mu$. Similarly, add an edge connecting $v^0_{0 \leadsto (\mu-1)}$ to $r^2_1$. On this edge, label the port at $v^0_{0 \leadsto (\mu-1)}$ with $\mu+1$, and label the port at $r^2_1$ with $\mu$.
	
%
%
	\item {\bf Edges between $L_m$ and $L_{m+1}$ when $2 \leq m < k-1$.} 
		\begin{itemize}
			\item Add an edge between $r^m_0$ and $r^{m+1}_0$, and add an edge between $r^m_1$ and $r^{m+1}_1$. Label the new ports at $r^m_0$ and $r^{m}_1$ with $\mu+1$, and label the new ports at $r^{m+1}_0$ and $r^{m+1}_1$ with $\mu$.
			\item Connect each non-middle node of $L_m$ (other than $r^m_0$ and $r^m_1$) to its corresponding non-middle node in $L_{m+1}$. Formally, for each sequence $\sigma$ consisting of integers from the range $\{0,\ldots,\mu-1\}$ such that $1 \leq |\sigma| < \left\lfloor\frac{m}{2}\right\rfloor$, add an edge between $v^m_{0 \leadsto \sigma}$ and $v^{m+1}_{0 \leadsto \sigma}$, and add an edge between $v^m_{1 \leadsto \sigma}$ and $v^{m+1}_{1 \leadsto \sigma}$. At $v^m_{0 \leadsto \sigma}$ and $v^m_{1 \leadsto \sigma}$, label the new port with $\mu+2$. At $v^{m+1}_{0 \leadsto \sigma}$ and $v^{m+1}_{1 \leadsto \sigma}$, label the new port with $\mu+1$.
			\item {\bf Case 1: $m$ is even.} Connect each middle node of $L_m$ to its two corresponding middle nodes in $L_{m+1}$. Formally, for each sequence $\sigma$ consisting of integers from the range $\{0,\ldots,\mu-1\}$ such that $|\sigma| = \frac{m}{2}$, add an edge between $v^m_{0 \leadsto \sigma}$ and $v^{m+1}_{0 \leadsto \sigma}$, label the new port at $v^m_{0 \leadsto \sigma}$ with 3 if $m=2$, or with 4 if $m > 2$, and label the new port at $v^{m+1}_{0 \leadsto \sigma}$ with 2. Also, add an edge between $v^m_{0 \leadsto \sigma}$ and $v^{m+1}_{1 \leadsto \sigma}$, label the new port at $v^m_{0 \leadsto \sigma}$ with 4 if $m=2$, or with $5$ if $m > 2$, and label the new port at $v^{m+1}_{1 \leadsto \sigma}$ with $2$.
			\item {\bf Case 2: $m$ is odd.} In what follows, we use $\doubleplus$ to denote sequence concatenation. Connect each middle node of $L_m$ to its corresponding node in $L_{m+1}$, as well as the $\mu$ middle nodes in $L_{m+1}$ adjacent to it. Formally, for each sequence $\sigma$ consisting of integers from the range $\{0,\ldots,\mu-1\}$ such that $|\sigma| = \frac{m-1}{2}$:
			\begin{itemize}
				\item Add an edge between $v^m_{0 \leadsto \sigma}$ and $v^{m+1}_{0 \leadsto \sigma}$, and add an edge between $v^m_{1 \leadsto \sigma}$ and $v^{m+1}_{1 \leadsto \sigma}$. At $v^m_{0 \leadsto \sigma}$ and $v^m_{1 \leadsto \sigma}$, label the new port with $3$. At $v^{m+1}_{0 \leadsto \sigma}$ and $v^{m+1}_{1 \leadsto \sigma}$, label the new port with $\mu+1$.
				\item For each $i \in \{0,\ldots,\mu-1\}$, add an edge between $v^m_{0 \leadsto \sigma}$ and $v^{m+1}_{0 \leadsto (\sigma \doubleplus i)}$, and add an edge between $v^m_{1 \leadsto \sigma}$ and $v^{m+1}_{1 \leadsto (\sigma \doubleplus i)}$. For each added edge, label the new port at $v^m_{0 \leadsto \sigma}$ and $v^m_{1 \leadsto \sigma}$ with $4+i$, label the new port at $v^{m+1}_{0 \leadsto (\sigma \doubleplus i)}$ with 2, and label the new port at $v^{m+1}_{1 \leadsto (\sigma \doubleplus i)}$ with 3.
			\end{itemize} 
		\end{itemize}
	
	\item {\bf Edges between $L_m$ and $L_{m+1}$ when $m = k-1$.} Add edges between $L_{k-1}$ and $L_{k,1}$ according to the $m < k-1$ case above. Then, add edges between $L_{k-1}$ and $L_{k,2}$ using a slightly modified version of the $m < k-1$ case: increase the values of port labels used at nodes in $L_{k-1}$ so that they do not conflict with the labels that were used when adding edges between $L_{k-1}$ and $L_{k,1}$.
	
\end{itemize}

	Figure \ref{fig:firstthreelayers} illustrates the edges added between layers $L_0$, $L_1$, $L_2$, and $L_3$. Figure \ref{fig:Interlayer-OddToEven} illustrates the edges added between layers $L_m$ and $L_{m+1}$ when $m$ is odd and strictly less than $k-1$. Figure \ref{fig:Interlayer-EvenToOdd} illustrates the edges added between $L_m$ and $L_{m+1}$ when $m$ is even and strictly less than $k-1$.
\begin{figure}
	\centering
	\includegraphics[width=\linewidth]{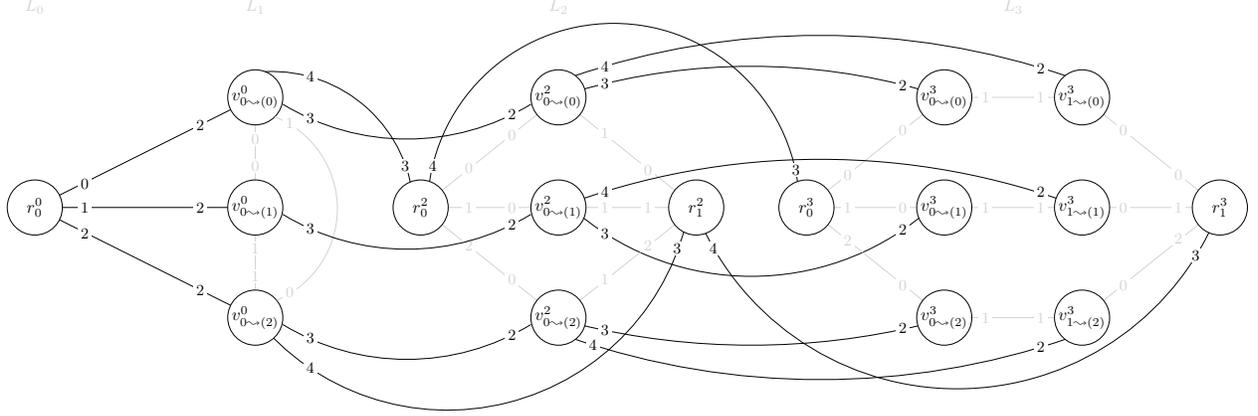}
	\caption{Subgraph of $H$ induced by layers $L_0$, $L_1$, $L_2$, and $L_3$ when $\mu = 3$. Grey edges were added in Part 1 of the construction (edges within a layer), and black edges were added in Part 2 of the construction (edges between different layers).}
	\label{fig:firstthreelayers}
\end{figure}

\begin{figure}
	\centering
	\includegraphics[width=\linewidth]{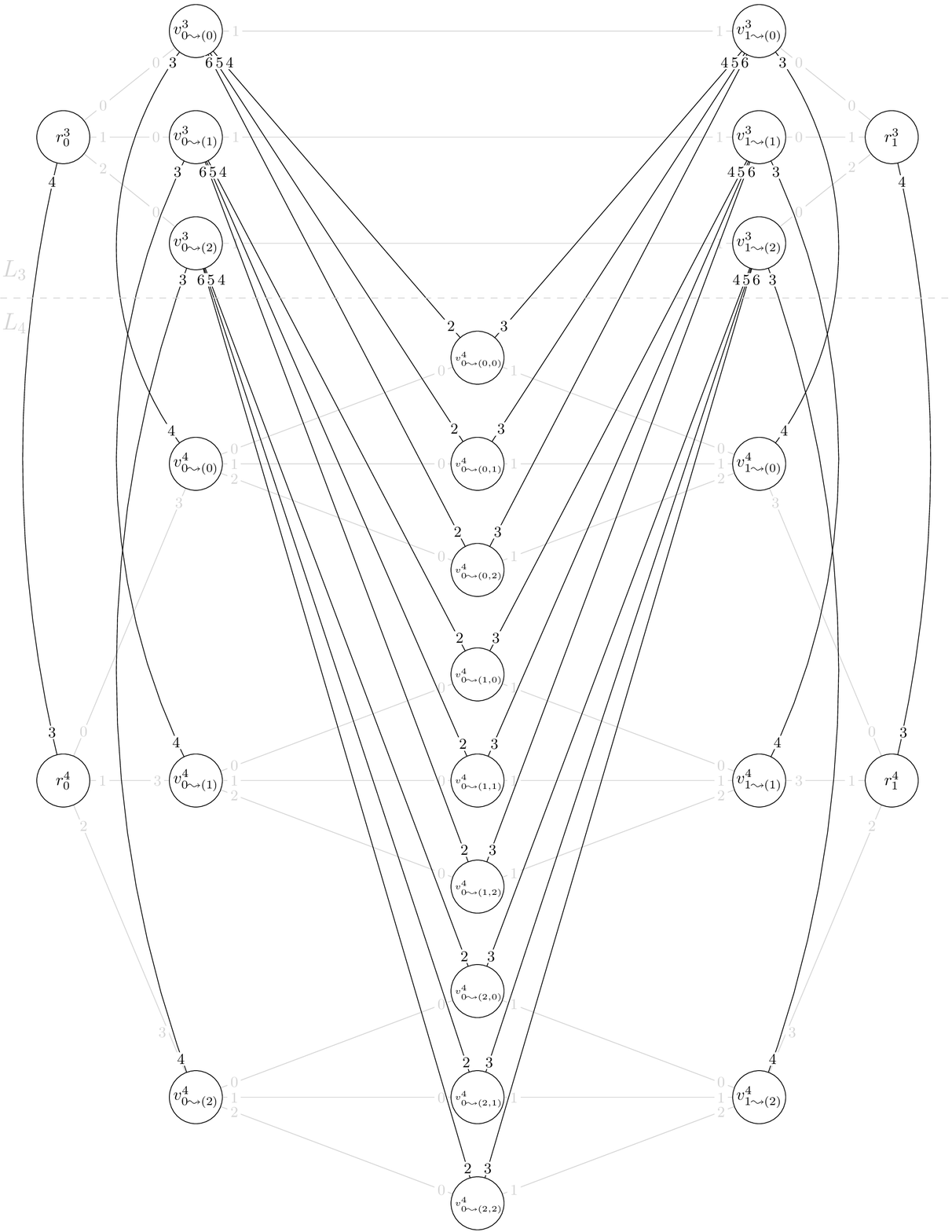}
	\caption{Subgraph of $H$ induced by layers $L_3$ and $L_4$ when $\mu = 3$ and $k > 5$. Grey edges were added in Part 1 of the construction (i.e., edges within a layer), and black edges were added in Part 2 of the construction (i.e., edges between different layers).}
	\label{fig:Interlayer-OddToEven}
\end{figure}
	
\begin{figure}
	\centering
	\includegraphics[width=\linewidth]{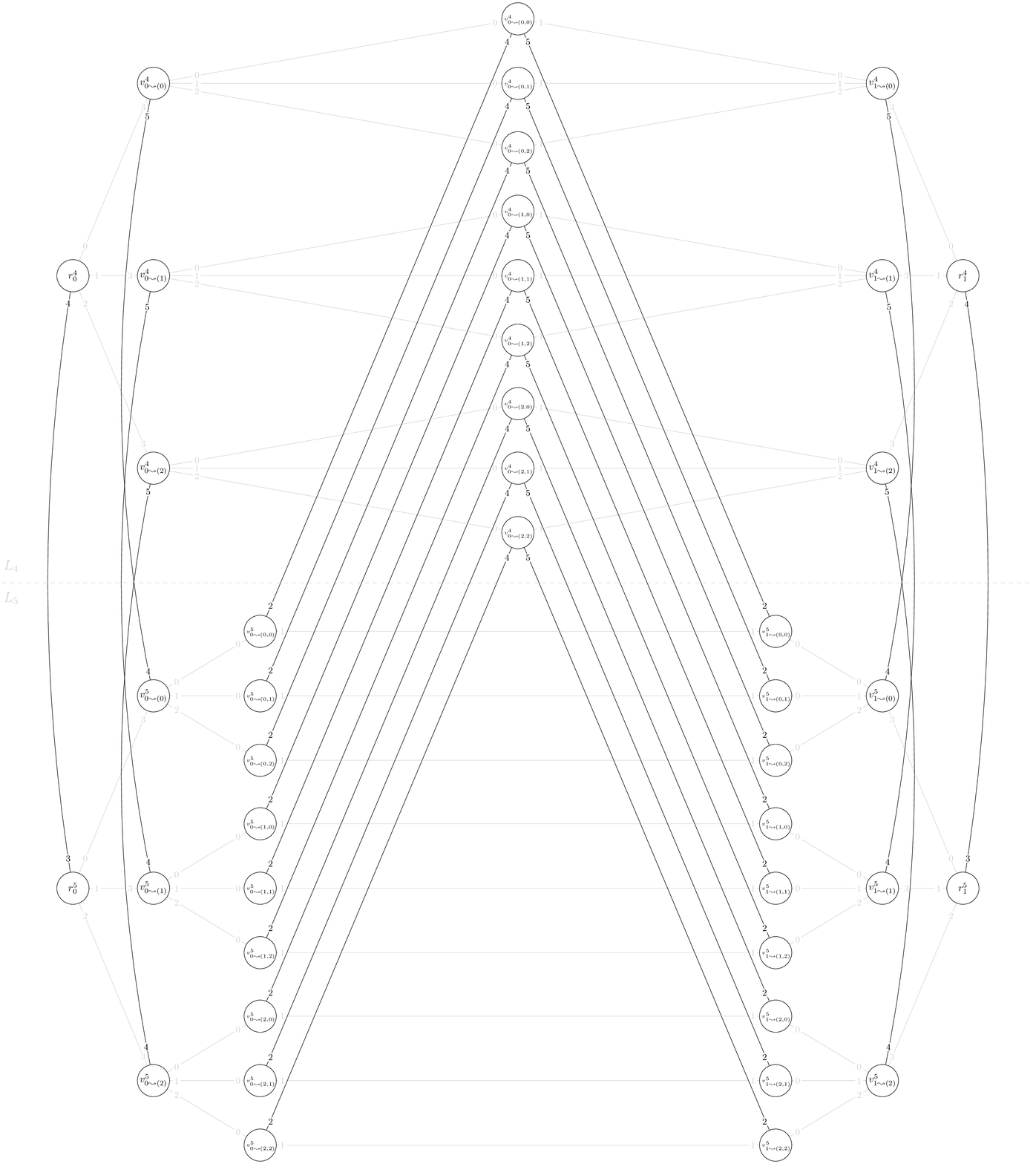}
	\caption{Subgraph of $H$ induced by layers $L_4$ and $L_5$ when $\mu = 3$ and $k > 5$. Grey edges were added in Part 1 of the construction (i.e., edges within a layer), and black edges were added in Part 2 of the construction (i.e., edges between different layers).}
	\label{fig:Interlayer-EvenToOdd}
\end{figure}

The edges we added to connect the layer graphs were chosen in such a way that, from any node $v$ in the resulting component graph $H$, all the nodes in $H$ are contained within $v$'s truncated view up to distance $k$, but, there exist some nodes in the last layer $L_k$ that are not within $v$'s truncated view up to distance $k-1$. This property will be formally proven and used later when arguing that the $\select$-index of each graph in our constructed class is at least $k$.

%

%


\subparagraph*{Part 3: Create a gadget graph using the component graph $H$.} Create four copies of the component graph $H$, which we refer to as \emph{left}, \emph{top}, \emph{right}, and \emph{bottom} component graphs, denoted by $H_L$, $H_T$, $H_R$, and $H_B$, respectively. Merge the four $r^0_0$ nodes of these component graphs together to create a new node $\rho$ that has degree $4\mu$. To avoid duplicate port labels at $\rho$ due to the merge, add $\mu$ to each port label at $\rho$ in $H_T$, add $2\mu$ to each port label at $\rho$ in $H_R$, and add $3\mu$ to each port label at $\rho$ in $H_B$. The resulting graph is called the \emph{gadget} $\widehat{H}$. Figure \ref{fig:gadgetDemo} illustrates the constructed graph $\widehat{H}$.

\begin{figure}[h!]
	\centering
	\includegraphics[scale=0.8]{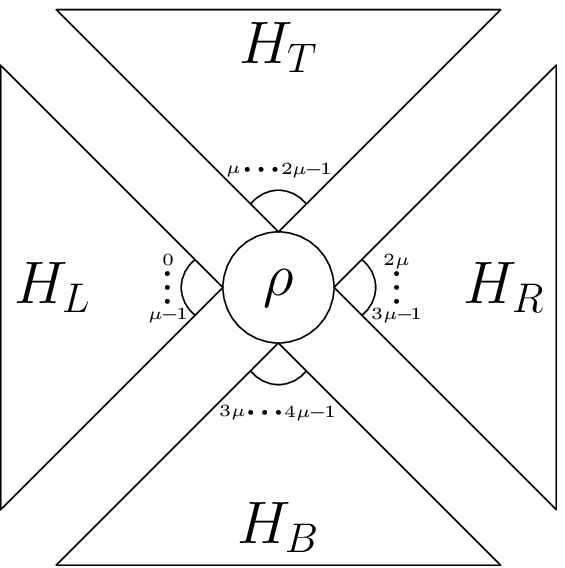}
	\caption{The gadget graph $\widehat{H}$}
	\label{fig:gadgetDemo}
\end{figure}

\subparagraph*{Part 4: Create a template graph by chaining together gadget graphs.} 

First, we induce an ordering on the vertices of the layer graph $L_{k}$. Denote by $z$ the number of nodes in $L_{k}$. As described in Part 2 of the construction above, each node in $L_{k}$ can be represented using $v^k_{b \leadsto \sigma}$ for some $b \in \{0,1\}$ and some sequence $\sigma$ of integers from the range $\{0,\ldots,\mu-1\}$. By prepending $b$ to the sequence $\sigma$, we obtain a sequence that identifies each node, and we order the resulting sequences using lexicographic order to obtain an ordered list $w_1,\ldots,w_z$ of the nodes of $L_{k}$. Recalling that each component graph $H$ contains two copies of $L_k$, i.e., $L_{k,1}$ and $L_{k,2}$, we use the notation $w_{1,1}, \ldots, w_{z,1}$ to refer to the nodes in $L_{k,1}$, and we use $w_{1,2},\ldots,w_{z,2}$ to refer to the nodes in $L_{k,2}$.

Next, we describe how to build our template graph. For each $i \in \{0,\ldots,2^z-1\}$, denote by $x_i$ the $z$-bit binary representation of $i$. Moreover, for each $i \in \{0,\ldots,2^z-1\}$, create a copy of the gadget graph $\widehat{H}$, denote it by $\widehat{H_i}$ and denote its $\rho$ node by $\rho_i$. We create a template graph $J$ by taking the disjoint union of all $\widehat{H_0},\ldots,\widehat{H_{2^z-1}}$, and adding edges as follows. For each $i \geq 1$, and for each $q \in \{1,\ldots,z\}$ such that the $q^{th}$ bit of $x_i$ is 1:
\begin{enumerate}[label=\arabic*.]
	\item Add an edge between $w_{q,1}$ and $w_{q,2}$ in component $H_B$ of gadget $\widehat{H_{i-1}}$.
	\item Add an edge between $w_{q,1}$ and $w_{q,2}$ in component $H_T$ of gadget $\widehat{H_i}$.
	\item Add an edge between $w_{q,1}$ in component $H_R$ of gadget $\widehat{H_{i-1}}$ and $w_{q,2}$ in component $H_L$ of gadget $\widehat{H_{i}}$.
	\item Add an edge between $w_{q,2}$ in component $H_R$ of gadget $\widehat{H_{i-1}}$ and $w_{q,1}$ in component $H_L$ of gadget $\widehat{H_i}$.
\end{enumerate}
Figure \ref{fig:connectGadgets} illustrates how gadgets are chained together to form the template graph $J$.

\begin{figure}[h!]
	\centering
	\includegraphics[scale=1]{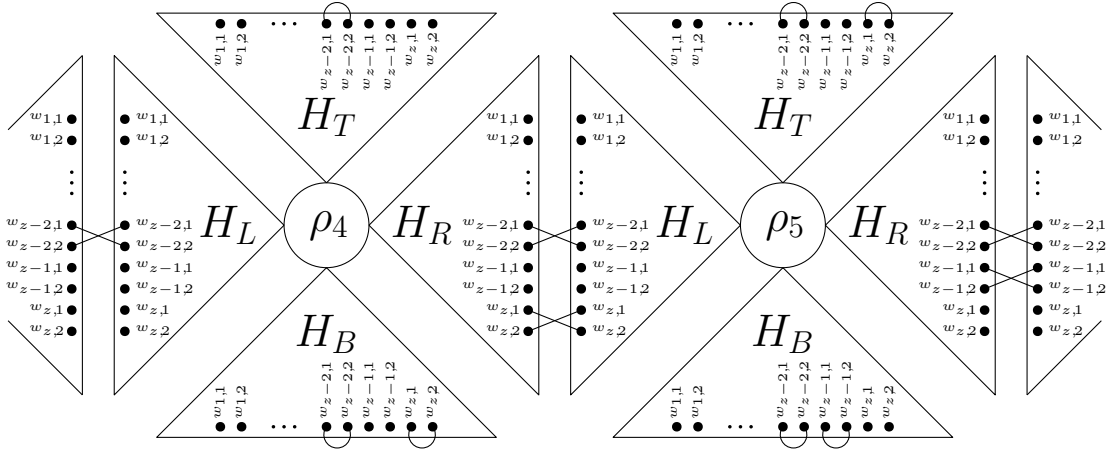}
	\caption{Example of edges added at layer-$k$ nodes within gadgets $\widehat{H_4}$ and $\widehat{H_5}$ in the construction of template graph $J$ (Part 4)}
	\label{fig:connectGadgets}
\end{figure}

For each added edge $\{u_1,u_2\}$, label the port at $u_1$ with $deg_H(u_1)$ (which is equal to $deg_{J}(u_1)-1$) and label the port at $u_2$ with $deg_H(u_2)$ (which is equal to $deg_{J}(u_2)-1$). Note that, for any fixed $q$, all port labels on all added edges are the same, since all endpoints are copies of the same node $w_q$ in the component graph $H$.

\subparagraph*{Part 5: Construct the final class of graphs by copying the template graph and swapping ports.} Consider all binary sequences with length $2^{z-1}$. For an arbitrary such sequence $Y = (y_0,\ldots,y_{2^{z-1}-1})$, construct a graph $J_Y$ by taking a copy of the template graph $J$, and, for each $i$ such that $y_i = 1$, perform the following modifications:
\begin{enumerate}[label=\arabic*.]
	\item For each $x \in \{2\mu,\ldots,3\mu-1\}$, swap ports $x$ and $x+\mu$ at node $\rho_i$\\(i.e., the ports at $H_R$ and $H_B$ of node $\rho$ of gadget $\widehat{H_i}$).
	\item For each $x \in \{0,\ldots,\mu-1\}$, swap ports $x$ and $x+\mu$ at node $\rho_{2^z-1-i}$\\(i.e., the ports at $H_L$ and $H_T$ of node $\rho$ of gadget $\widehat{H_{2^z-1-i}}$).
\end{enumerate}
Figure \ref{fig:possibleSwaps} demonstrates the three possible outcomes of performing the above to an arbitrary gadget $\widehat{H_i}$, depending on the values of $i$ and $y_i$.

\begin{figure}[h!]
	\centering
	\subcaptionbox{(a)}[0.32\linewidth]{\includegraphics[scale=0.85]{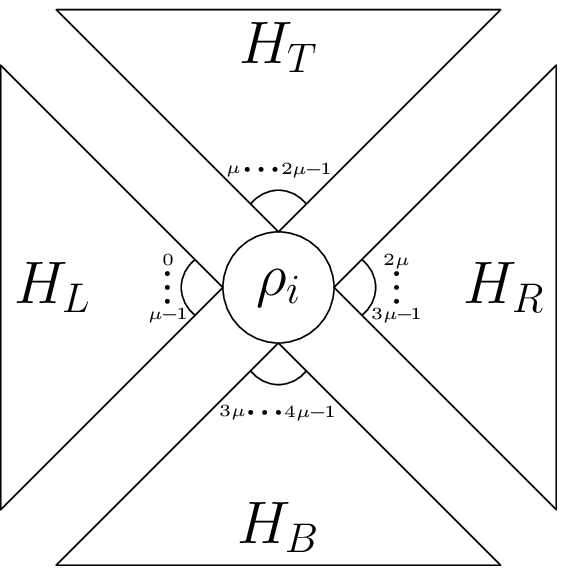}}
	\subcaptionbox{(b)}[0.32\linewidth]{\includegraphics[scale=0.85]{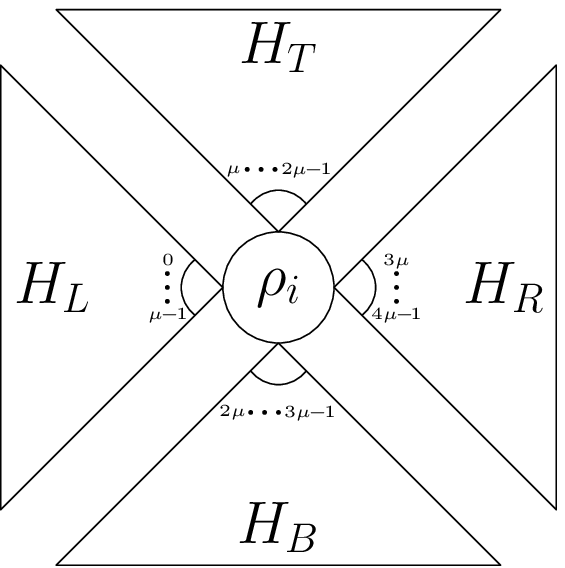}}
	\subcaptionbox{(c)}[0.32\linewidth]{\includegraphics[scale=0.85]{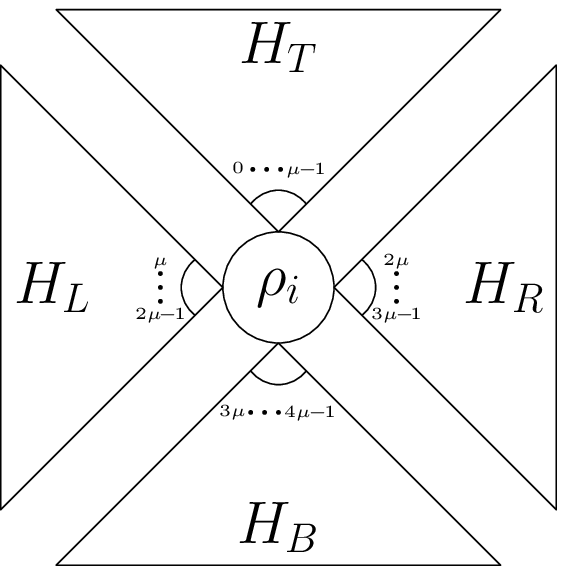}}
	\caption{For an arbitrary gadget $\widehat{H_i}$ in $J_Y$, the outcome of performing Part 5 when: (a) $y_i=0$, (b) $y_i=1$ and $i \in \{0,\ldots,2^{z-1}-1\}$, (c) $y_{2^z-1-i}=1$ and $i \in \{2^{z-1},\ldots,2^{z}-1\}$}
	\label{fig:possibleSwaps}
\end{figure}

Figure \ref{fig:JY-example-10} gives an example of a fully constructed graph $J_Y$.

\begin{figure}[h!]
	\centering
	\includegraphics[width=\linewidth]{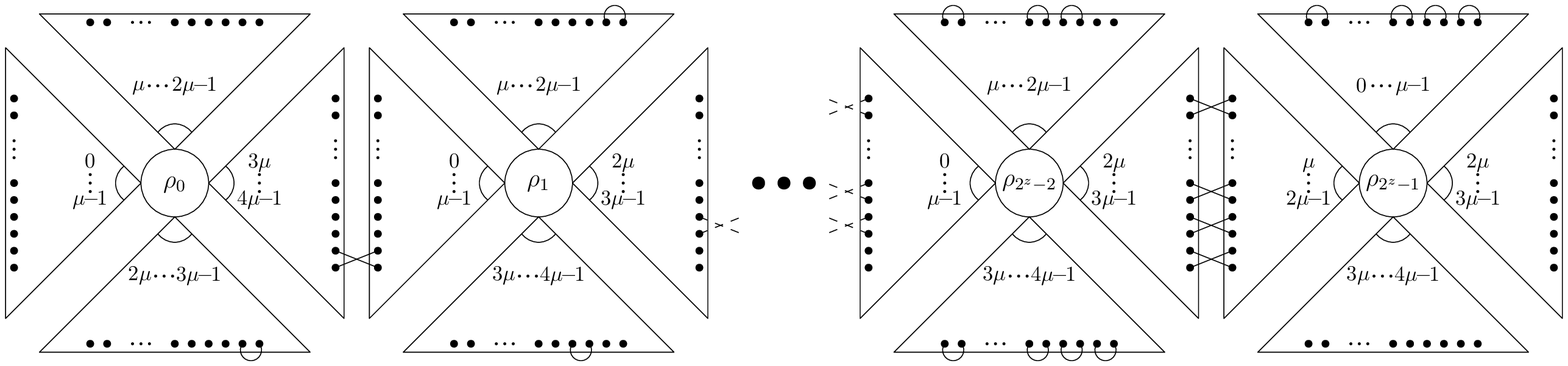}
	\caption{The graph $J_Y$ when $Y = (1,0,\ldots,0)$}
	\label{fig:JY-example-10}
\end{figure}

Taken over all possible binary sequences $Y$ of length exactly $2^{z-1}$, the graphs $J_Y$ together form the class $\mathcal{J}_{\mu,k}$. We see immediately that the number of graphs in $\mathcal{J}_{\mu,k}$ is the number of binary sequences of length $2^{z-1}$, where $z$ is the number of nodes in layer graph $L_k$ (and this value can be obtained using Fact \ref{fact:layersizes}).

\begin{fact}\label{fact:JClassSize}
	Let $z$ be the number of nodes in layer graph $L_k$. For any positive integers $\mu \geq 2$ and $k \geq 4$, we have $|\mathcal{J}_{\mu,k}| = 2^{2^{z-1}}$, where $\mu^{\lfloor k/2 \rfloor} \leq z \leq 4\mu^{\lfloor k/2 \rfloor}$.
\end{fact}
\begin{proof}
	The number of binary sequences of length $2^{z-1}$ is $2^{2^{z-1}}$. For the upper bound on $z$, note that by Fact \ref{fact:layersizes}, we have $L_{2j} \leq L_{2j+1} = \frac{2\mu^{j+1}-2}{\mu-1} \leq \frac{2\mu^{j+1}}{\mu/2} = 4\mu^j$. For the lower bound on $z$, note that by Fact \ref{fact:layersizes}, we have $L_{2j+1} \geq L_{2j} = \frac{\mu^{j+1} + \mu^j - 2}{\mu-1} \geq \frac{\mu^{j+1}}{\mu} = \mu^j$. Substituting $j = \lfloor k/2 \rfloor$ gives us the desired bounds on the number of nodes in layer graph $L_k$.
\end{proof}

The following result about each $J_Y \in \mathcal{J}_{\mu,k}$ will be instrumental in proving that the $\select$-index is always at least $k$, and it motivates many of the decisions made in the design of the graph class. In particular, within $k-1$ rounds of communication, no node can see all nodes in the $k^{th}$ layer of its component, so it cannot determine which integer is `encoded' by the added edges in the $k^{th}$ layer. This will later be used to show that no node can determine with certainty in which gadget it is located.

\begin{lemma}\label{lem:invisibleLeaf}
	For any $i \in \{0,\ldots,2^z-1\}$, any $D \in \{L,T,R,B\}$, and any node $v$ in component $H_D$ of gadget $\widehat{H_i}$, there exists an $\ell \in \{1,\ldots,z\}$ such that $w_{\ell,1}$ and $w_{\ell,2}$ in $H_D$ of $\widehat{H_i}$ are not contained in $\cB^{k-1}(v)$.
\end{lemma}
\begin{proof}
	We introduce some new terminology to distinguish two types of edges. An edge whose endpoints are within the same layer graph $L_j$ for some $j \in \{0,\ldots,k\}$ will be called a \emph{layer edge}. An edge whose endpoints are in consecutive layer graphs $L_j$ and $L_{j+1}$ for some $j \in \{0,\ldots,k-1\}$ will be called an \emph{inter-layer edge}. To clarify to which layer a node $u$ belongs, we will often write the layer number as a superscript, e.g., $u^j$ belongs to layer graph $L_j$. We begin by proving some technical claims about our construction of $H$. Claim \ref{claim:distance-preserved} can be verified by a case analysis of Part 2 of the construction.

	\begin{claim}\label{claim:distance-preserved}
	Consider any $m \in \{0,\ldots,k\}$, and any two inter-layer edges $\{u^m, u^{m+1}\}$ and $\{w^m,w^{m+1}\}$ such that $u^{m+1} \neq w^{m+1}$. If $d(u^{m},w^{m}) = t$ for some non-negative integer $t$, then $d(u^{m+1},w^{m+1}) \geq t$.
	\end{claim}

	\begin{claim}\label{claim:layer-edges-in-j1}
	Consider any node $u$ in layer $L_{j_1}$ and any node $w$ in layer $L_{j_2}$ such that $0 \leq j_1 \leq j_2 \leq k$. There exists a shortest path between $u$ and $w$ such that all layer edges in the path have both endpoints in $L_{j_1}$. 
	\end{claim}
 	To prove the claim, let $w^{j_1}$ be a node in layer $L_{j_1}$ such that $w$ can be reached from $w^{j_1}$ using $j_2 - j_1$ inter-layer edges. Denote by $v_{j_1+1},\ldots,v_{j_2-1}$ the interior nodes along this path from $w^{j_1}$ to $w$, and note that each $v_i$ belongs to layer $L_{i}$. 
 	Let $P_1$ be any shortest path from $u$ to $w$. Note that $P_1$ is some sequence of layer edges and inter-layer edges. For each $i \in j_1,\ldots,j_2$, denote by $V_i$ the subsequence of vertices on the path $P_1$ that are contained in layer $L_i$. Note that each $V_i$ is non-empty, since, by construction, each edge in $H$ connects two nodes in the same layer or in consecutive layers. In view of Claim \ref{claim:distance-preserved}, we may assume, without loss of generality, that $P_1$ has the form $V_{j_1}\cdot V_{j_1+1} \cdots V_{j_2}$. Let $j'$ be the largest index in $\{j_1,\ldots,j_2\}$ such that $V_{j'}$ contains more than one node. If $j' = j_1+1$, then we are done, as this would mean that only $V_{j_1}$ contains multiple nodes, and thus all layer edges would be contained in $L_{j_1}$, as desired. So, we proceed under the assumption that $j' > j_1+1$, and demonstrate a procedure that will strictly decrease the value of $j'$. Denote by $a_{j'-1}$ the last node in $V_{j'-1}$, and denote by $b_{j'}$ the first node of $V_{j'}$. 
 		\begin{itemize}
 			\item Suppose that $b_{j'} = v_{j'}$. We create a new path $P_2$ consisting of the node sequence $V_{j_1}\cdots V_{j'-1}\cdot v_{j'} \cdot v_{j'+1} \cdots v_{j_2}$.
 			\item Suppose that $b_{j'} \neq v_{j'}$. Note that the edge $\{a_{j'-1},b_{j'}\}$ is an inter-layer edge. Consider the distance between $a_{j'-1}$ and $v_{j'-1}$, and denote this distance by $t$. By Claim \ref{claim:distance-preserved}, the distance between $b_{j'}$ and $v_{j'}$ is at least $t$. So, we create a new path $P_2$ consisting of the node sequence $V_{j_1}\cdots V'_{j'-1}\cdot v_{j'} \cdot V_{j'+1}\cdots V_{j_2}$, where $V'_{j'-1}$ is the concatenation of the sequence $V_{j'-1}$ and the shortest path from $a_{j'-1}$ and $v_{j'-1}$.
 		\end{itemize}
 	 In both cases, $P_2$ is also a shortest path between $u$ and $w$. However, observe that $P_2$ can be written as the concatenation of sequences $W_{j_1} \cdots W_{j_2}$, where each $W_i$ is the subsequence of vertices on the path $P_2$ that are contained in layer $L_i$, and, note that the largest index $j'$ in $\{j_1,\ldots,j_2\}$ such that $W_{j'}$ contains more than one node is strictly smaller than the value of $j'$ for the path $P_1$. Repeating the above procedure enough times, we will eventually reach the case where $j' = j_1+1$, which, as remarked earlier, would complete the proof of Claim \ref{claim:layer-edges-in-j1}.

	\begin{claim}\label{claim:unique-Interlayer}
	Consider any $j \in \{2,\ldots,k\}$ and any node $w^{j}$ in layer $L_j$ of $H$, where $w^j = v^j_{b \leadsto \sigma}$ for some fixed $b \in \{0,1\}$ and some integer sequence $\sigma$. There exists a unique simple path $Q$ starting at $v^k_{b \leadsto \sigma}$ consisting only of inter-layer edges such that the other endpoint of $Q$ is in layer $L_j$. Moreover, $Q$ consists of exactly $k-j$ edges, and the two endpoints of $Q$ are $v^k_{b \leadsto \sigma}$ and $v^j_{b \leadsto \sigma}$.
	\end{claim}
	To prove the claim, we proceed by induction on the value of $j$. The result is trivial when $j=k$. Assume that the statement holds for some $j \in \{3,\ldots,k\}$, and consider any $w^{j-1} = v^{j-1}_{b \leadsto \sigma}$ for some fixed $b \in \{0,1\}$ and some integer sequence $\sigma$. From Part 2 of the construction, exactly one inter-layer edge is added between layers $L_{j-1}$ and $L_{j}$ with endpoint $v^{j}_{b \leadsto \sigma}$, and this edge is $\{v^{j-1}_{b \leadsto \sigma},v^{j}_{b \leadsto \sigma}\}$. By the induction hypothesis, there exists a unique simple path $Q$ starting at $v^k_{b \leadsto \sigma}$ consisting only of inter-layer edges such that the other endpoint of $Q$ is in layer $L_j$. Moreover, $Q$ consists of exactly $k-j$ edges, and the two endpoints of $Q$ are $v^k_{b \leadsto \sigma}$ and $v^j_{b \leadsto \sigma}$. Appending the edge $\{v^{j-1}_{b \leadsto \sigma},v^{j}_{b \leadsto \sigma}\}$ to $Q$ gives the unique simple path $Q'$ with endpoint in $L_{j-1}$, the length is $k-j+1 = k-(j-1)$, and the two endpoints are $v^k_{b \leadsto \sigma}$ and $v^{j-1}_{b \leadsto \sigma}$, which completes the induction step and the proof of Claim \ref{claim:unique-Interlayer}.

	\begin{claim}\label{claim:distance-j-in-Lj}
	For any $m \in \{0,\ldots,k\}$ and any node $u^{m}$ in layer $L_m$ of $H$, there exists a node $w^m$ in layer $L_m$ of $H$ such that the distance between $u^m$ and $w^m$ in $L_m$ is exactly $m$.
	\end{claim}
	To prove the claim, note that the cases $m=0$ and $m=1$ are trivial to verify by inspection. Another simple case is when $u^m = r^m_b$ for some $b \in \{0,1\}$, since, in this case, the node $w^m = r^m_{1-b}$ is at distance $m$ from $u^m$. In all other cases, note that $u^m$ is an internal node on some path $P$ of length $m$ between $r^m_0$ and $r^m_1$, where the port label at $r^m_0$ and the port label at $r^m_1$ on this path are both equal to some $p \in \{0,\ldots,\mu-1\}$. Consider the path $Q$ of length $m$ between $r^m_0$ and $r^m_1$ such that the port label at $r^m_0$ and the port label at $r^m_1$ on this path are both equal to $p+1$ modulo $\mu$. As $P$ and $Q$ are internally disjoint, they form a cycle of length $2m$. Thus, picking the node $w^m$ that is diametrically opposite to $u^m$ on this cycle satisfies the desired conditions. This completes the proof of Claim \ref{claim:distance-j-in-Lj}.

	We now proceed to prove the lemma. Consider any $j \in \{0,\ldots,k\}$ and any node $u^j$ in layer $L_j$ of $H$. We show that there exists an $\ell \in \{1,\ldots,z\}$ such that $w_{\ell,1}$ and $w_{\ell,2}$ are not contained in $\cB^{k-1}(u^j)$. By Claim \ref{claim:distance-j-in-Lj}, there exists a node $w^j$ in layer $L_j$ such that $d(u^j,w^j) = j$. Write $w^j = v^j_{b \leadsto \sigma}$ for some fixed $b \in \{0,1\}$ and some fixed integer sequence $\sigma$. Consider the node $v^k_{b \leadsto \sigma}$ from $L_k$. By Claim \ref{claim:layer-edges-in-j1}, there exists a shortest path between $u^j$ and $v^k_{b \leadsto \sigma}$ such that all layer edges in this path have both endpoints in $L_j$. It follows that there exists a shortest path $P$ between $u^j$ and $v^k_{b \leadsto \sigma}$ that, starting at $v^k_{b \leadsto \sigma}$, consists only of inter-layer edges until a node in layer $L_j$ is reached, and then consists only of layer edges within $L_j$. By Claim \ref{claim:unique-Interlayer}, the prefix of $P$ consisting of inter-layer edges has length exactly $k-j$, and the endpoint of this prefix in $L_j$ is $v^j_{b \leadsto \sigma}$. Recalling that $v^j_{b \leadsto \sigma} = w^j$, and that $d(u^j,w^j) = j$, it follows that the remainder of $P$ consisting of layer edges has length exactly $j$. Thus, we have shown that the length of shortest path $P$ between $u^j$ and $v^k_{b \leadsto \sigma}$ has length exactly $k$. The above proof applies to either copy of $v^k_{b \leadsto \sigma}$ from $L_{k,1}$ or $L_{k,2}$, which are represented as nodes $w_{\ell,1}$ and $w_{\ell,2}$ for some fixed $\ell \in \{1,\ldots,z\}$. Thus, we have shown that $d(u^j,w_{\ell,1})$ and $d(u^j,w_{\ell,2})$ are exactly $k$, which proves that $w_{\ell,1}$ and $w_{\ell,2}$ are not contained in $\cB^{k-1}(u^j)$. 
\end{proof}

\subsection{Minimum Election Time and Advice}
For each graph in $\mathcal{J}_{\mu,k}$ (where $\mu \geq 2$ and $k \geq 4$) we first show that the $\select$-index, the $\ppe$-index, and the $\cppe$-index are all equal to $k$.

To prove that the $\select$-index is at least $k$ for any graph $J_Y \in \mathcal{J}_{\mu,k}$, the idea is that each $J_Y$ was carefully constructed such that, when considering truncated views up to distance $k-1$, each node $v$ has at least one `twin' elsewhere in the graph with the same view. In particular, using only $k-1$ rounds, an arbitrary node $u$ in some component $H_D$ of a gadget $\widehat{H_i}$ cannot see all of the nodes in the $k^{th}$ layer of $H_D$, which is where the value of $i$ is encoded (see Part 4 of the construction). Supposing some $w_{q,1}$ and $w_{q,2}$ in the $k^{th}$ layer are not seen by $u$, then effectively $u$ cannot distinguish whether it is in gadget $\widehat{H_{i}}$ or $\widehat{H_{i'}}$, where $i'$ differs from $i$ only in the $q^{th}$ bit. This is the key observation needed to find a $u'$ with the same truncated view up to distance $k-1$ as $u$, which proves that $k-1$ rounds are not enough to elect a unique leader.

To formalize this idea, we first prove that all nodes in $\{\rho_0,\ldots,\rho_{2^z-1}\}$ have the same truncated view up to any distance $h \leq k-1$. This fact means that none of the $\rho$-nodes have a unique view up to distance $k-1$ (so none of them can be selected as leader within $k-1$ rounds), but it will also be a useful ingredient in the proof that no other node in any gadget has a unique view up to distance $k-1$.

\begin{proposition}\label{prop:rhoSame}
	For any $h \in \{0,\ldots,k-1\}$ and any $J_Y \in \mathcal{J}_{\mu,k}$, we have $\cB^{h}(\rho_0) = \cdots = \cB^{h}(\rho_{2^z-1})$ in $J_Y$.
\end{proposition}
\begin{proof}
	It is sufficient to prove the result for $h=k-1$, as two equal views up to distance $k-1$ are also equal up to any smaller distance. By Parts 1 and 2 of the construction, every edge in the component graph $H$ is either between two nodes in the same layer graph, or in consecutive layer graphs $L_m$ and $L_{m+1}$ for some $m \in \{0,\ldots,k-1\}$. It follows that $\cB^{k-1}(r^0_0)$ in the component graph $H$ only contains nodes from the layers $L_0,\ldots,L_{k-1}$, i.e., does not contain any nodes from either $L_{k,1}$ or $L_{k,2}$. By Part 3 of the construction, the gadget graph $\widehat{H}$ consists of four copies of $H$ (named $H_L$, $H_T$, $H_R$, $H_B$) such that the four $r^0_0$ nodes have been merged at a single node named $\rho$, and then the port numbers at $\rho$ are modified to be distinct in the range $0,\ldots,4\mu-1$. It follows that $\cB^{k-1}(\rho)$ in $\widehat{H}$ only contains nodes from the layers $L_0,\ldots,L_{k-1}$ from each of $H_L$, $H_T$, $H_R$, and $H_B$. By Part 4 of the construction, the only edges added when forming the template graph $J$ are incident to nodes in layer graphs $L_{k,1}$ and $L_{k,2}$ of $H_L, H_T, H_R, H_B$ in each gadget, i.e., these edges do not affect the degrees of nodes in layer graphs $L_0,\ldots,L_{k-1}$. It follows that $\cB^{k-1}(\rho_i)$ in each $\widehat{H_i}$ for $i \in \{0,\ldots,2^z-1\}$ is the same as $\cB^{k-1}(\rho)$ in $\widehat{H}$. Finally, when creating each $J_Y$ in Part 5 of the construction, note that in each $\widehat{H_i}$ for $i \in \{0,\ldots,2^z-1\}$, either the ports at $\rho_i$ in $H_T$ and $H_L$ are swapped, or the ports at $\rho_i$ in $H_B$ and $H_R$ are swapped. In either case, since $\cB^{k-1}(\rho_i)$ within each component $H_L$, $H_T$, $H_R$ and $H_B$ is identical, it follows that swapping corresponding ports at $\rho_i$ between two components does not change $\cB^{k-1}(\rho_i)$. This proves that, in an arbitrary $J_Y \in \mathcal{J}_{\mu,k}$, the views $\cB^{k-1}(\rho_0), \ldots, \cB^{k-1}(\rho_{2^z-1})$ are all the same as $\cB^{k-1}(\rho)$ in $\widehat{H}$, which implies the desired result.
\end{proof}

Next, we introduce the following notation to help us refer to specific nodes within a graph $J_Y$. Each node of $J_Y$ belongs to some gadget, and within the gadget, belongs to some component. So each node corresponds to some $u$ in the component graph $H$, we use a subscript from $\{L,T,R,B\}$ to indicate which of the four components of its gadget it belongs to, and also a subscript $i \in \{0,\ldots,2^z-1\}$ to indicate which gadget $\widehat{H_i}$ within $J_Y$ it belongs to. For example, $v = u_{L,2}$ indicates that the node $v$ belongs to $H_L$ in gadget $\widehat{H_2}$, and corresponds to node $u$ within the component graph $H$ defined in Part 2 of the construction. For any node $u_{D,i}$ for any $D \in \{L,T,R,B\}$ and any $i \in \{0,\ldots,2^z-1\}$, we say that $u_{D,i}$ is a \emph{border} node if it is located in either layer graph $L_{k,1}$ or $L_{k,2}$ in component $H_D$ of gadget $\widehat{H_i}$. Node $u_{D,i}$ is called an \emph{internal node} if it is not equal to $\rho_i$ and is not a border node. We say that a sequence of nodes $(v_1,\ldots,v_h)$ forms an \emph{internal path} if $v_1,\ldots,v_{h-1}$ are internal nodes, and $v_h$ is either an internal node or a border node.

The following fact will help us make arguments about symmetry in $J_Y$ and demonstrate that certain nodes have the same views up to distance at most $k-1$. It makes a key observation about the construction of each $J_Y$: the differences between different gadgets in $J_Y$, and the differences between components within a gadget, can only be noticed at the $\rho$ nodes and the border nodes. In particular, any fixed labeled path $P$ that avoids the $\rho$ nodes and the border nodes actually appears in every component of every gadget. To see why, note that the construction of each $J_Y$ is identical up to Part 3, and, in Parts 4 and 5, the only modifications involve adding edges between border nodes, and swapping edges at $\rho$ nodes.

\begin{fact}\label{fact:internalSame}
	Consider any node $u \in H$ such that $u \neq \rho$ and $u \not\in L_{k,1} \cup L_{k,2}$, and consider any node $u' \in H$ such that $u' \neq \rho$. For any $D,D' \in \{L,T,R,B\}$, any $i,i' \in \{0,\ldots,2^z-1\}$, any $h \in \{0,\ldots,k-1\}$, and any sequence $\sigma$ of length $2h$, the sequence $\sigma$ appears as the port labels along an internal path of length $h$ from $u_{D,i}$ to $u'_{D,i}$ in $J_Y$ if and only if $\sigma$ appears as the port labels along an internal path of length $h$ from $u_{D',i'}$ to $u'_{D',i'}$.
\end{fact}

We now prove that, when considering truncated views up to distance $k-1$, each node $v$ in $J_Y$ has a `twin' $v'$ that has the same view.

\begin{lemma}\label{lem:PPEtwin}
	For any $J_Y \in \mathcal{J}_{\mu,k}$ and any $i \in \{0,\ldots,2^z-1\}$, let $v$ be any node in $\widehat{H_i}$. There exists a node $v' \neq v$ in $J_Y$ such that $\cB^{k-1}(v) = \cB^{k-1}(v')$.
\end{lemma}
\begin{proof}
	If $v = \rho_i$, then, by Proposition \ref{prop:rhoSame}, setting $v'$ to be any node from $\{\rho_0,\ldots,\rho_{2^z-1}\} \setminus \{\rho_i\}$ gives the desired result. So we proceed under the assumption that $v \neq \rho_i$.
	
	Consider the case where $v$ is contained in component $H_L$ of $\widehat{H_i}$, i.e., $v$ is equal to $u_{L,i}$ for some node $u \in H - \{r^0_0\}$. By Lemma \ref{lem:invisibleLeaf}, there exists an $\ell \in \{1,\ldots,z\}$ such that the nodes $w_{\ell,1}$ and $w_{\ell,2}$ in $H_L$ of $\widehat{H_i}$ are not contained in $\cB^{k-1}(u_{L,i})$, and the nodes $w_{\ell,1}$ and $w_{\ell,2}$ in $H_R$ of $\widehat{H_{i-1}}$ are not contained in $\cB^{k-1}(u_{L,i})$. Let $i'$ be the integer whose $z$-bit binary representation is the same as the $z$-bit binary representation of $i$ except with the $\ell^{th}$ bit flipped. 
	
	From the description of Part 5 of the construction of $J_Y$, there are four cases to consider based on the values of $i$, $i'$, $y_{i}$, and $y_{i'}$:
	\begin{itemize}
		\item {\bf no-swap}: [$i \in \{0,\ldots,2^{z-1}-1\}$ or ($i \in \{2^{z-1},\ldots,2^{z}-1\}$ and $y_{2^z - 1 - i} = 0$)] and [$i' \in \{0,\ldots,2^{z-1}-1\}$ or ($i' \in \{2^{z-1},\ldots,2^{z}-1\}$ and $y_{2^z - 1 - i'} = 0$)]. In this case, no port swapping occurs at $\rho_i$ or $\rho_{i'}$ involving edges in $H_L$ and $H_T$.
		\item {\bf $i$-swap}: [$i \in \{2^{z-1},\ldots,2^{z}-1\}$ and $y_{2^z - 1 - i} = 1$] and [$i' \in \{0,\ldots,2^{z-1}-1\}$ or ($i' \in \{2^{z-1},\ldots,2^{z}-1\}$ and $y_{2^z - 1 - i'} = 0$)]. In this case, port swapping occurs at $\rho_{i}$ involving edges in $H_L$ and $H_T$, i.e., for each $g \in \{0,\ldots,\mu-1\}$, port $g$ at $\rho_{i}$ is swapped with port $g+\mu$ at $\rho_{i}$. No port swapping occurs at $\rho_{i'}$ involving edges in $H_L$ and $H_T$.
		\item {\bf $i'$-swap}: [$i \in \{0,\ldots,2^{z-1}-1\}$ or ($i \in \{2^{z-1},\ldots,2^{z}-1\}$ and $y_{2^z - 1 - i} = 0$)] and [$i' \in \{2^{z-1},\ldots,2^{z}-1\}$ and $y_{2^z - 1 - i'} = 1$]. In this case, no port swapping occurs at $\rho_i$ involving edges in $H_L$ and $H_T$. Port swapping occurs at $\rho_{i'}$ involving edges in $H_L$ and $H_T$, i.e., for each $g \in \{0,\ldots,\mu-1\}$, port $g$ at $\rho_{i'}$ is swapped with port $g+\mu$ at $\rho_{i'}$.
		\item {\bf both-swap}: [$i \in \{2^{z-1},\ldots,2^{z}-1\}$ and $y_{2^z - 1 - i} = 1$] and [$i' \in \{2^{z-1},\ldots,2^{z}-1\}$ and $y_{2^z - 1 - i'} = 1$]. In this case, port swapping occurs at both $\rho_i$ and $\rho_{i'}$ involving edges in $H_L$ and $H_T$, i.e., for each $g \in \{0,\ldots,\mu-1\}$, port $g$ at $\rho_{i}$ is swapped with port $g+\mu$ at $\rho_{i}$, and, port $g$ at $\rho_{i'}$ is swapped with port $g+\mu$ at $\rho_{i'}$.
	\end{itemize}
Wherever necessary, the proof will branch out to separately provide the details for the above four cases.

First, we specify the node $v'$ for which we will prove that $\cB^{k-1}(v) = \cB^{k-1}(v')$.
\begin{itemize}
	\item For cases {\bf no-swap} and {\bf both-swap}, we set $v' = u_{L,i'}$.
	\item For cases {\bf $i$-swap} and {\bf $i'$-swap}, we set $v' = u_{T,i'}$.
\end{itemize}

Our goal is to show that $\cB^{k-1}(v) = \cB^{k-1}(v')$, and we begin by proving that $\cB^{k-1}(v) \subseteq \cB^{k-1}(v')$. To this end, we consider any labeled root-to-leaf path $P$ of length $k-1$ in $\cB^{k-1}(v)$, and prove that $P$ also appears as a labeled root-to-leaf path in $\cB^{k-1}(v')$.

There are three cases to consider based on the nature of $P$.
\begin{itemize}
	\item {\bf Case 1:} The path $P$ contains node $\rho_i$. Denote by $u''_{L,i}$ the neighbour of $\rho_i$ that appears immediately before the first occurrence of $\rho_i$ in $P$. Consider $P = P_1\cdot e \cdot P_2$, where:
	\begin{itemize}
		\item $P_1$ is the (possibly empty) path in $\cB^{k-1}(v)$ from $v=u_{L,i}$ to $u''_{L,i}$,
		\item $e$ is the edge $\{u''_{L,i},\rho_i\}$, and, 
		\item $P_2$ is the (possibly empty) remainder of the path from $\rho_i$ to the leaf of $P$. 
	\end{itemize}
	Let $\sigma_1$ be the port sequence on path $P_1$ from $u_{L,i}$ to $u''_{L,i}$, and let $\sigma_2$ be the port sequence on path $P_2$ from $\rho_i$ to the leaf of $P$.
	As the total length of $P$ is $k-1$ and $P$ contains $\rho_i$, it follows that $P$ does not contain any nodes from layer graphs $L_{k,1}$ or $L_{k,2}$. Therefore, the path from $v=u_{L,i}$ to $u''_{L,i}$ consists completely of internal nodes, so by Fact \ref{fact:internalSame}, $\sigma_1$ is also a port sequence from $u_{L,i'}$ to $u''_{L,i'}$ (needed for cases {\bf no-swap} and {\bf both-swap}), and is also a port sequence from $u_{T,i'}$ to $u''_{T,i'}$ (needed for cases {\bf $i$-swap} and {\bf $i'$-swap}).
	 
	 Next, for the edge $e$, we consider the four cases:
	 \begin{itemize}
	 	\item In case {\bf no-swap}, we argue that the edge between $u''_{L,i}$ and $\rho_i$ is labeled the same as the edge between $u''_{L,i'}$ and $\rho_{i'}$. By Part 2 of the construction, the edge between $u''$ and $r^0_0$ in $H$ is labeled with some $g \in \{0,\ldots,\mu-1\}$ at $r^0_0$, and labeled with $\mu-1$ at $u''$. These are the same port labels on the edge between $u''$ and $\rho$ in $H_L$, by Part 3 of the construction. No port labels at any $\rho$ node are modified in Part 4 of the construction.  Since no port swapping occurs involving $H_L$ at $\rho_i$ and $\rho_{i'}$, we conclude that the edge between $u''_{L,i}$ and $\rho_i$ is labeled the same way as the edge between $u''_{L,i'}$ and $\rho_{i'}$, i.e., labeled $\mu-1$ at both $u''_{L,i}$ and $u''_{L,i'}$, and labeled $g$ at both $\rho_i$ and $\rho_{i'}$, as desired.
	 	
	 	\item In case {\bf $i$-swap}, we argue that the edge between $u''_{L,i}$ and $\rho_i$ is labeled the same as the edge between $u''_{T,i'}$ and $\rho_{i'}$. By Part 2 of the construction, the edge between $u''$ and $r^0_0$ in $H$ is labeled with some $g \in \{0,\ldots,\mu-1\}$ at $r^0_0$, and labeled with $\mu-1$ at $u''$. These are the same port labels on the edge between $u''$ and $\rho$ in $H_L$, by Part 3 of the construction. No port labels at any $\rho$ node are modified in Part 4 of the construction. In Part 5 of the construction, port swapping occurs at $\rho_i$, in particular, ports $g$ and $g+\mu$ are swapped at $\rho_i$, which implies that the edge between $u''_{L,i}$ and $\rho_i$ is labeled with $\mu-1$ at $u''_{L,i}$ and labeled with $g+\mu$ at $\rho_i$. However, no port swapping occurs at $\rho_{i'}$ involving edges in $H_L$ and $H_T$ so the edge between $u''_{T,i'}$ and $\rho_{i'}$ is labeled with $\mu-1$ at $u''_{T,i'}$ and labeled with $g+\mu$ at $\rho_{i'}$, as desired.
	 	
	 	\item In case {\bf $i'$-swap}, we argue that the edge between $u''_{L,i}$ and $\rho_i$ is labeled the same as the edge between $u''_{T,i'}$ and $\rho_{i'}$. By Part 2 of the construction, the edge between $u''$ and $r^0_0$ in $H$ is labeled with some $g \in \{0,\ldots,\mu-1\}$ at $r^0_0$, and labeled with $\mu-1$ at $u''$. These are the same port labels on the edge between $u''$ and $\rho$ in $H_L$, by Part 3 of the construction. No port labels at any $\rho$ node are modified in Part 4 of the construction. In Part 5 of the construction, no port swapping occurs at $\rho_i$ involving edges in $H_L$ and $H_T$, so the edge between  $u''_{L,i}$ and $\rho_i$ is labeled with $\mu-1$ at $u''_{L,i}$ and labeled with $g$ at $\rho_i$. However, port swapping does occur at $\rho_{i'}$, in particular, ports $g$ and $g+\mu$ are swapped at $\rho_{i'}$, which implies that the edge between $u''_{T,i'}$ and $\rho_{i'}$ is labeled with $\mu-1$ at $u''_{T,i'}$ and labeled with $g$ at $\rho_{i'}$, as desired.
	 	
	 	\item In case {\bf both-swap}, we argue that the edge between $u''_{L,i}$ and $\rho_i$ is labeled the same as the edge between $u''_{L,i'}$ and $\rho_{i'}$. By Part 2 of the construction, the edge between $u''$ and $r^0_0$ in $H$ is labeled with some $g \in \{0,\ldots,\mu-1\}$ at $r^0_0$, and labeled with $\mu-1$ at $u''$. These are the same port labels on the edge between $u''$ and $\rho$ in $H_L$, by Part 3 of the construction. No port labels at any $\rho$ node are modified in Part 4 of the construction. In Part 5 of the construction, port swapping occurs at $\rho_i$, in particular, ports $g$ and $g+\mu$ are swapped at $\rho_i$, which implies that the edge between $u''_{L,i}$ and $\rho_i$ is labeled with $\mu-1$ at $u''_{L,i}$ and labeled with $g+\mu$ at $\rho_i$.  Moreover, port swapping occurs at $\rho_{i'}$, in particular, ports $g$ and $g+\mu$ are swapped at $\rho_{i'}$, which implies that the edge between $u''_{L,i'}$ and $\rho_{i'}$ is labeled with $\mu-1$ at $u''_{L,i'}$ and labeled with $g+\mu$ at $\rho_{i'}$, as desired.
	 \end{itemize}
 	
 	Finally, as $|P_2| \leq |P| - 1 < k-1$, it follows from Proposition \ref{prop:rhoSame} that $\cB^{|P_2|}(\rho_{i}) = \cB^{|P_2|}(\rho_{i'})$, so the same port sequence $\sigma_2$ exists on a root-to-leaf path in $\cB^{|P_2|}(\rho_{i'})$.
 	
 	We have shown that the three pieces of $P = P_1\cdot e \cdot P_2$ appear in the desired order starting from node $v'$, which completes the proof that the labeled root-to-leaf path $P$ also exists in $\cB^{k-1}(v')$. 
 	
	\item {\bf Case 2:} The path $P$ is an internal path. Then $P$ is a path of length $k-1$ from $v = u_{L,i}$ to some $u'_{L,i}$. Denote by $\sigma$ the port sequence from $u_{L,i}$ to $u'_{L,i}$ on path $P$. By Fact \ref{fact:internalSame}, the port sequence $\sigma$ also labels a path from $u_{L,i'}$ to $u'_{L,i'}$ (needed for cases {\bf no-swap} and {\bf both-swap}) and also labels a path from $u_{T,i'}$ to $u'_{T,i'}$ (needed for cases {\bf $i$-swap} and {\bf $i'$-swap}), which completes the proof that the labeled root-to-leaf path $P$ also exists in $\cB^{k-1}(v')$.

	\item {\bf Case 3:} The path $P$ does not contain $\rho_i$ but contains node $w_{q,b}$ for some $q \in \{1,\ldots,z\} - \{\ell\}$ and $b \in \{1,2\}$.
	We prove a stronger statement: if $P$ (with the given property) is a labeled root-to-leaf path in $\cB^{h}(u_{L,i})$ or $\cB^{h}(u_{R,i-1})$, then $P$ appears as a labeled root-to-leaf path in $\cB^{h}(u_{L,i'})$ and $\cB^{h}(u_{T,i'})$ (which covers both of our possible choices for $v'$).
	
	We proceed by induction on the length $h \geq 0$ of $P$. For base case, consider $h=0$. As $P$ contains node $w_{q,b}$ for some $q \in \{1,\ldots,z\} - \{\ell\}$ and $b \in \{1,2\}$, it follows that $u = w_{q,b}$. By Part 4 of the construction, the $q^{th}$ bit of $i$ is 1 if and only if an edge was added incident to $u_{L,i}$ and an edge was added incident to $u_{R,i-1}$. Thus, the nodes $u_{L,i}$ and $u_{R,i-1}$ have the same degree. Further, by our choice of $i'$, the $q^{th}$ bit of $i'$ is the same as the $q^{th}$ bit of $i$. So, by Part 4 of the construction, we have that each of the nodes $u_{L,i'}$ and $u_{T,i'}$ has the same degree as each of the three nodes $u_{L,i}$ and $u_{R,i-1}$. This proves that $\cB^{0}(u_{L,i}) = \cB^{0}(u_{R,i-1}) = \cB^{0}(u_{L,i'}) = \cB^{0}(u_{T,i'})$, which implies the desired result.
	
	As induction hypothesis, assume that, for some $h \in \{0,\ldots,k-2\}$, if $P'$ (with the given property) is a labeled root-to-leaf path in $\cB^{h}(u'_{L,i})$ or $\cB^{h}(u'_{R,i-1})$, then $P'$ appears as a labeled root-to-leaf path in $\cB^{h}(u'_{L,i'})$ and $\cB^{h}(u'_{T,i'})$.
	
	For the induction step, suppose that $P$ (with the given property) is a labeled root-to-leaf path in $\cB^{h+1}(u_{L,i})$ or $\cB^{h+1}(u_{R,i-1})$. Denote by $w'$ a node from $\{u_{L,i},u_{R,i-1}\}$ for which the previous sentence holds. Let $w''$ denote the first border node along the path $P$. Re-write $P = P_1\cdot e \cdot P_2$, where $P_1$ is the path from $w'$ to $w''$, edge $e$ is the outgoing edge from $w''$ along the path $P$ to the first node of $P_2$ (which we denote by $w'''$), and $P_2$ is the remainder of the path $P$. Since $P$ starts at $u_{L,i}$ or $u_{R,i-1}$, has length at most $k-1$, and contains a border node (i.e., contains a node at distance at least $k$ from $\rho_i$ and $\rho_{i-1}$), we can conclude that $P$ does not contain $\rho_i$ or $\rho_{i-1}$. It follows that: 
	\begin{itemize}
		\item $w''$ (the last node of $P_1$) is in the same component as $w'$, i.e., if $w' = u_{L,i}$, then $w''=u''_{L,i}$ for $u''=w_{q,b}$, where $q \in \{1,\ldots,z\} - \{\ell\}$ and $b \in \{1,2\}$, and, if $w' = u_{R,i-1}$, then $w''=u''_{R,i-1}$ for $u''=w_{q,b}$, where $q \in \{1,\ldots,z\} - \{\ell\}$ and $b \in \{1,2\}$.
		\item $w'''$ (the first node of $P_2$) is equal to node $u'''_{L,i}$ or $u'''_{R,i-1}$ for some node $u''' \in H - \{r^0_0\}$.
	\end{itemize}
	
	First, consider the path $P_1$ from $w'$ to $w''$. By our choice of $w''$, i.e., the first border node along path $P$, we know that the path $P_1$ from $w'$ to $w''$ is an internal path, so by Fact \ref{fact:internalSame}, $P_1$ appears as a labeled root-to-leaf path in each of $\cB^{|P_1|}(u_{L,i'})$ and $\cB^{|P_1|}(u_{T,i'})$.
	
	Next, consider the edge $e=\{w'',w'''\}$. By our choice of $i'$ and the fact that $q \neq \ell$, we know that the $q^{th}$ bit of $i$ is equal to the $q^{th}$ bit of $i'$. So, after Part 4 of the construction of $J_Y$, either all of the following pairs are edges (with all ports labeled with the same integer $deg_H(w_{q,1})$) or all are non-edges:
	\begin{itemize}
		\item $w_{q,b}$ in $H_L$ of $\widehat{H_{i}}$ and $w_{q,2-b}$ in $H_R$ of $\widehat{H_{i-1}}$
		\item $w_{q,b}$ in $H_R$ of $\widehat{H_{i-1}}$ and $w_{q,2-b}$ in $H_L$ of $\widehat{H_{i}}$
		\item $w_{q,b}$ in $H_L$ of $\widehat{H_{i'}}$ and $w_{q,2-b}$ in $H_R$ of $\widehat{H_{i'-1}}$
		\item $w_{q,b}$ in $H_R$ of $\widehat{H_{i'-1}}$ and $w_{q,2-b}$ in $H_L$ of $\widehat{H_{i'}}$
		\item $w_{q,b}$ and $w_{q,2-b}$ in $H_T$ of $\widehat{H_{i}}$
		\item $w_{q,b}$ and $w_{q,2-b}$ in $H_T$ of $\widehat{H_{i'}}$
	\end{itemize}
	Recall that $u''=w_{q,b}$ and the immediate neighbourhoods of $u''_{L,i}, u''_{R,i-1}, u''_{L,i'}, u''_{T,i'}$ are all the same at the end of Part 3 of the construction (all components in all gadgets are just copies of $H$). Together with the fact that either all of the above pairs are non-edges or they are identically-labeled edges, it follows that $\cB^{1}(u''_{L,i}) = \cB^{1}(u''_{R,i-1}) = \cB^{1}(u''_{L,i'}) = \cB^{1}(u''_{T,i'})$. This is sufficient to show that the edge $e$ (which is an edge in $\cB^{1}(u''_{L,i})$ or $\cB^{1}(u''_{R,i-1})$) also appears as an edge in $\cB^{1}(u''_{L,i'})$ and $\cB^{1}(u''_{T,i'})$.
	
	Finally, there are two possibilities for the path $P_2$. If $P_2$ is an internal path, then by Fact \ref{fact:internalSame}, $P_2$ appears as a labeled root-to-leaf path  in each of $\cB^{|P_2|}(u'''_{L,i'})$ and $\cB^{|P_2|}(u'''_{T,i'})$. The other possibility is that $P_2$ contains a border node. Note that $P_2$ has length strictly less than $P$, so we apply the induction hypothesis to $P_2$ starting at node $w'''$ (which is equal to either $u'''_{L,i}$ or $u'''_{R,i-1}$), and we conclude that $P_2$ also exists as a labeled root-to-leaf path in each of $\cB^{|P_2|}(u'''_{L,i'})$ and $\cB^{|P_2|}(u'''_{T,i'})$.
	
	This concludes the induction step: we showed that $P$ appears as a labeled root-to-leaf path in $\cB^{h+1}(u'_{L,i'})$ and $\cB^{h+1}(u'_{T,i'})$ by showing that each of the three pieces of $P = P_1\cdot e \cdot P_2$ appears in the desired order.
\end{itemize}
In all three cases, we proved that $P$ also exists in $\cB^{k-1}(v')$, i.e., we have shown that $\cB^{k-1}(v) \subseteq \cB^{k-1}(v')$. A symmetric argument proves that any labeled root-to-leaf path $P'$ of length $k-1$ in $\cB^{k-1}(v')$ also exists as a labeled root-to-leaf path in $\cB^{k-1}(v)$, which completes the proof that $\cB^{k-1}(v) = \cB^{k-1}(v')$.

The above proof considers the case where $v$ is assumed to be in component $H_L$ of $\widehat{H_i}$, i.e., $v$ is equal to $u_{L,i}$ for some node $u \in H - \{r^0_0\}$. A nearly identical proof works under the assumption that $v$ is in component $H_T$ of $\widehat{H_i}$. Then, by symmetry, one can write similar proofs for $v$ in $H_R$ of $\widehat{H_i}$ and $v$ in $H_B$ of $\widehat{H_i}$ (where $R$ takes on the role of $L$, and $B$ takes on the role of $T$).
\end{proof}

The previous result implies that no node in any $J_Y$ can have a unique truncated view up to distance $k-1$, which gives us the following lower bound on the $\select$-index of $J_Y$.
\begin{lemma}\label{lem:SgeqkJ}
	$\psi_{\select}(J_Y) \geq k$ for any graph $J_Y \in \mathcal{J}_{\mu,k}$.
\end{lemma}

Next, we prove that the $\cppe$-index is at most $k$ for any graph $J_Y \in \mathcal{J}_{\mu,k}$. We specify an algorithm that, when given a map of $J_Y$ as input, gets each node to output a complete port sequence on a path from itself to node $\rho_0$ of gadget $\widehat{H_0}$. The idea is that, within $k$ rounds, each node can determine in which gadget $\widehat{H_i}$ it is located: it will be able to see the entire $k^{th}$ layer of whichever component it is in (one of $H_L$, $H_T$, $H_R$ or $H_B$ of its gadget $\widehat{H_i}$) and, as it knows the procedure used in Part 4 to `encode' the value of $i$ via added edges in the $k^{th}$ layer, it can deduce the value of $i$ used in the encoding. Then, each node computes a complete port sequence on a path from itself to $\rho_i$ in its gadget $\widehat{H_i}$ (as $\rho_i$ is the unique node in its view with largest degree), then consults the given map of $J_Y$ to determine a complete port sequence on a path from $\rho_i$ to $\rho_0$.

\begin{lemma}\label{lem:CPPEleqkJ}
	$\psi_{\cppe}(J_Y) \leq k$ for any graph $J_Y \in \mathcal{J}_{\mu,k}$.
\end{lemma}
\begin{proof}
	We present an algorithm that solves $\cppe$ within $k$ rounds when executed by the nodes of any graph $J_Y \in \mathcal{J}_{\mu,k}$ when the full map of $J_Y$ is provided to each node.
	
	First, we introduce some new notation for the purposes of our algorithm. For any fixed $i \in \{0,\ldots,2^z-1\}$, recall that the gadget $\widehat{H_i}$ has four `sub-components' (in the construction, these were called $H_L$, $H_T$, $H_R$, and $H_B$). However, in a given map of $J_Y$, or in a node's view, these sub-components are not labeled as $H_L$, $H_T$, $H_R$, and $H_B$. Instead, for each $c \in \{0,1,2,3\}$, we denote by $H_{i,c}$ the subgraph induced by nodes within distance $k$ of $\rho_i$ that can be reached using the outgoing ports $\mu c,\ldots,\mu(c+1)-1$. Further, for each $c \in \{0,1,2,3\}$, we associate an integer $W_{i,c}$ whose $z$-bit binary representation is encoded in the $k^{th}$ layer of $H_{i,c}$. In particular, the $q^{th}$ bit of the $z$-bit binary representation of $W_{i,c}$ is 1 if and only if $deg_{(J_Y)}(w_{q,1}) = deg_{H}(w_{q,1})+1$ in $H_{i,c}$. At a high level, the value of $W_{i,c}$ is obtained by `decoding' the value that was encoded using edges in Part 4 of the construction in building the template graph $J$.
	
	Our $\cppe$ algorithm proceeds as follows. First, assuming that a full map of $J_Y$ is given as input to each node, the following pre-processing occurs at each node before any communication takes place.
	\begin{enumerate}[label=\arabic*.]
		\item Use the map of $J_Y$ to find the nodes $\rho_0,\ldots,\rho_{2^z-1}$: these are the $2^z$ nodes that have the (same) largest degree in the map. By construction, these nodes have degree $4\mu$, so we deduce the value of $\mu$. Next, to determine which node is $\rho_0$ and which is $\rho_{2^z-1}$, compute the four integers $W_{x,0}$, $W_{x,1}$, $W_{x,2}$, and $W_{x,3}$ at each of the two extreme gadgets in the given map of $J_Y$. In particular, the gadget for which two of these integers are 0 and the other two are 1 is necessarily gadget $\widehat{H_0}$, and the gadget for which two of these integers are 0 and the other two are $2^z-1$ is necessarily gadget $\widehat{H_{2^z-1}}$.  Using this information, label the center node of each gadget on the map of $J_Y$ as $\rho_i$ using the correct subscript $i$. 
		\item For each $i \in \{1,\ldots,2^z-1\}$, compute a shortest path $P_i$ in the map from $\rho_i$ to $\rho_{i-1}$, and define $\sigma_i$ to be the sequence of ports along such a path. These paths and port sequences will be used later when determining the algorithm's output.
	\end{enumerate}
	
	Next, each node $v$ considers its own degree and behaves according to one of the two following cases.
	\begin{itemize}
		\item If $v$ has degree $4\mu$, then $v$ concludes that it is some $\rho_x \in \{\rho_0,\ldots,\rho_{2^z-1}\}$, and determines which one it is, as follows. First, using $k$ communication rounds, $v$ computes $\cB^k(v)$, and using this view, computes the four integers $W_{x,0}, W_{x,1}, W_{x,2}, W_{x,3}$. There are 2 possible cases: if two of these integers are 0 and the other two are $2^z-1$, then $v$ concludes that $x = 2^z-1$; otherwise, two of these integers are equal to some fixed $j \in \{0,\ldots,2^z-2\}$ and the other two are $j+1$, in which case $v$ concludes that $x = j$. 
		
		After determining the value of $x$ such that $v = \rho_x$, node $v$ produces an output, and there are two possible cases: if $x=0$, then $v$ outputs `leader', and, otherwise, $v$ outputs $\sigma_x \cdots \sigma_1$, i.e., the concatenation of port sequences that label a shortest path from $\rho_i$ to $\rho_{i-1}$ for each $i = x, \ldots, 1$.
		\item If $v$ does not have degree $4\mu$, then $v$ concludes that it is not one of the nodes $\rho_0,\ldots,\rho_{2^z-1}$. Node $v$ must belong to gadget $\widehat{H_x}$ for some $x \in \{0,\ldots,2^z-1\}$, and it determines the value of $x$, as follows. First, using $k$ communication rounds, $v$ computes $\cB^k(v)$. There is exactly one node with degree $4\mu$ in this view, and this node is $\rho_x$. Denote by $p$ the last port number on a shortest path from $v$ to $\rho_x$ (this is the port incident to $\rho_x$). This port number $p$ is in the range $\mu c,\ldots, \mu(c+1)-1$ for some $c \in \{0,1,2,3\}$, and it is straightforward for $v$ to use $p$ to determine this value of $c$. The nodes $w_{1,1},\ldots,w_{z,1},w_{1,2},\ldots,w_{z,2}$ of $H_{x,c}$ are all within $v$'s view up to distance $k$, so, $v$ uses their degrees to compute $W_{x,c}$, as described above. Next, $v$ uses the values of $c$ and $W_{x,c}$ to determine $x$ based on the following cases:
		\begin{itemize}
			\item if $W_{x,c} \leq 2^{z-1} - 1$, and,
			\begin{itemize}
				\item if $c \in \{0,1\}$, then set $x = W_{x,c}$.
				\item if $c \in \{2,3\}$, then set $x = W_{x,c}-1$.
			\end{itemize}
			\item if $W_{x,c} \geq 2^{z-1}$, and,
			\begin{itemize}
				\item if $c \in \{0,1\}$, then set $x = W_{x,c}-1$.
				\item if $c \in \{2,3\}$, then set $x = W_{x,c}$.
			\end{itemize}		
		\end{itemize}
		
		After determining the value of $x$ such that $v$ belongs to gadget $\widehat{H_x}$, node $v$ constructs its output, as follows. First, $v$ computes a shortest path $Q_x$ in $\cB^{k}(v)$ from itself to $\rho_x$ (where $\rho_x$ is the only node with degree $4\mu$ in $J_Y$ that is within $v$'s view up to distance $k$). Let $u$ be the first node on the path $Q_x$ that is contained in $P_x$. (Note that, in many cases, node $u$ is simply $\rho_x$, but $u$ might be a different node if $v$ is in $H_L$ of $\widehat{H_x}$). Define $s_x$ to be the port sequence that labels the part of $Q_x$ from $v$ to $u$, and define $t_x$ to be the port sequence that labels the part of $P_x$ from $u$ to $\rho_{x-1}$ (if $u = \rho_x$, then $t_x$ is simply $\sigma_x$). Finally, $v$ outputs the following concatenation of port sequences: $s_x\cdot t_x\cdot\sigma_{x-1}\cdots\sigma_1$.
	\end{itemize}
\end{proof}

By Fact \ref{indexHierarchy}, Lemma \ref{lem:SgeqkJ} and Lemma \ref{lem:CPPEleqkJ}, we see that $k \geq \psi_{\cppe}(J_Y) \geq \psi_{\ppe}(J_Y) \geq \psi_{\select}(J_Y) \geq k$ for all $J_Y \in \mathcal{J}_{\mu,k}$, which allows us to conclude that the $\select$-index, $\ppe$-index, and $\cppe$-index are all equal to $k$ in our constructed class $\mathcal{J}_{\mu,k}$.

\begin{lemma}\label{lem:PPECPPEindexk}
	For any integers $\mu \geq 2$ and $k \geq 4$, $\psi_{\cppe}(J_Y) = \psi_{\ppe}(J_Y) = \psi_{\select}(J_Y)  = k$ for any graph $J_Y \in \mathcal{J}_{\mu,k}$.
\end{lemma}

We now proceed to analyze the amount of advice needed to solve Port Path Election and Complete Port Path Election in the graph class $\mathcal{J}_{\mu,k}$. To show that a large amount of advice is needed, the idea is to think of each $J_Y$ as having a `left' half (gadgets $\widehat{H_0},\ldots,\widehat{H_{2^{z-1}-1}}$) and a `right' half (gadgets $\widehat{H_{2^{z-1}}},\ldots,\widehat{H_{2^{z}-1}}$) and observing that, whichever node $u$ is elected as leader, there exists a node $v$ on the opposite half that must output a sequence of ports corresponding to a very long simple path from $v$ to $u$. But, most of the ports on this path are outside of $v$'s view up to distance $k$, so $v$ has to depend on the given advice to help it determine the port sequence. The goal is to show that $v$ has to output a different port sequence for each graph in the class, and that a different piece of advice is needed for each. By the construction of $J_Y$, the $i^{th}$ binary entry of $Y$ determines whether or not some ports at $\rho_i$ are swapped, so even changing one entry of $Y$ will affect whether or not a port leads in the correct direction towards the opposite half of $J_Y$. The following result formalizes these observations, in particular, (1) a node that is on one `edge' of $J_Y$ cannot see far enough to detect any port swaps due to values in $Y$, and, (2) for any two different graphs $J_{\alpha}$ and $J_{\beta}$, no fixed port sequence can correspond to a simple path in both graphs that starts from one `edge' node on one half and ending in the opposite half.

\begin{lemma}\label{lem:valphavbeta}
	Consider any positive integers $\mu \geq 2$ and $k \geq 4$. Let $z$ be the number of nodes in layer graph $L_k$, and let $\alpha, \beta$ be distinct binary sequences of length exactly $2^{z-1}$. For each $Y \in \{\alpha,\beta\}$, let $v_Y$ be the node $w_{1,1}$ in component $H_L$ of gadget $\widehat{H_0}$ of $J_Y \in \mathcal{J}_{\mu,k}$. Then:
	\begin{enumerate}[label=(\arabic*)]
		\item $\cB^{k}(v_\alpha)$ in $J_\alpha$ is equal to $\cB^{k}(v_\beta)$ in $J_\beta$.\label{part1:valphavbeta}
		\item Suppose that $\sigma$ is a fixed sequence of ports that corresponds to a simple path $P_\alpha$ in $J_\alpha$ starting at node $v_\alpha$ such that $P_\alpha$ contains at least one node from gadget $\widehat{H_{2^{z-1}}}$ of $J_\alpha$. Then $\sigma$ corresponds to a path $P_\beta$ in $J_\beta$ starting at node $v_\beta$ such that either $P_\beta$ is not simple, or, $P_\beta$ only contains nodes from gadgets $\widehat{H_0},\ldots,\widehat{H_{2^{z-1}-1}}$ of $J_\beta$.\label{part2:valphavbeta}
	\end{enumerate}
\end{lemma}
\begin{proof}
	To prove statement \ref{part1:valphavbeta}, first recall from Part 4 of the construction that node $w_{1,1}$ is in layer $L_k$ of its component. By Part 4 of the construction, no edges are added incident to nodes in $H_L$ of $\widehat{H_0}$. Further, by Part 2 of the construction, the distance from any node in $L_k$ to $L_0$ in $H_L$ is $k$. So, if we consider node $w_{1,1}$ in $H_L$ of gadget $\widehat{H_0}$, the previous two facts allow us to conclude that its truncated view up to distance $k$ is completely contained in $H_L$ of gadget $\widehat{H_0}$ in the template graph $J$. By Part 3 of the construction, the port numbers at $\rho$ at edges in $H_L$ are in the range $0,\ldots,\mu-1$, so, in Part 5 of the construction, the port numbers at edges in $H_L$ are never swapped with others. Altogether, we get that the node $w_{1,1}$ in $H_L$ of gadget $\widehat{H_0}$ has the same truncated view up to distance $k$ in the template graph $J$ as it does in every graph of the class $\mathcal{J}_{\mu,k}$, which is sufficient to prove statement \ref{part1:valphavbeta}.
	
	To prove statement \ref{part2:valphavbeta}, consider any $Y \in\{\alpha,\beta\}$, and consider any fixed port sequence $\sigma$ that corresponds to a path $P_Y$ in $J_Y$ starting at node $v_Y$ such that $P_Y$ has the following two properties: it is a simple path, and, it contains at least one node from gadget $\widehat{H_{2^{z-1}}}$ of $J_Y$. By the choice of $v_Y$ and Part 4 of the construction, to have both of these properties, $P_Y$ necessarily only contains nodes in $H_L$ and $H_R$ of each gadget $\widehat{H_0},\ldots,\widehat{H_{2^{z-1}-1}}$ of $J_Y$. This is because the only edges that have endpoints in two different gadgets are incident to nodes in components $H_L$ and $H_R$, and, any path between two components of a gadget must pass through node $\rho$ of the gadget (and a simple path can only use node $\rho$ once).
	
	Assume that $\sigma$ is a port sequence that corresponds to a path $P_\alpha$ in $J_\alpha$ starting at node $v_\alpha$ such that $P_\alpha$ is a simple path, and, it contains at least one node from gadget $\widehat{H_{2^{z-1}}}$ of $J_\alpha$. Let $m \leq 2^{z-1}-1$ be the smallest index where the sequences $\alpha$ and $\beta$ differ. Without loss of generality, assume that $\alpha_m = 0$ and $\beta_m=1$. By Part 5 of the construction, the gadgets $\widehat{H_0},\ldots,\widehat{H_{m-1}}$ in $J_\alpha$ and $J_\beta$ are identical, so, following the port sequence $\sigma$ will trace out the same path up to node $\rho_m$ (in gadget $\widehat{H_{m}}$) in both $J_\alpha$ and $J_\beta$. Since $\alpha_m=0$, gadget $\widehat{H_m}$ of $J_\alpha$ is the same as in $J$, which means that the port in $\sigma$ that is used to leave node $\rho_m$ in $\widehat{H_m}$ of $J_\alpha$ to enter component $H_R$ is in the range $2\mu,\ldots,3\mu-1$. In $J_\beta$, the ports $2\mu,\ldots,3\mu-1$ at node $\rho_m$ of gadget $\widehat{H_m}$ have been swapped with the ports $3\mu,\ldots,4\mu-1$. But, in the template graph, the ports $3\mu,\ldots,4\mu-1$ lead to nodes in $H_B$. So, following the port sequence $\sigma$ in $J_\beta$ will result in a path $P_\beta$ that contains a node in component $H_B$ of gadget $\widehat{H_m}$. As remarked in the previous paragraph, this means that $P_\beta$ does not have both specified properties, i.e., either $P_\beta$ is not simple, or, $P_\beta$ does not contain at least one node from gadget $\widehat{H_{2^{z-1}}}$ of $J_\beta$, which concludes the proof of statement \ref{part2:valphavbeta}.
\end{proof}

We prove that a large amount of advice is needed by any algorithm that solves $\ppe$ in time $k$ for all graphs in $\mathcal{J}_{\mu,k}$. Without loss of generality, we can assume that the algorithm elects a leader in the `right' half of $J_Y$ for at least half of the graphs $J_Y \in \mathcal{J}_{\mu,k}$, and we restrict our attention to these graphs. We consider a node that is on the extreme `left' side of each such $J_Y$, and we use the previous result to conclude that this node must output a different port sequence for each such $J_Y$. Again, by the previous result, this extreme node cannot see far enough to detect any swapped ports, so it must rely entirely on the advice it is given, i.e., producing two different outputs requires two different pieces of advice. This shows that the oracle gives a different piece of advice for each graph, which leads to our lower bound on the length of the advice strings in the worst case.
\begin{theorem}
	Consider any algorithm $\cA$ that solves $\ppe$ in $\psi_{\ppe}(G)$ rounds for every graph $G$. For all integers $\Delta \geq 16, k\ge 6$, there exists a
	graph $G$ with maximum degree $O(\Delta)$ and with $\psi_{\ppe}(G)=k$ for which algorithm $\cA$ requires
	advice of size $\Omega(2^{\Delta^{k/6}})$.
\end{theorem}
\begin{proof}
	Let $\mu = \lceil \Delta/4 \rceil$ and note that $\mu \geq 4$. To obtain a contradiction, assume that there exists an algorithm $\cA$ that solves $\ppe$ in $k$ rounds for all graphs in the class $\mathcal{J}_{\mu,k}$ with the help of an oracle that provides advice of size $2^{(4\mu)^{k/6}}$. By construction, every graph $G$ in $\mathcal{J}_{\mu,k}$ has maximum degree $4\mu \in O(\Delta)$, and, by Lemma \ref{lem:PPECPPEindexk}, has $\psi_{\ppe}(G)=k$. Let $\mathcal{J}^{right}_{\mu,k}$ be the subset of $\mathcal{J}_{\mu,k}$ consisting of graphs such that algorithm $\cA$ elects as leader a node contained in some $\widehat{H_i}$ with $i \in \{2^{z-1},\ldots,2^z-1\}$. Let $\mathcal{J}^{left}_{\mu,k} = \mathcal{J}_{\mu,k} \setminus \mathcal{J}^{right}_{\mu,k}$, i.e., the subset of $\mathcal{J}_{\mu,k}$ consisting of graphs such that algorithm $\cA$ elects as leader a node contained in some $\widehat{H_i}$ with $i \in \{0,\ldots,2^{z-1}-1\}$. The proof proceeds by considering the subset $\mathcal{J}^{left}_{\mu,k}$ or $\mathcal{J}^{right}_{\mu,k}$ that contains at least half of the graphs from $\mathcal{J}_{\mu,k}$. Without loss of generality, we assume that $|\mathcal{J}^{right}_{\mu,k}| \geq |\mathcal{J}_{\mu,k}|/2$
	
	There are at most $2^{1+2^{(4\mu)^{k/6}}}$ binary advice strings whose length is at most $2^{(4\mu)^{k/6}}$. The number of graphs in $\mathcal{J}^{right}_{\mu,k}$ is at least $|\mathcal{J}_{\mu,k}|/2$, which, by Fact \ref{fact:JClassSize}, is at least $2^{2^{z-1}-1}$. We now set out to show that the number of graphs in $\mathcal{J}^{right}_{\mu,k}$ is strictly larger than the number of possible binary advice strings. As $\mu \geq 4$ and $k \geq 6$, it follows that $(4\mu)^{k/6} \leq \mu^{k/3}$ and $\mu^{k/3} < \mu^{\lfloor k/2 \rfloor} - 2$. By Fact \ref{fact:JClassSize}, $\mu^{\lfloor k/2 \rfloor} - 2 \leq z - 2$. Thus, we have shown that $(4\mu)^{k/6} < z-2$, from which it follows that $2^{(4\mu)^{k/6}} < 2^{z-2}$.  As $\mu \geq 4$ and $k \geq 6$, from Fact \ref{fact:JClassSize} we get that $z \geq 64$, so $2^{z-2} < 2^{z-1}-2$. Therefore, $2^{(4\mu)^{k/6}} < 2^{z-1} -2$, from which it follows that $1+2^{(4\mu)^{k/6}} < 2^{z-1} -1$, so $2^{1+2^{(4\mu)^{k/6}}} < 2^{2^{z-1} -1}$, as desired.

	By the Pigeonhole Principle, the oracle provides the same advice for at least two different graphs $J_\alpha$ and $J_\beta$ from $\mathcal{J}^{right}_{\mu,k}$.
	For each $Y \in \{\alpha,\beta\}$, let $v_Y$ be the node $w_{1,1}$ in component $H_L$ of gadget $\widehat{H_0}$ of $J_Y$. 
	By statement \ref{part1:valphavbeta} of Lemma \ref{lem:valphavbeta}, we know that $\cB^{k}(v_\alpha)$ in $J_\alpha$ is equal to $\cB^{k}(v_\beta)$ in $J_\beta$, so together with the fact that the two nodes get the same advice, it follows that $v_\alpha$ and $v_\beta$ output the same port sequence $\sigma$ when $\cA$ terminates. By the definition of $\mathcal{J}^{right}_{\mu,k}$ and the assumed correctness of $\cA$, port sequence $\sigma$ corresponds to a simple path in $J_\alpha$ starting at $v_\alpha$ that terminates at the leader node in some $\widehat{H_i}$ with $i \in \{2^{z-1},\ldots,2^z-1\}$, and, also corresponds to a simple path in $J_\beta$ starting at $v_\beta$ that terminates at the leader node in some $\widehat{H_{i'}}$ with $i' \in \{2^{z-1},\ldots,2^z-1\}$. This contradicts statement \ref{part2:valphavbeta} of Lemma \ref{lem:valphavbeta}.
\end{proof}

The exact same proof applies to the amount of advice needed by any algorithm that solves the $\cppe$ task in time $k$.
\begin{theorem}
	Consider any algorithm $\cA$ that solves $\cppe$ in $\psi_{\cppe}(G)$ rounds for every graph $G$. For all integers $\Delta \geq 16, k\ge 6$, there exists a
	graph $G$ with maximum degree $O(\Delta)$ and with $\psi_{\cppe}(G)=k$ for which algorithm $\cA$ requires
	advice of size $\Omega(2^{\Delta^{k/6}})$.
\end{theorem}

\section{Conclusion}

We showed that the size of advice required to accomplish the weakest version of leader election in minimum time is exponentially smaller than that needed for any of the strong versions. A natural open question is whether, for all strong versions,
the sizes of advice required to accomplish these tasks in minimum time differ only polynomially, or if there are also exponential gaps between some of them. Another open problem is whether the same relations between the four studied ``shades of leader election'' hold if we allocate to these tasks some other amount of time that is larger than the strict minimum.

\bibliographystyle{plain}


\end{document}